\documentclass[format=acmsmall,review=false,screen=true,nonacm=true]{acmart}

\def\BibTeX{{\rm B\kern-.05em{\sc i\kern-.025em b}\kern-.08emT\kern-.1667em\lower.7ex\hbox{E}\kern-.125emX}}

\usepackage[linesnumbered,algoruled,noend]{algorithm2e}
\usepackage{amsfonts}
\usepackage{amsthm}
\usepackage{amsmath}
\usepackage{graphicx}
\usepackage{tikz}
\usetikzlibrary{arrows}
\usetikzlibrary{shapes,backgrounds,shapes.misc}
\usepackage{subcaption}
\usepackage{siunitx}
\usepackage[capitalize]{cleveref}
\usepackage[appendix=inline]{apxproof} 
\newtoggle{TR}
\toggletrue{TR} 

\usepackage{mathtools}
\mathtoolsset{showonlyrefs} 

\newtheoremrep{theorem}{Theorem}[section]
\newtheoremrep{lemma}[theorem]{Lemma}
\newtheoremrep{corollary}[theorem]{Corollary}
\newtheoremrep{assumption}[]{Assumption}

\newcommand\MASty[1]{\mbox{\footnotesize \sffamily{\textup{#1}}}}
\newcommand{\BStable}{\lozenge \MASty{STABLE}_n}

\newcommand{\namedref}{\cref}
\newcommand{\sectionref}[1]{\namedref{#1}}

\newcommand{\figureref}[1]{\namedref{#1}}
\newcommand{\lemmaref}[1]{\namedref{#1}}

\newcommand{\algref}[1]{\namedref{#1}}

\newcommand{\N}{\mathbb{N}}

\newcommand{\R}{\mathbb{R}}

\newcommand{\T}{\mathcal{T}}
\renewcommand{\L}{\mathcal{L}}
\renewcommand{\P}{\mathcal{P}}
\providecommand\G{\mathcal{G}}
\renewcommand\G{\mathcal{G}}
\newcommand\PTG{\mathcal{PT}}
\newcommand\PT{\mathcal{PT}}

\newcommand\Gomega{\mathcal{G}^{\omega}}
\newcommand\PTomega{\mathcal{PT}^{\omega}}

\newcommand\MA{\text{\sffamily{MA}}}

\newcommand\Ker{\text{Ker}}

\newcommand\A{\mathcal{A}}
\providecommand\C{\mathcal{C}}
\renewcommand\C{\mathcal{C}}

\newcommand\Comega{\mathcal{C}^{\omega}}

\newcommand\otau{\hat{\tau}}

\newcommand\V{\mathcal{V}}

\SetKwData{var}{var}
\SetKwData{xDat}{x}

\SetKwData{EDat}{E}
\SetKwData{VDat}{V}
\SetKwData{SDat}{S}
\SetKwData{CDat}{C}
\SetKwData{KDat}{K}
\SetKwData{RDat}{R}

\newcommand{\cl}[1]{\overline{#1}}

\newcommand{\ob}{Ob}
\newcommand{\ho}{HO}

\newcommand{\dunif}{d_{\mathrm{u}}}
\newcommand{\dnonunif}{d_{\mathrm{nu}}}
\newcommand{\IN}{\mathbb{N}}
\newcommand{\IR}{\mathbb{R}}

\usepackage{todonotes}
\newcommand\ACT{\text{\sffamily{ACT}}}
\newcommand\act{\text{\sffamily{act}}}
\newcommand\Sched{\text{\sffamily{Sched}}}
\newcommand\recv{\text{\sffamily{delv}}}
\newcommand\done{\text{\sffamily{done}}}
\newcommand\fail{\text{\sffamily{fail}}}

\newcommand{\rST}[1]{r_{\text{stab},#1}}
\newcommand\true{\text{\sffamily{true}}}
\newcommand\false{\text{\sffamily{false}}}
\newcommand\heartbeat{\text{\sffamily{heartbeat}}}
\newcommand\accusation{\text{\sffamily{accusation}}}
\newcommand\accuse{\text{\sffamily{accuse}}}
\newcommand\nottimelyrec{\text{\sffamily{nottimelyrec}}}
\newcommand\nottimely{\text{\sffamily{nottimely}}}
\newcommand\accusationcounter{\text{\sffamily{accusationcounter}}}
\newcommand\hearedof{\text{\sffamily{heardof}}}
\newcommand\minhearedof{\text{\sffamily{minheardof}}}
\newcommand\oldenough{\text{\sffamily{oldenough}}}
\newcommand\mature{\text{\sffamily{mature}}}
\newcommand\simmaj{\sim_{\geq n-f}}

\newcommand{\myFunc}{\Delta}

\setcopyright{acmlicensed}

\acmDOI{0000001.0000001}

\begin{document}
\title{Topological Characterization of Consensus in Distributed Systems}
\subtitle{Dedicated to the 2018 Dijkstra Prize winners Bowen Alpern and Fred B. Schneider}
\author{Thomas Nowak}
\authornote{Thomas Nowak has been supported by the Universit\'e Paris-Saclay project DEPEC MODE and the ANR project DREAMY (ANR-21-CE48-0003).}
\affiliation{%
  \institution{Universit\'e Paris-Saclay, CNRS, ENS Paris-Saclay}
  \city{Gif-sur-Yvette}
  \country{France}
}
\affiliation{%
  \institution{Institut Universitaire de France}
  \city{Paris}
  \country{France}
}
\email{thomas@thomasnowak.net}

\author{Ulrich Schmid}
\authornote{Ulrich Schmid has been supported by the Austrian Science Fund (FWF) under project ADynNet (P28182), RiSE/SHiNE (S11405), DMAC (P32431), and ByzDEL (P33600).}
\affiliation{%
  \institution{TU Wien}
  \city{Vienna}
  \country{Austria}
}
\email{s@ecs.tuwien.ac.at}

\author{Kyrill Winkler}
\authornote{Kyrill Winkler has been supported by the Austrian Science Fund (FWF) under project ADynNet (P28182) and RiSE/SHiNE (S11405).
When this work was initiated, Kyrill Winkler was with TU Wien.}
\affiliation{%
  \institution{ITK Engineering}
  \country{Austria}
}
\email{kyrill.winkler@itk-engineering.com}

\begin{abstract}
  We provide a complete characterization of both uniform and non-uniform
deterministic consensus solvability in distributed systems with benign process
and communication faults using point-set topology. More specifically, we
non-trivially extend the approach introduced by Alpern and Schneider in 1985,
by introducing novel fault-aware topologies on the space of infinite
executions: the process-view topology, induced by a distance function that
relies on the local view of a given process in an execution, and the minimum
topology, which is induced by a distance function that focuses on the local
view of the process that is the last to distinguish two executions. Consensus
is solvable in a given model if and only if the sets of admissible executions
leading to different decision values is disconnected in these topologies.
By applying our
approach to a wide range of different applications, we provide a topological
explanation of a number of existing algorithms and impossibility results and
develop several new ones, including a general equivalence of the strong and weak validity conditions.

\end{abstract}



\maketitle

\section{Introduction}
\label{sec:intro}

We provide a complete characterization
of the solvability of deterministic
non-uniform and uniform consensus in distributed systems with benign process and/or
communication failures, 
using point-set topology as introduced in the Dijkstra Prize-winning paper by
Alpern and Schneider \cite{AS84}.
Our results hence precisely delimit the consensus 
solvability/impossibility border in very different distributed systems such as
dynamic networks \cite{KO11:SIGACT} controlled by a message adversary 
\cite{AG13}, synchronous distributed systems with processes that
may crash or commit send and/or receive omission failures \cite{PT86},
or purely asynchronous systems with crash failures \cite{FLP85}, for example.
Whereas we will primarily focus on message-passing architectures in our examples, 
our topological approach also covers shared-memory systems \cite{MRR03:JACM}.

Deterministic consensus, where every process starts with some initial input value 
picked from a finite set $\V$ and has
to irrevocably compute a common output value, is arguably the most well-studied
problem in
distributed computing. Both impossibility results and consensus
algorithm are known for virtually all distributed computing that
have been proposed so far. However, they have been obtained primarily
on a case-by-case basis, using classic combinatorial analysis 
techniques \cite{FR03}. Whereas
there are also some generic characterizations~\cite{MR02,LM95:DC},
i.e., ones that can be applied to different models of computation,
we are not aware of any approach
that allows to precisely characterize the consensus solvability/impossibility
border for arbitrary distributed systems with benign process- and 
communication-failures.

\begin{figure*}
\input{subdivision.tex}
\hspace{2cm}
\begin{tikzpicture}[>=latex']
\node[very thick, draw, circle] (C0) at (0, 0) {$C_0$};

\node[draw, circle] (C11) at (-2, 2) {$C_1$};
\node[draw, circle] (C12) at ( 0, 2) {$C_1'$};
\node[very thick, draw, circle] (C13) at ( 2, 2) {$C_1''$};
\draw[->] (C0) -- (C11);
\draw[->] (C0) -- (C12);
\draw[very thick, ->] (C0) -- (C13);

\node (C21) at (-2.5, 4) {};
\node (C22) at (-2.0, 4) {};
\node (C23) at (-1.5, 4) {};
\draw[->] (C11) -- (C21);
\draw[->] (C11) -- (C22);
\draw[->] (C11) -- (C23);

\node (C24) at (-0.5, 4) {};
\node (C25) at (-0.0, 4) {};
\node (C26) at ( 0.5, 4) {};
\draw[->] (C12) -- (C24);
\draw[->] (C12) -- (C25);
\draw[->] (C12) -- (C26);

\node (C27) at (1.5, 4) {};
\node (C28) at (2.0, 4) {};
\node (C29) at (2.5, 4) {};
\draw[very thick, ->] (C13) -- (C27);
\draw[->] (C13) -- (C28);
\draw[->] (C13) -- (C29);
\end{tikzpicture}
\caption{Comparison of the combinatorial topology approach and the point-set topology approach:
The combinatorial topology approach (left) studies sequences of increasingly
refined spaces in which the objects of interest are simplices (corresponding to
configurations).
The point-set topology approach (right) studies a single space in which the objects of interest are executions (i.e., infinite sequences of configurations).}
\label{fig:combtop}
\end{figure*}

In this paper, we provide such a characterization based on
point-set topology, as introduced
by Alpern and Schneider~\cite{AS84}. Regarding topological methods in distributed computing, one 
has to distinguish point-set topology, which considers the space of infinite
executions of a distributed algorithm, from combinatorial topology, which 
studies the topology of reachable states of prefixes of admissible executions
using simplicial complexes.
\figureref{fig:combtop} illustrates the objects studied in
combinatorial topology vs.\ point-set topology. As of today, 
combinatorial topology has been developed into a quite widely 
applicable tool for the analysis of distributed systems~\cite{HKR13}. 
A celebrated result in this area is the \emph{Asynchronous Computability
Theorem}~\cite{HS99:ACT}, \cite{GRS22:ITCS}, for example, which characterizes solvable tasks
in wait-free asynchronous shared memory systems with crashes.

By contrast, point-set topology has only rarely been used in
distributed computing. The primary objects
are the \emph{infinite} executions of a distributed algorithm~\cite{AS84}. 
By defining a suitable metric 
between two infinite executions $\gamma$ and $\delta$, each considered as the corresponding infinite sequence of global states 
of the algorithm in the respective execution, they can be viewed as 
elements of a topological space. For example, according to the common-prefix
metric $d_{\max}(\gamma,\delta)$, the executions $\gamma$ and $\delta$
are close if the common prefix where
\emph{no} process can distinguish them is long.
A celebrated general result~\cite{AS84} is that closed
and dense sets in the resulting space precisely characterize 
safety and liveness properties, respectively. 

Prior to our paper, however, point-set
topology has only occasionally been used for establishing
impossibility results.
We are only aware of some early work by one of the 
authors of this paper on a generic topological impossibility proof 
for consensus in compact models~\cite{Now10:master}, and a topological 
study of the strongly dependent decision problem~\cite{BR19:ICDCN}.
Lubitch and Moran~\cite{LM95:DC}  introduced a construction for schedulers, which
leads to limit-closed submodels\footnote{Informally, a model is limit-closed
if the limit of a sequence of growing prefixes of admissible executions is admissible.
Note that the wait-free asynchronous model is limit-closed.} 
of classic non-closed distributed computing
models (like asynchronous systems consisting of $|\Pi|=n$ processes, up to which
$t<n-1$ may crash). In a similar spirit, Kuznetsov, Rieutord, and He~\cite{KRH18:PODC}
showed, in the setting of combinatorial topology, how to
reason 
about non-closed models by considering equivalent 
affine tasks that are closed. 
Gafni, Kuznetsov, and Manolescu~\cite{GKM14:PODC} tried to extend the ACT to
also cover some non-compact shared memory models that way.
A similar purpose is served by defining layerings, as introduced by 
Moses and Rajsbaum~\cite{MR02}. Whereas such constructions of closed submodels 
greatly simplify impossibility proofs, they do not lead to a precise characterization 
of consensus solvability in non-closed models: We are not aware
of any proof that there is an equivalent closed submodel for every non-closed
model. 

\medskip
\noindent
\textbf{Contributions.} 
Building on our PODC'19 paper~\cite{NSW19:PODC} devoted
to consensus in dynamic networks under message adversaries~\cite{AG13}, the present paper
provides a complete topological
characterization of both the non-uniform and uniform deterministic consensus
solvability/impossibility border for general distributed systems with 
benign
process and/or communication faults. To achieve
this, we had to add several new topological ideas to the setting 
of Alpern and Schneider~\cite{AS84}, as detailed below, which not only allowed
us to
deal with both closed and non-closed models, but also provided us with
a topological explanation of bivalence~\cite{FLP85} and bipotence~\cite{MR02} 
impossibility proofs. In more detail:

(i) We introduce a simple generic system model for full-information protocols
that covers all distributed system models with benign faults we are aware of.
We define new topologies on the execution space of general distributed
algorithms in this model, which allow us to reason about sequences of local 
views of (correct) processes, rather than about global configuration sequences.
The \emph{$p$-view topology} is defined by a distance function $d_p(\gamma,\delta)$ based on 
the common prefix of $p$'s local views in the executions $\gamma$ and $\delta$.
The uniform and non-uniform minimum topology are induced by the last (correct) 
process to notice a difference between two executions. In the appendix, we introduce process-time
graphs~\cite{BM14:JACM} as a succinct alternative to configuration sequences in executions, and show that they are equivalent w.r.t.\ our topological
reasoning. This is accomplished by instantiating our generic system model as an 
``operational'' system model, based on the widely applicable modeling framework 
introduced by Moses and Rajsbaum~\cite{MR02}.

(ii) We show that consensus can be modeled as a continuous decision 
function $\Delta$ in our topologies, 
which maps an admissible execution to its unique decision value. 
This allows us to prove that consensus is solvable if and only 
if all the decision sets, i.e., the pre-images $\Sigma_v=\Delta^{-1}(v)$
for every decision value $v \in \V$, 
are disconnected from each other. We also provide a universal uniform and non-uniform
consensus algorithm, which rely on this separation. 

(iii) We provide an alternative characterization of uniform and non-uniform 
consensus solvability based on the broadcastability of the decision sets and
their connected components. It applies for the usual situation where every 
vector of values from $\V$ is an allowed assignment of the input values of the 
processes (which is not the case for condition-based consensus \cite{MRR03:JACM}, however). Interestingly, our respective results imply that solving consensus with weak validity
and consensus with strong validity is equivalent in any model with benign faults. 
Moreover, we provide a characterization of consensus
solvability based on the limits of two infinite sequences of 
admissible executions, taken from different decision sets. Consensus
is impossible if there is just one pair of such limits
with distance~0, which actually coincide with the forever bivalent/bipotent
executions constructed in previous proofs~\cite{FLP85,MR02}.

(iv) We apply our topological approach to different distributed computing
models. This way, we provide
a topological explanation of well-known classic results like bivalence proofs
and consensus solvability/impossibility in synchronous systems with general 
omission faults.
Despite the fact that consensus has been thoroughly studied in virtually any
conceivable distributed computing model, we also provide some new results:
We provide a new necessary and sufficient condition for solving condition-based consensus
with strong validity in asynchronous shared-memory systems~\cite{MRR03:JACM},
comprehensively characterize consensus solvability in dynamic networks
with both compact and non-compact message adversaries \cite{AG13}, and
give a novel consensus algorithm that does not rely on an implementation
of the $\Omega$ failure detector for systems with an eventually timely
$f$-source \cite{ADGFT04,HMSZ08:TDSC}.

\medskip
\noindent
\textbf{Paper organization.} In \cref{sec:general:model},
we define the elements
of the space that will be endowed with our new topologies in
\cref{sec:structure:executions}. 
Section~\ref{sec:consensus} introduces the consensus problem in topological
terms and
provides our abstract characterization result for uniform consensus
(Theorem~\ref{thm:char:unif}) and
non-uniform consensus (Theorem~\ref{thm:char:nonunif}), which also
provide universal algorithms.  Alternative characterizations
based on limit exclusion and broadcastability are provided in \cref{sec:fairunfair} and \cref{sec:broadcastability}, respectively. Our topological characterizations are
complemented by Section~\ref{sec:applications}, which is devoted to applications.
Some conclusions in Section~\ref{sec:conclusions} round off our paper.
In \cref{sec:model}, we introduce process-time graphs and an
operationalization of our generic system model for some 
classic distributed computing models.

\section{Related Work}
\label{sec:relwork}

\medskip
\noindent

Besides the few point-set topology papers~\cite{AS84,Now10:master,BR19:ICDCN}
and the closed model constructions~\cite{LM95:DC,MR02,KRH18:PODC,GKM14:PODC} 
already mentioned in \cref{sec:intro}, there is an abundant literature on
consensus algorithms and impossibility proofs. 

Regarding combinatorial topology, it is worth mentioning that 
our study of the indistinguishability relation of prefixes of
executions is closely connected to connectivity properties of the $r$-round
protocol complex. However, in non-limit-closed models, we need to go 
beyond a uniformly bounded prefix length. This is in sharp contrast
to the models usually considered in combinatorial topology~\cite{CFPR19:SSS,ACR20:OPODIS}, which are limit-closed (typically, 
wait-free asynchronous).

A celebrated paper on the impossibility of consensus in asynchronous systems with crash failures
is by Fischer, Lynch, and Paterson~\cite{FLP85}, who also introduced the bivalence proof technique. This impossibility can be avoided by means of 
unreliable failure detectors~\cite{CT96} or condition-based approaches
restricting the allowed inputs \cite{MRR03:JACM}.
Consensus in synchronous systems with Byzantine-faulty processes has
been introduced by Lamport, Shostak, and Pease~\cite{LSP82}. The seminal works by Dolev, Dwork, and Stockmeyer~\cite{DDS87} and Dwork, Lynch, and Stochmeyer~\cite{DLS88} on partially synchronous systems introduced important abstractions like eventual stabilization and eventually bounded message delays, and provided a characterization of consensus solvability under various combinations of synchrony and failure models.
Consensus in systems with weak timely links and crash failures was considered~\cite{ADGFT04,HMSZ08:TDSC}. Algorithms for consensus in systems with general omission process faults were provided by Perry and Toueg~\cite{PT86}.

Perhaps one of the earliest characterizations of consensus solvability in
synchronous distributed systems prone to communication errors is the seminal work by Santoro and Widmayer~\cite{SW89}, where
it was shown that consensus is impossible if up to $n-1$ messages may be lost in
each round.
This classic result was refined by several authors~\cite{SWK09, CBS09} and, more recently,  by
Coulouma, Godard, and Peters~\cite{CGP15}, where a property
of an equivalence relation on the sets of communication graphs was found that
captures exactly the source of consensus impossibility.
Following Afek and Gafni~\cite{AG13}, such distributed systems are nowadays known as dynamic networks, 
where the per-round directed communication graphs are controlled by a message adversary. Whereas the paper by Coulouma, Godard, and Peters~\cite{CGP15} and follow-up work \cite{WPRSS23:ITCS}
studied oblivious message adversaries, where the communication graphs 
are picked arbitrarily from a set of candidate graphs, more recent papers~\cite{BRSSW18:TCS,WSS19:DC} 
studied eventually stabilizing message adversaries, which guarantee that some rounds with ``good'' communication graphs will eventually be generated. Note that oblivious message adversaries
are limit-closed, which is not the case for message adversaries like the eventually stabilizing ones.
Raynal and Stainer explored the relation between message adversaries and failure detectors~\cite{RS13:PODC}.

The first characterization of consensus solvability under general message adversaries was provided 
by Fevat and Godard~\cite{FG11}, albeit only for systems that consist of two processes.
A bivalence argument was used there to show that certain communication patterns,
namely, a ``fair'' or a special pair of ``unfair'' communication patterns (see \cref{def:fairunfair} for more information), must be excluded by the message adversary for consensus to become solvable. 

\section{Generic System Model}\label{sec:general:model}

We consider distributed message passing or shared memory
systems made up of a set of $n$ deterministic processes $\Pi$
with unique identifiers, taken 
from $[n]=\{1,\dots,n\}$ for simplicity. 
We denote individual processes by letters $p$, $q$, etc.

For our characterization of consensus solvability, we restrict
our attention to \emph{full-information executions}, in which
processes continuously relay all the information they gathered to
all other processes, and eventually apply some local decision function.
The exchanged information includes the process's initial value, but also, more
importantly, a record of all events (message receptions, shared memory
readings, object invocations, \dots) witnessed by
the process.
As such, our general system model is hence applicable whenever no constraints are
placed on the size of the local memory and the size of values to be
communicated (e.g., message/shared-register size).
In particular, it is applicable to classical synchronous and asynchronous
message-passing and shared-memory models, with benign
process and communication faults. 
In \cref{sec:model}, we will also provide a topologically equivalent
``operationalization'' of our
generic system model built on top of process-time graphs~\cite{BM14:JACM},
based the modeling framework introduced by Moses and Rajsbaum~\cite{MR02}.

Formally, a (full-information)
execution is a sequence of (full-information) configurations.
For every process $p\in\Pi$, there is an equivalence relation~$\sim_p$ on the
set~$\C$ of configurations---the
{$p$-indistinguishability relation}---indicating whether process~$p$ can
locally distinguish two configurations, i.e., if it
has the same \emph{view} $V_p(C)=V_p(D)$ in~$C$ and~$D$.
In this case we write $C\sim_p D$.
Note that two configurations
that are indistinguishable for all processes need not be equal.
In fact, configurations usually include some state of the communication media
that is not accessible to any process.

In addition to the indistinguishability relations, we assume the existence of a
function $\ob:\C\to2^\Pi$ that specifies the set of \emph{obedient} processes in a
given configuration. Obedient processes must follow the algorithm and satisfy
the (consensus) specification;
usually, $\ob(C)$ is the set of non-faulty processes.
Again, this information is usually not accessible to the processes.
We make the restriction that disobedient processes cannot
recover and become obedient again, 
i.e., that $\ob(C) \supseteq \ob(C')$ if~$C'$ is reachable from~$C$.
We extend the obedience function to the set $\Sigma\subseteq\C^\omega$ of
\emph{admissible executions} in a given model by
setting 
$\ob:\Sigma \to 2^\Pi$,
$\ob(\gamma) = \bigcap_{t\geq 0} \ob(C^t)$
where $\gamma = (C^t)_{t\geq 0}$.
Here, $t \in \IN_0 = \IN \cup \{0\}$ denotes a notion
of \emph{global} time that is not accessible to the processes.
Consequently, a process is obedient in an 
execution if it is obedient in all of its
configurations.
We further make the restriction that there is at least one obedient process
in every execution, i.e., that $\ob(\gamma)\neq \emptyset$ for all
$\gamma\in\Sigma$. Moreover, we assume that $\ob(C)=\Pi$ for every
initial configuration, in order to make input value assignments (see below)
well-defined for all processes.

We also assume that every process has the possibility to weakly count the steps it
has taken.
Formally, we assume the existence of weak clock functions
$\chi_p:\C\to\IN_0$
such that for every execution 
$\delta = (D^t)_{t\geq0}\in\Sigma$
and every configuration $C\in\C$,
the relation $C\sim_p D^t$ implies $t\geq \chi_p(C)$.
Additionally, we assume that $\chi_p(D^t)\to\infty$ as $t\to\infty$ for every
execution $\delta\in\Sigma$ and every obedient process $p\in\ob(\delta)$.
$\chi_p$ hence ensures
that a configuration $D^t$ where $p$ has some specific view $V_p(D^t)=V_p(C)$ 
cannot occur before time $t=\chi_p(C)$ in any execution $\delta$.
Our weak clock functions hence allow to model lockstep synchronous
rounds by choosing
$\chi(D^t)=t$ for any execution $\delta = (D^t)_{t\geq0}\in\Sigma$,
but are also suitable for modeling non-lockstep, even asynchronous, executions
(see \cref{sec:cc}).

For the discussion of decision problems, we need to introduce the notion of
input values, which will also be called initial values in the sequel.
Since we limit ourselves to the consensus problem, we need not distinguish
between the sets of input values and output values.
We thus just assume the existence of a finite set~$\V$ of potential input values, and require
that the potential output values are also in $\V$.
Furthermore, the initial configuration $I=I(\gamma)$ of any execution $\gamma$ 
is assumed to contain an input value $I_p \in \V$ for every process $p\in \Pi$.
This information is locally accessible to the processes, i.e., each process can access its own initial value (and those it has heard from).
We assume that there is a unique initial configuration for every input-value assignment of the processes.

A \emph{decision algorithm} is a collection of functions 
$\Delta_p:\C \to \V\cup\{\perp\}$ 
such that
$\Delta_p(C) = \Delta_p(D)$ if $C \sim_p D$
and
$\Delta_p(C') = \Delta_p(C)$ 
if~$C'$ is reachable
from~$C$
and 
$\Delta_p(C)\neq \perp$, where $\perp\not\in \V$ 
represents the fact that $p$ has not decided yet.
That is, decisions depend on local information only and are irrevocable.
Every process~$p$ thus has at most one decision value in an execution.
We can extend the decision function to executions by setting
$\Delta_p:\Sigma\to\V\cup\{\perp\}$, $\Delta_p(\gamma) = \lim_{t\to\infty} \Delta_p(C^t)$ where $\gamma = (C^t)_{t\geq0}$.
We say that~$p$ has decided value $v\neq\perp$ in configuration~$C$ or
execution~$\gamma$ if $\Delta_p(C)=v$ or $\Delta_p(\gamma)=v$, respectively.

We will consider both non-uniform and uniform consensus 
with either weak or strong validity as our decision tasks, 
which are defined as follows:

\begin{definition}[Non-uniform and uniform consensus]\label{def:consensus}
A \emph{non-uniform consensus} algorithm $\A$ is a decision
algorithm that ensures
the following properties in all of its admissible executions: 
\begin{enumerate}
\item[(T)] Eventually, every obedient process must irrevocably decide.
(Termination)
\item[(A)] If two obedient processes have decided, then their decision values
are equal. (Agreement)
\item[(V)] If the initial values of processes are all equal to~$v$, then~$v$ is the only possible decision value. (Validity)
\end{enumerate}
In a \emph{strong consensus} algorithm $\A$, weak validity (V) is replaced by
\begin{enumerate}
\item[(SV)] The decision value must be the input value of some process. (Strong Validity)
\end{enumerate}
A \emph{uniform consensus} algorithm $\A$ must ensure (T), (V) or (SV), and
\begin{enumerate}
\item[(UA)] If two processes have decided, then their decision values
are equal. (Uniform Agreement)
\end{enumerate}
\end{definition}
Note that we will primarily focus on consensus with weak validity, which
is the usual meaning of the term consensus unless otherwise noted. 

By Termination, Agreement, and the fact that every execution has at least one
obedient process, for every consensus algorithm, we can define the consensus
decision function
$\Delta : \Sigma \to \V$ by setting $\Delta(\gamma) = \Delta_p(\gamma)$
where~$p$ is any
process that is obedient in execution~$\gamma$, i.e., $p\in \ob(\gamma)$.
Recall that the initial value of process $p$ in the execution $\gamma$ is denoted
$I_p(\gamma)$ or just $I_p$ if $\gamma$ is clear from the context.

To illustrate\footnote{We chose this simplistic illustrating example in order
not to obfuscate the essentials. See \cref{sec:applications}
for more realistic
examples.} the difference between uniform and non-uniform consensus,
as well as to motivate the two topologies serving to characterize their solvability,
consider the example of two synchronous non-communicating processes.
The set of processes is $\Pi = \{1,2\}$ and the set of possible values is
$\V = \{0,1\}$.
Processes proceed in lock-step synchronous rounds, but do not communicate.
Thus, the only information a process has access to is its own initial value
and the current time.
The set of executions~$\Sigma$ and the obedience function~$\ob$ are defined
such that one of the processes eventually becomes disobedient in every
execution, but not both processes.
In this model, it is trivial to solve non-uniform consensus by immediately
deciding on one's own initial value, but uniform consensus is impossible.

\section{Topological Structure of Full-Information Executions}
\label{sec:structure:executions}

In this section, we will
endow the various sets introduced in Section~\ref{sec:general:model}
with suitable topologies.
We first recall briefly the basic topological notions that are needed for our
exposition.
For a more thorough introduction, however, the reader is advised to refer to a
textbook~\cite{Munkres}.

A topology on a set $X$ is a family $\T$ of subsets of $X$
such that $\emptyset \in \T$, $X \in \T$, and $\T$ contains all arbitrary
unions as well as all finite intersections of its members.
We call $X$ endowed with $\T$, often written as $(X, \T)$, a topological
space and the members of $\T$ open sets. 
The complement of an open set is called closed and
sets that are both open and closed, such as $\emptyset$ and $X$ itself, are called clopen.
A topological space is disconnected,
if it contains a nontrivial clopen set, which means that it it can
be partitioned into two disjoint open sets. It is connected
if it is not disconnected.

A function from space $X$ to space $Y$ is continuous if the pre-image of every
open set in $Y$ is open in $X$.
Given a space $(X, \T)$, $Y \subseteq X$ is called a
subspace of $X$ if $Y$ is equipped with the subspace topology
$\{ Y \cap U \mid U \in \T \}$.
Given $A \subseteq X$, the closure of $A$ is the intersection of all closed
sets containing $A$.
For a space $X$, if $A \subseteq X$, we call $x$ a limit point of~$A$
if it belongs to the closure of $A \setminus \{ x \}$.
It can be shown that the closure of $A$ is the union of~$A$ with all
limit points of~$A$.
Space $X$ is called compact if every family of open sets that covers~$X$
contains a finite sub-family that covers~$X$.

If~$X$ is a nonempty set, then
we call any function $d:X\times X\to \R_+$ a \emph{distance function} on~$X$.
Define $\T_d \subseteq 2^X$ by setting
$U\in\T_d$ if and only if
for all $x\in U$ there exists some $\varepsilon > 0$ such that
$B_\varepsilon(x) = 
\{ y \in X \mid d(x,y) < \varepsilon \}
\subseteq
U$.

Many topological spaces are defined by metrics, i.e., symmetric, positive definite
distance functions for which the triangle inequality
$d(x,y) \leq d(x,z)+d(z,y)$ holds for any $x,y,z \in X$.
For a distance function to define a (potentially non-metrizable) topology
though, no additional assumptions are necessary:

\begin{lemma}\label{lem:pseudosemi}
If~$d$ is a distance function on~$X$, then $\T_d$ is a topology on~$X$.
\end{lemma}
\begin{proof}
Firstly, we show that~$\T_d$ is closed under unions.
So let $\mathcal{U}\subseteq \T_d$.
We will show that $\bigcup \mathcal{U}\in \T_d$.
Let $x\in \bigcup\mathcal{U}$.
Then, by definition of the set union, there exists some $U\in\mathcal{U}$
such that $x\in U$.
But since $U\in\T_d$, there exists some $\varepsilon>0$ such that
\begin{equation}
B_\varepsilon(x) \subseteq U \subseteq \bigcup \mathcal{U}
\enspace,
\end{equation}
which shows that $\bigcup\mathcal{U}\in\T_d$.

Secondly, we show that~$\T_d$ is closed under finite intersections.
Let $U_1,U_2,\dots,U_k\in\T_d$.
We will show that $\bigcap_{\ell=1}^k U_\ell \in \T_d$.
Let $x\in \bigcap_{\ell=1}^k U_\ell$.
Then, by definition of the set intersection, $x\in U_\ell$ for all
$1\leq \ell \leq k$.
Because all~$U_\ell$ are in~$\T_d$, there exist 
$\varepsilon_1, \varepsilon_2,\dots,\varepsilon_k > 0$
such that
$B_{\varepsilon_\ell}(x) \subseteq U_\ell$
for all $1\leq \ell\leq k$.
If we set $\varepsilon = \min \{\varepsilon_1, \varepsilon_2,\dots,\varepsilon_k\}$,
then $\varepsilon > 0$.
Since we have $B_\gamma(x) \subseteq B_\delta(x)$ whenever $\gamma\leq\delta$,
we also have 
\begin{equation}
B_\varepsilon(x)
\subseteq
B_{\varepsilon_\ell}(x)
\subseteq
U_\ell
\end{equation}
for all $1\leq \ell \leq k$.
But this shows that $B_\varepsilon(x)\subseteq \bigcap_{\ell=1}^k U_\ell$,
which means that $\bigcap_{\ell=1}^k U_\ell \in \T_d$.

Since it is easy to check that $\emptyset, X\in \T_d$ as well, $\T_d$ is indeed
a topology.
\end{proof}

We will henceforth refer to~$\T_d$ as the topology induced by~$d$.

An execution is a sequence of configurations, i.e., an element of the 
product space~$\C^\omega$. Since our primary object of study are executions,
we will endow this space with a topology as follows:
The product topology, which is a distinguished
topology on any product space $\Pi_{\iota\in I}
X_\iota$ of topological spaces, is defined as the coarsest topology such that
all projection maps $\pi_i:\Pi_{\iota\in I}X_\iota\to X_i$ (where $\pi_i$ extracts the
$i$-th element of the sequence) are continuous.
Recall that a topology $\T'$ is coarser than a topology~$\T$ for
the same space if every open set $U \in \T'$ is also open in $\T$.

It turns out that the product topology on the space~$\C^\omega$ is induced
by a distance function, whose form is known in a special case that covers
our needs:

\begin{lemma}\label{lem:pseudosemi:product}
Let~$d$ be a distance function on~$X$ that only takes the values~$0$ or~$1$.
Then the product topology $\T^\omega$ of $X^\omega$, where every copy of~$X$ is endowed
with the topology induced by~$d$, is induced by the distance function 
\begin{equation}
X^\omega \times X^\omega \to \R
\quad,\quad
(\gamma,\delta)
\mapsto
2^{-\inf\{t\geq0\mid d(C^t,D^t) > 0\}}
\label{eq:defminmetric}
\end{equation}
where $\gamma = (C^t)_{t\geq0}$ and $\delta = (D^t)_{t\geq0}$.
\end{lemma}
\begin{proof}
We first show that all projections $\pi^t : X^\omega \to X$ are continuous
when endowing~$X^\omega$ with the product topology~$\T^\omega$:
Let $U\subseteq X$ be open and $C \in U$, i.e., $d(C,D) = 0$ implies $D\in U$.
Let $\gamma = (C^t)_{t\geq0} \in (\pi^t)^{-1}[U]$ and set $\varepsilon = 2^{-t}$.
Then, 
\begin{equation}
\begin{split}
B_{\varepsilon}(\gamma)
& =
\big\{ \delta = (D^t)_{t\geq0} \in X^\omega \mid \forall 0\leq s\leq t\colon d(C^s, D^s) = 0 \big\}
\\ & \subseteq
\big\{ \delta = (D^t)_{t\geq0} \in X^\omega \mid d(C^t,D^t) = 0 \big\}
\\ & 
=
(\pi^t)^{-1}\big[ \{D\in X \mid d(C^t,D) =0\} \big]
\subseteq
(\pi^t)^{-1}[ U ],
\end{split}
\end{equation}
where the last inclusion follows from the openness of~$U$. Since $(\pi^t)^{-1}[
U ]$ is hence open in $\T^\omega$, the continuity of~$\pi^t$ follows.

Let now~$\T_0$ be an arbitrary topology on~$X^\omega$ for which all projections~$\pi^t$
are continuous. We will show that $\T^\omega \subseteq \T_0$, which reveals that $\T^\omega$
is the coarsest topology with continuous projections, i.e., the 
product topology of $X^\omega$ where every copy of $X$ is endowed by $\T_d$. This will establish our lemma.

So let $E \in \T^\omega$ and take any $\gamma=(C^t)_{t\geq0}\in E$.
There exists some $\varepsilon>0$ such that $B_{\varepsilon}(\gamma) \subseteq E$.
Choose $t\in\N_0$ such that $2^{-t} \leq \varepsilon$, and
set
\begin{equation}
\begin{split}
F & = 
\bigl(\prod_{s=0}^t B_1(C^s)\bigr) \times X^\omega
=
\bigcap_{s=0}^t (\pi^s)^{-1}\big[B_1(C^s) \big] 
\\ & \subseteq 
\big\{ \delta = (D^t)_{t\geq0} \in X^\omega \mid \forall 0\leq s\leq t\colon d(C^s,D^s) = 0 \big\}
= B_{\varepsilon}(\gamma)\enspace.
\end{split}
\end{equation}
Then, $F$ is open with respect to~$\T_0$ as a finite intersection of open sets:
After all, every $(\pi^s)^{-1}\big[B_1(C^s) \big]$ is open by the continuity
of the projection $\pi^s$. 
But since $F \subseteq B_{\varepsilon}(\gamma)\subseteq E$, this
shows that~$E$ contains a $\T_0$-open neighborhood
for each of its points, i.e., $E \in \T_0$.
\end{proof}



\subsection{Process-view distance function for executions}
\label{sec:p:distance}

In previous work on point-set topology in distributed computing~\cite{AS84,Now10:master}, the set of
configurations~$\C$ of some fixed algorithm $\A$ was endowed with the discrete
topology, where every subset $U \subseteq \C$ is open. The discrete topology
is induced by the discrete metric $d_{\max}(C,D)=1$ if $C\neq D$ and $0$
otherwise (for configurations $C,D\in\C$).
Moreover, $\Comega$ was endowed with the corresponding
product topology, which is induced by the \emph{common-prefix
metric}
\begin{equation}\label{eq:commonprefixmetric}
d_{\max}:\Sigma \times \Sigma \to \IR_+
\quad,\quad
d_{\max}(\gamma,\delta) 
=2^{-\inf\{t\geq0\mid C^t \neq D^t\}}
\enspace,
\end{equation}
where $\gamma = (C^t)_{t\geq0}$ and $\delta = (D^t)_{t\geq0}$,
according to \lemmaref{lem:pseudosemi:product}. Informally, $d_{\max}(\gamma,\delta)$  decreases
with the length of the common prefix where \emph{no} process can distinguish
$\gamma$ and $\delta$.

By contrast, we define the \emph{$p$-view distance function}~$d_p$
on the set~$\C$ of configurations for every process $p\in\Pi$ by
\begin{equation}
d_p(C,D)
=
\begin{cases}
0 & \text{if } C \sim_p D \text{ and } p\in\ob(C)\cap\ob(D)\text{, or } C=D\\
1 & \text{otherwise}
\enspace.
\end{cases}
\end{equation}

Extending this distance function from configurations to executions,
we define the \emph{$p$-view distance function} by
\begin{equation}
d_p:\Sigma \times \Sigma \to \IR_+
\quad,\quad
d_p(\gamma,\delta)
=
2^{-\inf\{t\geq0\mid d_p(C^t,D^t) > 0\}}\label{eq:Pviewpseudometric}
\end{equation}
where $\gamma = (C^t)_{t\geq0}$ and $\delta = (D^t)_{t\geq0}$.

\begin{figure}
\begin{tikzpicture}[>=latex']
\draw[fill=blue!70!white] (-1.0,0) rectangle ++(-0.4,-0.4);
\draw[fill=blue!70!white] (-1.4,0) rectangle ++(-0.4,-0.4);
\draw[fill=blue!70!white] (-1.8,0) rectangle ++(-0.4,-0.4);

\node at (-2.8, -0.2) {$\gamma^0$};

\node at (-1.2, -0.8) {$\gamma^0_3$};
\node at (-1.6, -0.8) {$\gamma^0_2$};
\node at (-2.0, -0.8) {$\gamma^0_1$};

\draw[fill=blue!70!white] (+1.0,0) rectangle ++(0.4,-0.4);
\draw[fill=blue!70!white] (+1.4,0) rectangle ++(0.4,-0.4);
\draw[fill=yellow!20!white] (+1.8,0) rectangle ++(0.4,-0.4);

\node at (+2.8, -0.2) {$\delta^0$};

\node at (+1.2, -0.8) {$\delta^0_1$};
\node at (+1.6, -0.8) {$\delta^0_2$};
\node at (+2.0, -0.8) {$\delta^0_3$};

\draw[fill=yellow!20!white] (-1.0,1.5) rectangle ++(-0.4,-0.4);
\draw[fill=blue!70!white] (-1.4,1.5) rectangle ++(-0.4,-0.4);
\draw[fill=blue!70!white] (-1.8,1.5) rectangle ++(-0.4,-0.4);

\node at (-2.8, 1.3) {$\gamma^1$};

\draw[fill=blue!70!white] (+1.0,1.5) rectangle ++(0.4,-0.4);
\draw[fill=yellow!20!white] (+1.4,1.5) rectangle ++(0.4,-0.4);
\draw[fill=yellow!20!white] (+1.8,1.5) rectangle ++(0.4,-0.4);

\node at (+2.8, 1.3) {$\delta^1$};

\draw[fill=yellow!20!white] (-1.0,3.0) rectangle ++(-0.4,-0.4);
\draw[fill=yellow!20!white] (-1.4,3.0) rectangle ++(-0.4,-0.4);
\draw[fill=yellow!20!white] (-1.8,3.0) rectangle ++(-0.4,-0.4);

\node at (-2.8, 2.8) {$\gamma^2$};

\draw[fill=yellow!20!white] (+1.0,3.0) rectangle ++(0.4,-0.4);
\draw[fill=yellow!20!white] (+1.4,3.0) rectangle ++(0.4,-0.4);
\draw[fill=yellow!20!white] (+1.8,3.0) rectangle ++(0.4,-0.4);

\node at (+2.8, 2.8) {$\delta^2$};

\draw[->] (-1.6, 0.0) -- (-1.6, 1.1);
\draw[->] (-1.6, 1.5) -- (-1.6, 2.6);

\draw[->] (+1.6, 0.0) -- (+1.6, 1.1);
\draw[->] (+1.6, 1.5) -- (+1.6, 2.6);
\end{tikzpicture}
\caption{Comparison of the $p$-view and common-prefix metric.
The first three configurations of each of the two executions~$\gamma$
and~$\delta$ with three processes
and two different possible local states (dark blue and
light yellow) are depicted.
We have $d_{\max}(\gamma,\delta) = d_3(\gamma,\delta) = 1$ and
$d_2(\gamma,\delta) = 1/2$.}
\label{fig:tops}
\end{figure}

Figure~\ref{fig:tops} illustrates the distance function~$d_p$, and
\cref{lem:Pseudometricproperties} reveals that it defines a pseudometric:

\begin{lemma}[Pseudometric $d_p$]\label{lem:Pseudometricproperties}
The $p$-view distance function $d_p$ is a pseudometric, i.e., it satisfies:
\begin{align*}
d_p(\gamma,\gamma) =  0 & \\
d_p(\gamma,\delta) =  d_p(\delta,\gamma) & \qquad\mbox{(symmetry)}\\
d_p(\beta,\delta) \leq d_p(\beta,\gamma) + d_p(\gamma,\delta) & \qquad\mbox{(triangle inequality)}
\end{align*}
\end{lemma}
\begin{proof}
We have $d_p(\gamma,\gamma)=0$ since $d_p(C^t,C^t)=0$ for all $t\geq 0$ where $\gamma = (C^t)_{t\geq0}$.
Symmetry follows immediately from the definition.
As for the triangle inequality,
write $\beta = (B^t)_{t\geq0}$, $\gamma = (C^t)_{t\geq0}$, and $\delta = (D^t)_{t\geq0}$.
We have:
\begin{equation}
\begin{split}
\max\{ d_p(\beta,\gamma) , d_p(\gamma,\delta) \}
& =
2^{-\inf\{t\geq0 \mid d_p(B^t,C^t) > 0\ \vee\ d_p(C^t,D^t)>0 \}}
\end{split}
\end{equation}
Since $d_p(B^t,D^t) > 0 \implies d_p(B^t,C^t)>0\ \vee\ d_p(C^t,D^t)>0$, it follows that
\begin{equation}
\inf\{t\geq0 \mid d_p(B^t,D^t) > 0 \}
\geq
\inf\{t\geq0 \mid d_p(B^t,C^t) > 0\ \vee\ d_p(C^t,D^t)>0 \}
\end{equation}
and thus
\begin{equation}
d_p(\beta,\delta) \leq \max\{ d_p(\beta,\gamma) , d_p(\gamma,\delta) \}
\enspace,
\end{equation}
which concludes the proof.
\end{proof}

\subsection{Uniform topology for executions}
\label{sec:uniftopology}

The \emph{uniform minimum topology} (abbreviated \emph{uniform topology}) on the set~$\Sigma$ of executions is induced by the distance function
\begin{equation}
\dunif(\gamma, \delta)
=
\min_{p\in\Pi} d_p(\gamma,\delta)
\enspace.\label{eq:dunif}
\end{equation}

Note that $\dunif$ does not necessarily satisfy the triangle inequality (nor definiteness): There may
be executions with $d_{p}(\beta,\gamma)=0$ and $d_{q}(\gamma,\delta)=0$
but $d_{r}(\beta,\delta)>0$ for all $r\in \Pi$. Hence, the topology
on $\Comega$ induced by $\dunif$ lacks many of the convenient (separation)
properties of metric spaces, but will turn out to be sufficient for
the characterization of the solvability
of uniform consensus (see Theorem~\ref{thm:char:unif}).

The next lemma shows that the decision function of an algorithm that solves
uniform consensus is always continuous with respect to
the uniform topology.

\begin{lemma}\label{lem:cont:unif:consensus}
Let $\Delta:\Sigma\to\V$ be the consensus decision function of a
uniform consensus algorithm.
Then, $\Delta$ is continuous with respect to 
the uniform distance function~$\dunif$.
\end{lemma} 
\begin{proof}
Let $v\in\V$ and let $\Sigma_v = \Delta^{-1}[\{v\}]$ be its inverse image under
the decision function~$\Delta$.
We will show that for all executions $\gamma\in\Sigma_v$ there exists a
time~$T$ such that $B_{2^{-T}}(\gamma)\subseteq \Sigma_v$, proving
that~$\Sigma_v$ is open.
Since the singleton sets~$\{v\}$ form a base of the discrete topology on~$\V$,
continuity follows.

Let~$\gamma\in\Sigma_v$.
Let~$T$ be a time greater than both the latest decision time of the processes
in $\ob(\gamma)$ and the latest time any process becomes disobedient in
execution $\gamma = (C^t)_{t\geq0}$.
By the Termination property and the fact that disobedient processes cannot
become obedient again, we have $T < \infty$.
Because~$T$ is larger than the latest time a process becomes disobedient, we
have $\ob(\gamma) = \ob(C^T)$.

Using the notation $\gamma = (C^t)_{t\geq0}$ and $\delta = (D^t)_{t\geq 0}$,
we have:
\begin{equation}
\begin{split}
B_{2^{-T}}(\gamma)
& =  
\big\{
\delta\in\Sigma \mid 
\dunif(\gamma, \delta) < 2^{-T}
\big\}
\\ & =  
\big\{
\delta\in\Sigma \mid 
\exists p\in\Pi\colon
d_p(C^t, D^t) < 2^{-T} 
\big\}
\\ &=
\big\{
\delta\in\Sigma \mid 
\exists p\in\Pi\ 
\forall t\leq T\colon
C^t \sim_p D^t
\wedge
p\in \ob(C^t) \cap \ob(D^t)
\big\}
\\& =
\big\{
\delta\in\Sigma \mid 
\exists p\in\Pi\colon 
C^T \sim_p D^T
\wedge
p\in \ob(C^T) \cap \ob(D^T)
\big\}
\end{split}
\end{equation}
If $\delta\in B_{2^{-T}}(\gamma)$,
then $C^T \sim_p D^T$ for some $p\in \ob(C^T)\cap\ob(D^T)$.
Since~$p$ has decided $\Delta(\gamma)$ at time~$T$ in execution~$\gamma$
and~$p$ is obedient until time~$T$ in execution~$\delta$,
process~$p$ has also decided $\Delta(\gamma)$ at time~$T$ in
execution~$\delta$.
By Uniform Agreement and Termination, all processes in $\ob(\delta)$ decide
$\Delta(\gamma)=v$ as well.
In other words $B_{2^{-T}}(\gamma)\subseteq \Sigma_v$,
which concludes the proof.
\end{proof}

For an illustration in our non-communicating two-process example,
denote by~$\gamma^{(T)}$ the execution in which process~$1$
has initial value~$0$, process~$2$ has initial value~$1$, and process~$1$
becomes disobedient at time~$T$.
Similarly, denote by~$\delta^{(U)}$ the execution with the same initial values
and in which process~$2$ becomes disobedient at time~$U$.
Since there is no means of communication between the two processes,
by Validity, each obedient process necessarily has to eventually decide on its
own initial value, i.e.,
$\Delta(\gamma^{(T)}) = 1$
and
$\Delta(\delta^{(T)}) = 0$.
The uniform distance between these executions is equal to
$\dunif(\gamma^{(T)}, \delta^{(U)}) = 2^{-\max\{T,U\}}$.
Thus, every $\varepsilon$-neighborhood of~$\gamma^{(T)}$ contains
execution~$\delta^{(U)}$ if $U$ is chosen large enough to ensure
$2^{-U} < \varepsilon$.
The set of $1$-deciding executions is thus not open in the uniform topology.
But this means that the algorithm's decision function~$\Delta$ cannot be
continuous. \cref{lem:cont:unif:consensus} hence implies that there is no 
uniform consensus algorithm in the non-communicating two-process model 
(which is also confirmed by the more realistic application example in \cref{sec:generalomissions}).

\subsection{Non-uniform topology for executions}
\label{sec:mintopology}

Whereas the $p$-view distance function given by \cref{eq:Pviewpseudometric}
is also adequate for non-uniform consensus, this is not the case for the
uniform distance function as defined in \cref{eq:dunif}.
The appropriate
\emph{non-uniform minimum topology} (abbreviated \emph{non-uniform topology}) 
on the set~$\Sigma$ of executions is induced by the distance function
\begin{equation}
\dnonunif(\gamma, \delta)
=
\begin{cases}
\min_{p\in\ob(\gamma)\cap\ob(\delta)} d_p(\gamma,\delta)
&
\text{if } \ob(\gamma)\cap\ob(\delta) \neq \emptyset
\\
1
&
\text{if } \ob(\gamma)\cap\ob(\delta) = \emptyset
\enspace.
\end{cases}\label{eq:dnonunif}
\end{equation}

Like for $\dunif$, neither definiteness nor the triangle inequality need to be
satisfied by $\dnonunif$. The resulting non-uniform topology is finer than the uniform topology, however, since
the minimum is taken over the smaller set
$\ob(\gamma)\cap\ob(\delta)\subseteq \Pi$, which means that
$\dunif(\gamma,\delta) \leq \dnonunif(\gamma,\delta)$.
In particular, this implies that every decision function that is continuous with
respect to the uniform topology is also continuous with respect to the 
non-uniform topology.
Of course, this also follows from Lemma~\ref{lem:cont:unif:consensus}
and the fact that every uniform consensus algorithm also solves
non-uniform consensus.

The following \cref{lem:cont:nonunif:consensus} is the analog of \cref{lem:cont:unif:consensus}:

\begin{lemma}\label{lem:cont:nonunif:consensus}
Let $\Delta:\Sigma\to\V$ be the consensus decision function of a
non-uniform consensus algorithm.
Then, $\Delta$ is continuous with respect to
the non-uniform distance function~$\dnonunif$.
\end{lemma}
\begin{proof}
We again prove that every inverse image $\Sigma_v = \Delta^{-1}[\{v\}]$ of a value $v\in\V$ is open.

Let $\gamma\in\Sigma_v$.
Let~$T$ be the latest decision time of the processes in $\ob(\gamma)$ in
execution~$\gamma$.
By the Termination property, we have $T < \infty$.
Using the notation $\gamma = (C^t)_{t\geq0}$ and $\delta = (D^t)_{t\geq 0}$,
we have:
\begin{equation}
\begin{split}
B_{2^{-T}}(\gamma)
& =  
\big\{
\delta\in\Sigma \mid 
\dnonunif(\gamma, \delta) < 2^{-T}
\big\}
\\ & =  
\big\{
\delta\in\Sigma \mid 
\exists p \in \ob(\gamma) \cap \ob(\delta)\colon
d_p(\gamma, \delta) < 2^{-T}
\big\}
\\ & =  
\big\{
\delta\in\Sigma \mid 
\exists p \in \ob(\gamma) \cap \ob(\delta)\colon
\forall t\leq T\colon
C^t \sim_p D^t
\big\}
\end{split}
\end{equation}
If $\delta\in B_{2^{-T}}(\gamma)$, then $C^T \sim_p D^T$ for some $p\in \ob(\gamma)\cap\ob(\delta)$.
Denote by~$T_p$ the decision time of process~$p$ in~$\gamma$.
Since $T_p \leq T$, we also have $C^{T_p} \sim_p D^{T_p}$
But this means that process~$p$ decides value $\Delta(\gamma)$ at time~$T_p$
in both executions~$\gamma$ and~$\delta$, hence $\Delta(\delta) = \Delta(\gamma)=v$ and $B_{2^{-T}}(\gamma)\subseteq \Sigma_v$.
\end{proof}

For an illustration in the non-communicating two-process example used in
\cref{sec:uniftopology}, note that
the trivial algorithm that immediately decides on its initial value
satisfies $\Delta(\gamma^{(T)}) = 1$ and $\Delta(\delta^{(U)}) = 0$.
The algorithm does solve non-uniform consensus, since it is guaranteed that
one of the processes eventually becomes disobedient.
In contrast to the uniform distance function, the non-uniform distance
function satisfies
$\dnonunif(\gamma^{(T)}, \delta^{(U)}) = 1$
since
$\ob(\gamma^{(T)}) \cap \ob(\delta^{(U)}) = \emptyset$.
This means that the minimum distance between any $0$-deciding and any $1$-deciding
execution is at least~$1$.
It is hence possible to separate the two sets of executions by sets that
are open in the non-uniform topology, so consensus is solvable here,
according to the considerations in the following section. Again, this is
confirmed by the more realistic application example in \cref{sec:generalomissions}.

\section{General Consensus Characterization for Full-Information Executions}
\label{sec:consensus}

In this section, we will provide our main topological conditions for uniform
and non-uniform consensus solvability.

\begin{definition}[$v$-valent execution]\label{def:vvalentexec}
We call an execution $\gamma_v \in \Sigma$, for $v\in\V$, $v$-valent, if it
starts from an initial configuration $I$ where all processes $p\in\Pi$ have the
same initial value $I_p=v$.
\end{definition}

\begin{theorem}[Characterization of uniform consensus]\label{thm:char:unif}
Uniform consensus is solvable if and only if there exists
a partition of the set~$\Sigma$ of admissible executions
into sets $\Sigma_v$, $v\in\V$, such that
the following holds:
\begin{enumerate}
\item Every $\Sigma_v$ is a clopen set in~$\Sigma$ with respect to the
uniform topology induced by~$\dunif$.
\item If execution $\gamma\in\Sigma$ is $v$-valent, then $\gamma \in \Sigma_v$.
\end{enumerate}
\end{theorem}
\begin{proof}
($\Rightarrow$):
Define $\Sigma_v = \Delta^{-1}(v)$, where~$\Delta$ is the decision function
of a uniform consensus algorithm.
This is a partition of~$\Sigma$ by Termination, and
Validity implies property~(2).
It thus only remains to show openness of the~$\Sigma_v$ (which
immediately implies clopenness, as $\Sigma\setminus\Sigma_v = \bigcup_{v\neq w \in \V} \Sigma_w$ must be open), which follows from the continuity of $\Delta: \Sigma \to \V$,
since every singleton set~$\{v\}$ is open in the discrete
topology.

($\Leftarrow$):
We define a uniform consensus algorithm by defining the decision functions
$\Delta_p:\C\to\V\cup\{\perp\}$ as
\begin{equation}
\Delta_p(C)
=
\begin{cases}
v & \text{if } 
\left\{ \delta\in\Sigma \mid \exists t\colon C \sim_p D^t \right\} 
\subseteq \Sigma_v, \\
\perp & \text{otherwise},
\end{cases}
\end{equation}
where we use the notation $\delta = (D^t)_{t\geq0}$.
The function~$\Delta$ is well defined since the sets~$\Sigma_v$ are 
pairwise disjoint.

We first show Termination of the resulting algorithm.
Let $\gamma\in\Sigma$, let $v\in\V$ such that $\gamma\in\Sigma_v$,
and let $p\in\ob(\gamma)$.
Since~$\Sigma_v$ is open with respect to the uniform topology, there exists
some $\varepsilon>0$ such that
$\left\{ \delta\in \Sigma \mid \dunif(\gamma,\delta) < \varepsilon \right\}
\subseteq
\Sigma_v$.
By definition of~$\dunif$, we have 
$\dunif(\gamma,\delta) \leq d_p(\gamma,\delta)$ and hence
$\left\{ \delta\in \Sigma \mid d_p(\gamma,\delta) < \varepsilon \right\}
\subseteq
\left\{ \delta\in \Sigma \mid \dunif(\gamma,\delta) < \varepsilon \right\}
\subseteq
\Sigma_v$.

Writing $\gamma = (C^t)_{t\geq0}$,
let~$T$ be the smallest integer such that $2^{-\chi_p(C^t)} \leq \varepsilon$ for all $t\geq T$.
Such a~$T$ exists since $\chi_p(C^t)\to\infty$ as $t\to\infty$.
Then, for every $t\geq T$, we have
$\left\{ \delta\in \Sigma \mid \exists s\colon C^t \sim_p D^s \right\}
\subseteq
\left\{ \delta\in \Sigma \mid d_p(\gamma,\delta) < 2^{-\chi_p(C^t)} \right\}
\subseteq
\Sigma_v$.
In particular, $\Delta_p(C^t) = v$ for all $t\geq T$, i.e., process~$p$
decides value~$v$ in execution~$\gamma$.

We next show Uniform Agreement.
For the sake of a contradiction, assume that process~$q$ decides value $w\neq v$
in configuration $C$ in execution $\gamma\in\Sigma_v$.
But then, by definition of the function~$\Delta_q$, we have
$\gamma \in \left\{ \delta\in\Sigma \mid \exists t\colon C \sim_q D^t \right\}
\subseteq \Sigma_w$.
But this is impossible since $\Sigma_v\cap\Sigma_w = \emptyset$.

Validity immediately follows from property~(2).
\end{proof}

\begin{theorem}[Characterization of non-uniform consensus]\label{thm:char:nonunif}
Non-uniform consensus is solvable if and only if there exists
a partition of the set~$\Sigma$ of admissible executions
into sets $\Sigma_v$, $v\in\V$, such that
the following holds:
\begin{enumerate}
\item Every $\Sigma_v$ is a clopen set in~$\Sigma$ with respect to the
non-uniform topology induced by~$\dnonunif$.
\item If execution $\gamma\in\Sigma$ is $v$-valent, then $\gamma \in \Sigma_v$.
\end{enumerate}
\end{theorem}
\begin{proof}
The proof is similar to that of Theorem~\ref{thm:char:unif}, except that the
definition of~$\Delta_p$ is 
\begin{equation}
\Delta_p(C)
=
\begin{cases}
v & \text{if } 
\left\{ \delta\in\Sigma \mid \exists t\colon C \sim_p D^t \wedge p\in\ob(\delta) \right\} 
\subseteq \Sigma_v, \\
\perp & \text{otherwise}
\enspace,
\end{cases}
\end{equation}
i.e., we just have to add the constraint that $p\in\ob(\delta)$ to the executions
considered in the proof.
\end{proof}

If~$\Sigma$ has only finitely many connected components, i.e., only finitely many maximal connected sets,
then every connected component is necessarily clopen.
Consequentely,
these characterizations give rise to the following meta-procedure for
determining whether consensus is solvable and constructing an algorithm if
it is.
It requires knowledge of the connected components
of the space~$\Sigma$ of admissible executions with respect to the appropriate topology:
\begin{enumerate}
\item Initially, start with an empty set~$\Sigma_v$ for every value~$v\in\V$.
\item Add to~$\Sigma_v$ the connected components of~$\Sigma$ that contain an execution with a $v$-valent initial configuration.
\item Add any remaining connected component of~$\Sigma$ to an arbitrarily
chosen set~$\Sigma_v$.
\item If the sets~$\Sigma_v$ are pairwise disjoint, then consensus is solvable.
In this case, the sets~$\Sigma_v$ determine a consensus algorithm via the
universal algorithm given in the proofs of \cref{thm:char:unif} and
\cref{thm:char:nonunif}.
If the~$\Sigma_v$ are not pairwise disjoint, then consensus is not solvable.
\end{enumerate}

Obviously, our solution algorithms need to know the decision sets 
$\Sigma_v$, $v\in\V$.
As they usually contain uncountable many infinite executions,
the question of how to obtain them in practice appears. In 
\cref{sec:applications}, we will provide several instances
of \emph{labeling algorithms}, which can be used here.
They are based on labeling prefixes of executions, so can
in principle even be computed incrementally by the processes
on-the-fly during the executions.

\section{Limit-based Consensus Characterization}
\label{sec:fairunfair}

It is possible to shed some additional light on our general consensus characterization
by considering limit points. In particular,
\cref{cor:consensusseparation} will show that consensus 
is impossible if and only if certain limit points in the appropriate topologies 
are admissible.

\begin{definition}[Distance of sets]\label{def:distancesets}
For $A,B \subseteq \Comega$ with distance function $d$, let $d(A,B)=\inf\{d(\alpha,\beta)\mid \mbox{$\alpha\in A$, $\beta\in B$}\}$.
\end{definition}

Before we state our general results, we illustrate the underlying
idea in a slightly restricted setting, namely when
the underlying space $\Sigma$ of configuration sequences is contained in a compact set $K\subseteq\Comega$.
Whereas one cannot assume this in
general, it can be safely assumed in settings where the operationalization
of our system model based on process-time graphs, as described in \cref{sec:model}, 
applies: Since the set of all process-time graphs $\PTomega$
turns out to be compact and the transition function $\otau: \PTomega \to \Comega$ is continuous, according to \cref{lem:tau:is:cont}, we can consider the compact set
$K=\otau(\PTomega)$ instead of $\Comega$.
In this case, it is not difficult to show that $d_p(A,B)=0$ if and only if there is a sequence of executions~$\alpha_k$ in~$A$ and a sequence of executions~$\beta_k$ in $B$ such that both sequences converge to the same limit with respect to~$d_p$.

This distance-based characterization 
allows us to distinguish 3 cases that cause $d_p(A,B)=0$: 
(i) If $\hat{\alpha} \in A\cap B \neq \emptyset$, one can choose the sequences defined
by $\alpha_k=b_k=\hat{\alpha}=\hat{\beta}$, $k\geq 1$. (ii) If $A\cap B = \emptyset$ and $\hat{\alpha}=\hat{\beta}$, there
is a ``fair'' execution \cite{FG11} as the common limit.
(iii) If $A\cap B = \emptyset$ and $\hat{\alpha}\neq \hat{\beta}$, there is a pair of
``unfair'' executions \cite{FG11} acting as limits, which have distance 0 (and are
hence also common limits w.r.t.\ the distance function $d_p$). We note, however, that
due to the non-definiteness of the pseudometric $d_p$ (recall 
\cref{lem:Pseudometricproperties}) and the resulting 
non-uniqueness of limits in the $p$-view topology, (iii) are
actually two instances of (ii). \cref{cor:consensusimpfair} below will reveal
that consensus is solvable if and only if no decision set $\Sigma_v$ contains any 
fair or unfair execution w.r.t any $\Sigma_w$, $v\neq w$. 

Unfortunately, generalizing the above 
distance-based characterizaion from  $p$-view-topologies to 
the uniform and non-uniform topologies is not possible: Albeit every 
convergent infinite sequence $(\alpha^t)$ w.r.t. $\dunif$ (\cref{eq:dunif})
resp.\ $\dnonunif$ (\cref{eq:dnonunif}) also contains a convergent subsequence 
w.r.t.\ some (obedient) 
$d_p$ by the pigeonhole principle, one might observe a different $d_{p'}$ for the 
convergent subsequence of $(\beta^t)$. In this case, not even $d_p(\hat{\alpha},\hat{\beta})=0$
or $d_{p'}(\hat{\alpha},\hat{\beta})=0$ would guarantee $\dunif(A,B)=0$
resp.\ $\dnonunif(A,B)=0$, as the triangle inequality does not hold in these 
topologies. 
%
%

On the other hand, $\dunif(A,B)=0$ resp.\ $\dnonunif(A,B)=0$ is trivially guaranteed if 
it is the case that $\hat{\alpha} \in B$ or $\hat{\beta} \in A$: If, say, $\hat{\alpha} 
\in B$, one can choose the constant sequence $(\beta^t)=(\hat{\alpha})\in B^\omega$, 
which obviously converges to $\hat{\alpha}$ in any $p$-view topology, 
including the particular $d_p$ obtained for the convergent subsequence of $(\alpha^t)\to\hat{\alpha}$ 
by the abovementioned pigeonhole argument. Consequently, 
$d_p(A,B)=0$ and hence also 
$\dunif(A,B)=0$ resp.\ $\dnonunif(A,B)=0$. This implies the following ``if-part''
of our distance-based characterization, which even holds for non-compact $\Comega$:

\begin{lemma}[General zero-distance condition]\label{lem:zerodistancecondition}
Let $A,B$ be arbitrary subsets of $\Comega$ with distance function $d$. 
If there are infinite sequences $(\alpha_k) \in A^\omega$
and $(\beta_k) \in B^\omega$ of executions, as well as $\hat{\alpha},\hat{\beta}\in \Comega$
with $\alpha_k\to \hat{\alpha}$ and $\beta_k\to \hat{\beta}$ with $d(\hat{\alpha},\hat{\beta})=0$
and $\hat{\alpha} \in B$ or $\hat{\beta} \in A$, then d(A,B)=0.
\end{lemma}

In order to obtain the general limit-based consensus characterization stated
in \cref{cor:consensusseparation} below, we will not use set distances directly, however, but 
rather the following Separation \cref{lem:separation} from \cite{Munkres}:

\begin{lemma}[Separation Lemma {\cite[Lemma~23.12]{Munkres}}]\label{lem:separation}
If $Y$ is a subspace of $X$, a separation of $Y$ is
a pair of disjoint nonempty sets $A$ and $B$ whose union is $Y$, neither of
which contains a limit point of the other. The space $Y$ is connected
if and only if there exists no separation of $Y$. Moreover, $A$ and $B$ of
a separation of $Y$ are clopen in $Y$.
\end{lemma}
\begin{proof}
The closure of a set $A$ in $Y$ is $(\cl{A}\cap Y)$, where $\cl{A}$ denotes the closure
in $X$. To show that $Y$ is not connected implies a separation, assume that $A, B$ are 
closed and open in $Y=A\cup B$, so $A=(\cl{A}\cap Y)$. Consequently,
$\cl{A}\cap B = \cl{A}\cap(Y-A) = \cl{A}\cap Y - \cl{A}\cap A = \cl{A}\cap Y - A = \emptyset$. Since $\cl{A}$ is
the union of $A$ and its limit points, none of the latter is in $B$. An analogous argument shows that none of
the limit points of $B$ can be in $A$.

Conversely, if $Y=A\cup B$ for disjoint non-empty sets $A$, $B$ which do not contain limit
points of each other, then $\cl{A}\cap B=\emptyset$ and $A \cap \cl{B} = \emptyset$.
From the equivalence above, we get $\cl{A}\cap Y = A$ and $\cl{B}\cap Y =B$, so
both $A$ and $B$ are closed in $Y$ and, as each others complement, also open in $Y$ as well.
\end{proof}

Applying \cref{lem:separation} to the findings of \cref{thm:char:unif} resp.\ \cref{thm:char:unif},
the following general consensus characterization can be proved:

\begin{theorem}[Separation-based consensus characterization]\label{cor:consensusseparation}
Uniform resp.\ non-uniform consensus is solvable in a model if and only if
there exists a partition of the set of admissible executions $\Sigma$ into decision sets $\Sigma_v,v\in\V$, such that
the following holds:
\begin{enumerate}
\item No $\Sigma_v$ contains a limit point of any other $\Sigma_w$ w.r.t.\ the uniform resp.\
non-uniform topology in~$\Comega$.
\item Every $v$-valent admissible execution $\gamma_v$ satisfies $\gamma_v\in \Sigma_v$.
\end{enumerate}
If consensus is not solvable, then $\dunif(\Sigma_v,\Sigma_w)=0$ resp.\ $\dnonunif(\Sigma_v,\Sigma_w)=0$ for some $w\neq v$.
\end{theorem}
\begin{proof}
($\Leftarrow$)
We need to prove that if (1) and (2) in the statement of our theorem hold,
then consensus is solvable by means of the algorithm given in 
\cref{thm:char:unif} resp.\ \cref{thm:char:nonunif}. This only
requires showing that all of the finitely many $\Sigma_v$, $v\in \V$, 
are clopen in $\Sigma$, which immediately follows from \cref{lem:separation} 
since $\Sigma_v$ and $\Sigma\setminus\Sigma_v$ form a separation of $\Sigma$.

($\Rightarrow$)
We prove the contrapositive, by showing that if (1) and (2) do not hold,
then either some $\Sigma_v$ is not closed or $\Sigma_v \cap \Sigma_w \neq 0$, 
which does not allow to solve
consensus by \cref{thm:char:unif} resp.\ \cref{thm:char:nonunif}.
If, say, $A=\Sigma_v$ contains any limit point of $B=\Sigma_w$ for
$v\neq w$, this means that there is a sequence of executions 
$(\beta_k) \in B^\omega$ with limit $\beta_k \to \beta$ and
some $\alpha \in A \subseteq \Sigma$ with $\dunif(\alpha,\beta)=0$ resp.\ 
$\dnonunif(\alpha,\beta)=0$. According to \cref{lem:zerodistancecondition}, 
we have $\dunif(A,B)=0$ resp.\ $\dnonunif(A,B)=0$ in
this case. If $\alpha \not\in B$, then $B=\Sigma_w$ 
is not closed, if $\alpha \in B$, then $A \cap B \neq \emptyset$,
which provides the required contradiction in either case.
\end{proof}

Note that \cref{cor:consensusseparation} immediately implies the following properties
of the distances of the decision sets in the case consensus is solvable
in a model:

\begin{corollary}[General decision set distances]\label{cor:setdistance:min}
If uniform resp.\ non-uniform consensus is solvable in a model, it may nevertheless
be the case that $\dunif(\Sigma_v,\Sigma_w)= 0$ resp.\ $\dnonunif(\Sigma_v,\Sigma_w)= 0$ 
for some $v, w\neq v$. On the other hand, if $\dunif(\Sigma_v,\Sigma_w)>0$ resp.\ 
$\dnonunif(\Sigma_v,\Sigma_w)>0$ for all $v,w\neq v$, then 
uniform resp.\ non-uniform consensus is solvable.
\end{corollary}

Our characterization \cref{cor:consensusseparation} can also be expressed
via the exclusion of fair/unfair executions \cite{FG11}:

\begin{definition}[Fair and unfair executions]\label{def:fairunfair}
Consider two executions $\rho, \rho' \in \Comega$
of some consensus algorithm with decision sets $\Sigma_v$, $v\in\V$, in any appropriate topology:
\begin{itemize}
\item $\rho$ is called \emph{fair}, if for some $v,w\neq v 
\in \V$ there are convergent sequences $(\alpha_k) \in \Sigma_v$
and $(\beta_k) \in \Sigma_w$ with
$\alpha_k\to \rho$ and $\beta_k\to \rho$.
\item $\rho$, $\rho'$ are called a pair of \emph{unfair} executions,  
if for some $v,w\neq v 
\in \V$ there are convergent sequences $(\alpha_k) \in \Sigma_v$
with $\alpha_k\to \rho$ and $(\beta_k) \in \Sigma_w$ with $\beta_k\to \rho'$
and $\rho$ and $\rho'$ have distance 0.
\end{itemize}
\end{definition}

From \cref{cor:consensusseparation}, we immediately obtain:

\begin{corollary}[Fair/unfair consensus characterization]\label{cor:consensusimpfair}
Condition (1) in \cref{cor:consensusseparation} is equivalent to requireing
that the decision sets $\Sigma_v$, $\Sigma_w$ for $w\neq v$ neither contain
any fair execution nor any pair $\rho,\rho'$ of unfair executions.
\end{corollary}

An illustration of our limit-based characterizations is provided by Figure~\ref{fig:noncompactMA}.
Note carefully that, in the uniform case, a fair/unfair execution
$\rho$ where some process $p$ becomes disobedient in round $t$ implies that the
same happens in all $\alpha\in B_{2^{-t}}(\rho) \cap \Sigma_v$ 
and $\beta\in B_{2^{-t}}(\rho) \cap \Sigma_w$. On the other hand, if
$p$ does not become disobedient in $\rho$, it may still be the case that $p$
becomes disobedient in every $\alpha_k$ in the sequence converging to $\rho$, 
at some time $t_k$ with $\lim_{k\to\infty} t_k=\infty$. In
the non-uniform case, neither of these possibilities exists:
$p$ cannot be disobedient in the limit $\rho$, and any $\alpha_k$ where $p$ 
is not obedient
is also excluded as its distance to any other sequence is 1.

%

\section{Consensus Characterization in Terms of Broadcastability}
\label{sec:broadcastability}

We will now develop another characterization of consensus solvability, 
with rests on the broadcastability of the decision sets $\Sigma_v \subseteq \Sigma$
and their connected components $\Sigma_{\gamma} \subseteq \Sigma_v$. It will explain topologically why the existence of a broadcaster is mandatory for solving
the ``standard version'' of consensus, where any assignment of inputs from $\V$ is permitted. We start with some definitions needed for formalizing this condition:

\begin{definition}[Heard-of sets]\label{def:HO}
For every process~$p\in\Pi$, there is a function $\ho_p:\C\to2^\Pi$ that maps a configuration~$C\in\C$ to the set of processes~$\ho_p(C)$ that~$p$ has (transitively) heard of in~$C$. Its extension to execution $\gamma=(C^t)_{t\geq 0}$ 
is defined as $\ho_p(\gamma) = \bigcup_{t\geq0} \ho_p(C^t)$.

Heard-of sets have the following obvious properties: For executions $\gamma = (C^t)_{t\geq0}$, $\delta = (D^t)_{t\geq0}$ and all $t\geq 0$,
\begin{enumerate}
\item[(i)] $p\in\ho_p(C^t)$, and $\ho_p(C^t) = \ho_p(D^t)$ if $C^t \sim_p D^t$,
\item[(ii)] $\ho_p(C^t)\subseteq \ho_p(C^{t+1})$,
\item[(iii)] for all $x\in\Pi$, if $x\in \ho_q(C^t) \cap \ho_q(D^t)$ and $C^t \sim_q D^t$, then $I_x(\gamma) = I_x(\delta)$ (where~$I_p(\gamma)$ denotes the initial value of process~$p$ in execution~$\gamma$).
\end{enumerate}
\end{definition}

The independent arbitrary input assignment condition stated in \cref{def:independentinputassumption} secures that, for every execution $\gamma$ with initial value assignment $I(\gamma)$, there is a
an isomorphic execution $\delta$ w.r.t.\ the HO sets of all processes 
that starts from an arbitrary other initial value assignment $I(\delta)$.

\begin{definition}[Independent arbitrary input assignment condition]\label{def:independentinputassumption}
Let $I:\Pi\to\V$ be some assignment of initial values to the processes, and
$\Sigma^{(I)}\subseteq \Sigma$ be the set of admissible executions with that initial value assignment. We say that $\Sigma$ satisfies the \emph{independent input
assignment condition}, if and only if 
for any two assignments~$I$ and~$J$, we have $\Sigma^{(I)} \cong \Sigma^{(J)}$,
that is, there is a bijective mapping $f_{I,J}:\Sigma^{(I)}\to\Sigma^{(J)}$ such that
for all $\gamma = (C^t)_{t\geq 0} \in \Sigma^{(I)}$ and $\delta = (D^t)_{t\geq 0} \in \Sigma^{(I)}$, writing $f_{I,J}(\gamma) = (C_f^t)_{t\geq 0}$ and $f_{I,J}(\delta) = (D_f^t)_{t\geq 0}$, the following holds for all $t\geq 0$ and all $p\in\Pi$:
\begin{enumerate}
\item $\ob(C^t) = \ob(C_f^t)$ 
\item $C^t \sim_p D^t$ if and only if $C_f^t \sim_p D_f^t$
\item $\ho_p(C^t) = \ho_p(C_f^t)$
\item $C^t \sim_p C_f^t$ if $I_q = J_q$ for all $q\in \ho_p(C^t)$
\end{enumerate}
We say that $\Sigma$ satisfies the \emph{independent arbitrary input
assignment condition}, if it satisfies the independent input assignment
condition for every choice of $I:\Pi\to\V$.
\end{definition}

In the main results  of this section (\cref{thm:equivuniform} resp.\ \cref{thm:equivnonuniform}), we will not only provide a necessary and sufficient condition for solving this variant of uniform resp.\ non-uniform consensus based on broadcastability, but also establish the
general equivalence of weak validity (V) and strong validity (SV) (recall
\cref{def:consensus}). For binary consensus, i.e., $|\V|=2$, this is a well-known 
fact \cite[Ex.~5.1]{AW04}, for larger input sets, it was, to the best of 
our knowledge, not known yet.

Since the concise but quite technical proofs of \cref{thm:equivuniform} and \cref{thm:equivnonuniform} somehow obfuscate the actual cause of this equivalence
(and the way we actually discovered it), we first provide an alternative
explanation based on the broadcastability of connected components in the
following \cref{sec:bcconnectedcomponents}, which also allows us to establish
some basic results needed in \cref{sec:closed}.

\subsection{Broadcastability of connected components}
\label{sec:bcconnectedcomponents}

\cref{lem:broadcastableCCs} below reveals that
if consensus (with weak validity)
and independent arbitrary inputs is solvable, then every connected
component of~$\Sigma$ needs to be broadcastable.

\begin{definition}[Broadcastability]\label{def:broadcastability}
We call a subset $A \subseteq \Sigma$ of admissible executions \emph{broadcastable} 
by the broadcaster $p\in \Pi$, if, in every execution $\gamma \in A$, every obedient process $q\in\ob(\gamma)$ eventually hears from process~$p$, i.e., $p\in\ho_q(\gamma)$, and hence knows $I_p(\gamma)$.
\end{definition}

\begin{lemma}[Broadcastable connected components]\label{lem:broadcastableCCs}
A connected component $\Sigma_\gamma$ of a set of admissible executions 
$\Sigma$ for uniform resp.\ non-uniform 
consensus with independent arbitrary input assignments that is not
broadcastable for some process contains $w$-valent executions for every
$w\in \V$. In order to solve uniform resp.\ non-uniform consensus with
  independent arbitrary input assignments, every connected
  component must hence be broadcastable by some process, and lead to the
same decision value in each of its executions.
\end{lemma}
\begin{proof}
To prove the first part of our lemma,
we consider the finite sequence of executions $\gamma=\alpha_0,\alpha_1,\dots,\alpha_n=\gamma_w$
obtained from $\gamma$ by changing the initial values of the processes $1,\dots,n$ in $I(\gamma)$ to an arbitrary but fixed $w$, one by one (it is here where we need the arbitrary
input assignment assumption). We show by induction that $\alpha_p \in \Sigma_{\gamma}$ for every $p\in \{0,\dots,n\}$, which proves our claim since $\alpha_n=\gamma_w$. 

The induction basis $p=0$ is trivial, so suppose $\alpha_{p-1}\in \Sigma_{\gamma}$ according to the
induction hypothesis. If it happens that $I_p(\alpha_{p-1})=I_{p}(\gamma)=w$ already, nothing needs to be done and we just set $\alpha_{p}=\alpha_{p-1} \in \Sigma_{\gamma}$. Otherwise, $\alpha_{p}$ is $\alpha_{p-1}$ with the initial value $I_{p}(\alpha_{p})$ changed to $w$. Now suppose for a contradiction that $\alpha_{p} \in\Sigma_{\alpha_{p}} \neq \Sigma_{\gamma}$. 

Since $\Sigma_{\gamma}$ is not broadcastable by any process, hence also not by $p$, there is some
execution $\eta\in \Sigma_{\gamma}$ with $\eta = (C^t)_{t\geq0}$ and a process $q\neq p$ with $q \in \ob(C^t)$ and the initial value $I_p(\eta)$ not in $q$'s view $V_{q}(C^t)$ for every $t\geq 0$.
Thanks to the independent input assignment property \cref{def:independentinputassumption}, there is 
also an execution $\delta=f_{I(\eta),I'}(\eta) \in \Sigma_{\alpha_{p}}$ that
matches $\eta$, i.e., is the same as $\eta$ except that $I(\delta)=I'$ with $I'_q=I_q(\eta)$ for $p \neq q \in \Pi$ but possibly $I'_p\neq I_p(\eta)$.
It follows that $d_{q}(\eta,\delta)=0$ with $q\in\ob(\eta)\cap\ob(\delta)$ and hence  $\dunif(\eta,\delta)=0$ resp.\ $\dnonunif(\eta,\delta)=0$. Consequently, $\delta \in \Sigma_{\gamma}$
and hence $\Sigma_{\alpha_{p}}=\Sigma_{\gamma}$, which provides the required contradiction and
completes the induction step.

For the second part of our lemma, assume for a contradiction that there is a non-broadcastable 
connected component $\Sigma_\gamma$ in the decision set $\Sigma_v$ containing all 
the $v$-valent executions $\gamma_v$. By our previous result, it would also contain some $w$-valent execution $\gamma_w$, $w\neq v$. Consequently, $\Sigma_v \cap \Sigma_w \neq \emptyset$, which
makes consensus impossibly by Theorem~\ref{thm:char:unif} resp.\ \cref{thm:char:nonunif}.
That every $\delta \in \Sigma_\gamma$ leads to the same decision value $\Delta(\delta)=\Delta(\gamma)$
follows from the continuity of the decision function and the connectedness of $\Sigma_\gamma$.
\end{proof}

In addition, \cref{lem:broadcastablediameter} below reveals that \emph{any} connected broadcastable set has a diameter strictly smaller than~$1$.

\begin{definition}[Diameter of a set]\label{def:diameterset}
For $A \subseteq \Comega$, depending on the distance function $d$ that induces
the appropriate topology,
define $A$'s diameter as $d(A)=\sup\{d(\gamma,\delta)\mid \mbox{$\gamma, \delta\in A$}\}$.
\end{definition}

\begin{lemmarep}[Diameter of broadcastable connected sets]\label{lem:broadcastablediameter}
If a connected set $A\subseteq \Sigma$ 
of admissible executions is broadcastable by some process $p$, then
$\dunif(A) \leq d_{p}(A)\leq 1/2$, as well as $\dnonunif(A) \leq 1/2$, i.e., $p$'s initial value satisfies $I_p(\gamma)=I_p(\delta)$ for all $\gamma,\delta \in A$. 
\end{lemmarep}
\begin{proof}
  Our proof below for $d_p(A)\leq 1/2$ translates literally to any $d \in \{d_p, \dunif, \dnonunif\}$;
  the statement $\dunif(A) \leq d_{p}(A)$ follows from
the definition in \cref{eq:dunif}.

Broadcastability by $p$ implies that, for any $\gamma\in A$ with $\gamma=(C^t)_{t\geq0}$, every process $q$ has $I_p(\gamma)$ in its local view 
$V_{q}(C^{T(\gamma)})$ for some $0<T(\gamma)<\infty$ or is not obedient any more. 
Abbreviating $t=T(\gamma)$, consider any $\delta\in B_{2^{-t}}(\gamma) \cap A$
with $\delta = (D^t)_{t\geq0}$.
By definition of $B_{2^{-t}}(\gamma)$, there must be some process $q \in
\ob(D^t)\cap\ob(C^t)$ with $V_{q}(D^{t})=V_{q}(C^{t})$. \cref{def:HO}.(iii)
thus guarantees $I_p(\delta)=I_p(\gamma)$. 

We show
now that this argument can be continued to reach every $\delta\in A$. For a contradiction, suppose that this is not the case and let $U(\gamma)$ be the union of the balls
recursively defined as follows: $U_0(\gamma)=\{\gamma\}$, for $m>0$,
$U_m(\gamma) = \bigcup_{\delta \in U_{m-1}(\gamma)} (B_{2^{-T(\delta)}}(\delta)
\cap A)$,  and finally $U(\gamma)=\bigcup_{m\geq 0} U_m(\gamma)$.
As a union of open balls intersected with $A$, which are all open in $A$, both
$U_m(\gamma)$ for every $m > 0$ and $U(\gamma)$ is hence open in $A$.
For every $\delta \in A\setminus U(\gamma)$, $U(\delta)$ is also open in~$A$, and so is $V(\gamma)=\bigcup_{\delta \in A\setminus U(\gamma)}U(\delta)$. However, the open sets $U(\gamma)$ and $V(\gamma)$ must satisfy $U(\gamma) \cap V(\gamma) = \emptyset$ (as
they would be the same otherwise) and $U(\gamma)\cup V(\gamma)=A$, hence $A$ cannot be connected. 
\end{proof}

Together, \cref{lem:broadcastableCCs} and \cref{lem:broadcastablediameter} imply:

\begin{corollary}[Broadcastable $\Sigma_{\gamma}$]\label{cor:broadcastablePS}
If uniform resp.\ non-uniform consensus with
  independent arbitrary input assignments is solvable, 
then every connected component $\Sigma_{\gamma}\subseteq \Sigma$ must be broadcastable by some process $p$. In every execution $\gamma'\in\Sigma_{\gamma}$, the broadcaster $p$ has the same initial value $I_p(\gamma')$, and the decision value is the same $\Delta(\gamma')=\Delta(\gamma)$.
\end{corollary}


To emphasize the key role of the consequences of \cref{cor:broadcastablePS} for
the equivalence of weak validity (V) and strong validity (SV), where in (SV) the consensus decision value must be the initial value of some process, we first observe that
the transition from (V) to (SV) in our \cref{thm:char:unif} resp.\ \cref{thm:char:nonunif} just requires the replacement 
of condition 2., i.e., \emph{``If execution $\gamma\in\Sigma$ is $v$-valent, then $\gamma \in \Sigma_v$''}, by 
\emph{``If execution $\gamma\in\Sigma_v$, then there is a process $p$ with initial
initial value $I_p(\gamma)=v$''.} This change would result in strong versions
of our theorems, since the above modification is in fact transparent
for the proofs of \cref{thm:char:unif} and \cref{thm:char:nonunif}. Note also
that both versions are equivalent for $v$-valent executions.
Similarly, to obtain a strong version of our meta-procedure, step (3)
\emph{``Add any remaining connected component of~$\Sigma$ to an arbitrarily chosen set~$\Sigma_v$''} must be replaced by \emph{``Add every remaining connected 
component $\Sigma_\gamma \subseteq \Sigma$, where execution $\gamma \in \Sigma_\gamma$
is arbitrary, to any set~$\Sigma_v$, where $v$ is the initial value $I_b(\gamma)=v$ of a process $b$ that is a broadcaster in every execution $\gamma'\in\Sigma_\gamma$''.} 

The crucial role of \cref{cor:broadcastablePS} is that it makes this modification \emph{always} possible also in the case of multi-valued consensus (in the case of binary consensus, it is obvious), as it reveals that if weak consensus is solvable, then
every connected component $\Sigma_\gamma$ must have at least one common broadcaster $b=b(\gamma')=b(\Sigma_\gamma)$ that has the same initial value $I_b(\gamma')=I_b(\gamma)= I_b(\Sigma_\gamma)$ in all executions $\gamma'\in\Sigma_\gamma$. Consequently, if decision sets resp.\ a meta-procedure
exists that allows to solve consensus with weak validity according to
\cref{thm:char:unif} and \cref{thm:char:nonunif}, one can always
reshuffle the connected components to form \emph{strong} 
decision sets, which use the initial value of some broadcaster
for assigning a connected component to a decision set:

\begin{definition}[Strong decision sets]\label{def:strong} Let $\Sigma$ be the
set of admissible executions of any (weak or strong) consensus algorithm with independent arbitrary input assignments. A \emph{strong broadcaster
decision set} $\Sigma_v^p$ for broadcaster $p\in \Pi$ and decision value
$v\in \V$ resp.\ a \emph{strong decision set} $\Sigma_v$ for $v\in \V$
satisfies
\begin{equation}
\Sigma_v^p = \bigcup_{\substack{\gamma\in\Sigma\\ b(\Sigma_\gamma)=p \\ I_{p}(\gamma)=v}} \Sigma_\gamma \qquad\mbox{resp.}\qquad \Sigma_v = \bigcup_{p \in \Pi}  \Sigma_v^p 
\label{eq:Sigmavp}.
\end{equation}
\end{definition}
Note that strong decision sets need not be unique, as some connected component
$\Sigma_\gamma$ might have several broadcasters, any of which could be used for
determining its decision value $v$. The canonical choice to make it uniquely
defined is to take the lexically smallest $p=b(\Sigma_\gamma)$ among
all broadcasters $p' \geq p$ in $\Sigma_\gamma$. In the rest of our paper,
all strong decision sets will be canonical.

\medskip

Since the canonical strong decision sets that can be formed via the
abovementioned reshuffling are easily shown to
satisfy the strong versions of \cref{thm:char:unif} and
\cref{thm:char:nonunif}, one obtains the broadcast-based
characterization of consensus stated in \cref{thm:charbroadcastability}.
Rather than proving it by formalizing the reasoning sketched above,
however, we will rely on the general equivalence results
\cref{thm:equivuniform} resp.\ \cref{thm:equivnonuniform}
developed in in the following \cref{sec:bcgeneral}. This way,
the somewhat tedious and non-constructive reshuffling of connected
components involved in the direct proof can be replaced by an
explicit construction of the canonical strong decision sets,
which utilizes binary consensus.

\begin{theorem}[Consensus characterization via broadcastability]\label{thm:charbroadcastability}
A model allows to solve uniform resp.\ non-uniform consensus with independent
arbitrary input assignments if and only if it guarantees that
(i) every connected component $\Sigma_\gamma$ of the set~$\Sigma$ of admissible executions is broadcastable for some process $p=b(\Sigma_\gamma)$ starting
with the same input $I_p(\gamma')$ in $\gamma' \in \Sigma_{\gamma}$, and
(ii) that the strong broadcaster decision sets $\Sigma_v^p$, $p\in \Pi$, $v\in \V$, as specified in \cref{eq:Sigmavp}, are clopen 
in $\Sigma$ in the uniform topology resp.\ the non-uniform topology.
\end{theorem}
\begin{proof}
  The broadcastability of the connected components follows from
  \cref{cor:broadcastablePS}, the clopenness of the strong broadcaster
  decision sets will be established in the proofs of
  \cref{thm:equivuniform} resp.\ \cref{thm:equivnonuniform}.
\end{proof}

\subsection{General broadcast-based characterization}
\label{sec:bcgeneral}

We will now provide our general broadcast-based characterization for uniform
and non-uniform consensus with arbitrary and independent input assignments
according to \cref{def:independentinputassumption}. In a nutshell, it uses
a reduction to (the solvability of) binary consensus, where weak and strong validity are
trivially equivalent, for explicitly constructing the canonical strong
broadcaster decision sets.

Let $\hat{\Sigma} \subseteq \Sigma$ denote the set of admissible executions 
of a multi-valued consensus algorithm starting from a single initial value
assignment~$\hat{I}:\Pi \to \V$ (any will do, the choice is arbitrary).

\begin{definition}[Uniform/non-uniform broadcastability]\label{def:unifnonunifbroadcastability}
We say that~$\hat{\Sigma}$ is \emph{uniformly resp.\ non-uniformly broadcastable} if there exist sets $\hat{\Sigma}_p\subseteq \hat{\Sigma}$ for $p\in\Pi$ such that:
\begin{enumerate}
\item The sets~$\hat{\Sigma}_p$ are pairwise disjoint and $\bigcup_{p\in\Pi} \hat{\Sigma}_p = \hat{\Sigma}$.
\item Every~$\hat{\Sigma}_p$ is $\dunif$-clopen resp.\ $\dnonunif$-clopen in~$\hat{\Sigma}$.
\item Every $\hat{\Sigma}_p$ is broadcastable by $p$ but not by any lexically smaller $p'<p$, i.e., every obedient process $q\in\ob(\gamma)$ satisfies $p\in\ho_q(\gamma)$ for every $\gamma \in \hat{\Sigma}_p$.
\end{enumerate}
\end{definition}

\begin{theorem}\label{thm:binary:to:broadcast}
If uniform resp.\ non-uniform binary consensus with arbitrary and independent input assignments is solvable, then~$\hat{\Sigma}$ is uniformly resp.\ non-uniformly broadcastable.
\end{theorem}
\begin{proof}
By \cref{thm:char:unif} resp.~\cref{thm:char:nonunif} restricted to $|\V|=2$, there exists a clopen partition $(\Sigma_0,\Sigma_1)$ of~$\Sigma$ such that~$\Sigma_0$ includes all $0$-valent executions and~$\Sigma_1$ includes all $1$-valent executions.

For $0\leq p\leq n$, let~$I_p$ be the initial value assignment in which all processes~$q\leq p$ have initial value~$1$ and all processes $q > p$ have initial value~$0$.
The assignment~$I_0$ is the all-$0$ assignment and~$I_n$ is the all-$1$ assignment.
According to \cref{def:independentinputassumption}, there is an isomorphism $g_p = f_{I_{p-1},I_p}: \Sigma^{(I_{p-1})} \to \Sigma^{(I_p)}$ for every $1\leq p\leq n$, as well as an isomorphism $h = f_{I_n,\hat{I}}:\Sigma^{(I_n)}\to\hat{\Sigma}$.

We now inductively define the set $\Sigma_{p,q}$, for $1\leq p\leq n$ and $1\leq q\leq p$, which consists of all 1-deciding executions starting
from $I_p$ where $q$ is the lexically smallest broadcaster. Note that 
both $p$ and $p'$ might be broadcasters in $\gamma \in \Sigma_{p,q}$, 
provided $p < p'$.
\begin{itemize}
\item[(i)] $\Sigma_{1,1} = \Sigma^{(I_1)}_1$, the set of $1$-deciding executions when starting with initial value assignment~$I_1$.
\item[(ii)] For $2\leq p\leq n$ and $1\leq q\leq p-1$, we set $\Sigma_{p,q} = g_p[\Sigma_{p-1,q}]$.
\item[(iii)] For $2\leq p\leq n$ and $q=p$, we set $\Sigma_{p,p} = \Sigma^{(I_p)}_1 \setminus \bigcup_{q=1}^{p-1} \Sigma_{p,q}$.
\end{itemize}
A trivial induction reveals that, for every $1\leq p\leq n$,
$\Sigma_{p,q} \subseteq \Sigma^{(I_p)}$, and that the sets~$\Sigma_{p,q}$
are pairwise disjoint since all the~$g_p$ are bijective.
Furthermore, since the decision sets~$\Sigma^{(I_p)}_1$ are clopen
in~$\Sigma^{(I_p)}$ and the~$g_p$ are homeomorphisms, every~$\Sigma_{p,q}$ is
clopen in~$\Sigma^{(I_p)}$.

We now prove, by induction on~$p$, that every $1 \leq q \leq p$ is the lexically
smallest broadcaster in every execution $\gamma\in\Sigma_{p,q}$, i.e.,
that $q\in\ho_r(\gamma)$ for every $r\in\ob(\gamma)$, and that there is no
smaller $q'$ with this property.
We start with the base case $p=q=1$, which is obviously the lexically smallest.
Let~$\gamma \in\Sigma^{(I_1)}_1$ and $r\in\ob(\gamma)$.
Assuming by contradiction that $1\not\in \ho_r(\gamma)$, we get $\Delta_r(\gamma) = \Delta_r(g_1^{-1}(\gamma)) = 0$ by \cref{def:independentinputassumption}.(4)
and Validity (V). This contradicts $\gamma \in\Sigma^{(I_1)}_1$, however.
Now let $2\leq p\leq n$.
For all $1\leq q\leq p-1$ and all $\gamma \in \Sigma_{p,q}$, we have that $q\in\ho_r(\gamma) = \ho_r(g_q^{-1}(\gamma))$ for all $r\in\ob(\gamma) = \ob(g_q^{-1}(\gamma))$ is the lexically smallest broadcaster by the induction hypothesis. 
For $q = p$, assuming by contradiction that $p\not\in \ho_r(\gamma)$ for $\gamma\in\Sigma_{p,p}$ and $r\in\ob(\gamma)$, we get $\Delta_r(\gamma) = \Delta_r(g_p^{-1}(\gamma)) = 0$ by  \cref{def:independentinputassumption}.(4) and the fact that (iii) implies $\Sigma^{(I_{p-1})}_1 \subseteq \bigcup_{q=1}^{p-1} \Sigma_{p-1,q}$
since  $\Sigma_{p-1,p-1} = \Sigma^{(I_{p-1})}_1 \setminus \bigcup_{q=1}^{p-2} \Sigma_{p-1,q}$; the latter also guarantees that there is no lexically smaller
broadcaster. This completes our induction proof.

We finally set $\hat{\Sigma}_p = h[\Sigma_{n,p}]$ for $1\leq p\leq n$ and show that the result satisfies uniform broadcastability according to
\cref{def:unifnonunifbroadcastability}:
(1)
Pairwise disjointness of the~$\hat{\Sigma}_p$ follows from pairwise disjointness of the~$\Sigma_{n,p}$.
The fact that $\bigcup_{p=1}^n \hat{\Sigma}_p = \hat{\Sigma}$ follows from the definition of~$\Sigma_{n,n}$ and the fact that $\Sigma^{(I_n)}_1 = \Sigma^{(I_n)}$ by Validity.
(2)
Clopenness of the~$\hat{\Sigma}_p$ follows from clopenness of the~$\Sigma_{n,p}$ and the fact that~$h$ is a homeomorphism.
(3)
For every $\gamma\in\hat{\Sigma}_p$ and $q\in\ob(\gamma)$, we have $p\in \ho_q(\gamma) = \ho_q(h^{-1}(\gamma))$.
This concludes the proof.
\end{proof}

With this result, we can prove the following equivalences \cref{thm:equivuniform} 
resp.\ \cref{thm:equivnonuniform} for uniform and non-uniform consensus:

\begin{theorem}\label{thm:equivuniform}
For a set of admissible executions $\Sigma$ where uniform consensus with arbitrary and independent input assignments is solvable, the following statements are equivalent:
\begin{enumerate}
\item Uniform binary consensus is solvable.
\item Foy any input assignment $\hat{I}:\Pi \to \V$, the subset
  of admissible executions $\hat{\Sigma} \subseteq \Sigma$ using $\hat{I}$
is uniformly broadcastable.
\item Strong uniform consensus is solvable for any set~$\V$ of initial values.
\item Weak uniform consensus is solvable for any set~$\V$ of initial values.
\end{enumerate}
\end{theorem}
\begin{proof}
The implications (3)$\Rightarrow$(4)$\Rightarrow$(1) are trivial.
The implication (1)$\Rightarrow$(2) follows from Theorem~\ref{thm:binary:to:broadcast}.
To prove the implication (2)$\Rightarrow$(3), we give an algorithm that solves strong consensus, akin to those used in the proofs of Theorem~\ref{thm:char:unif}.

Let~$\hat{\Sigma}$ be broadcastable and let~$\hat{\Sigma}_p$ be sets as in the definition of broadcastability.
For every initial value assignment $I:\Pi\to\V$, let $g_I = f_{\hat{I},I}:\hat{\Sigma}\to \Sigma^{(I)}$ be the corresponding isomorphism.
For $p\in\Pi$ and $v\in\V$, we define the canonical strong broadcaster decision
sets
\begin{equation}
\Sigma_v^p
=
\bigcup_{\substack{I:\Pi\to\V\\I_p=v}} g_I[\hat{\Sigma}_p]
\qquad
\text{and}
\qquad
\Sigma_v
=
\bigcup_{p\in\Pi} \Sigma_v^p
\enspace. \label{eq:strongsetsuniform}
\end{equation}
The sets~$\Sigma_v^p$ are $\dunif$-open in~$\Sigma$:
For any $\gamma\in\Sigma_v^p$, let~$T$ be a time at which, (i) in execution~$\gamma$, all processes have heard from~$p$ and (ii) $B_{2^{-T}}(g_I^{-1}(\gamma)) \subseteq \hat{\Sigma}_p$ in~$\hat{\Sigma}$ for all $I:\Pi\to\V$ with $I_p = v$, and choose the neighborhood
\begin{equation}
\begin{split}
\mathcal{N}
& =
\left\{
\delta\in\Sigma \mid \dunif(\gamma,\delta) < 2^{-T}
\right\}
\\ & =
\left\{
\delta\in\Sigma \mid \exists q\in\Pi\colon C^T \sim_q D^T 
\right\}
\\ & =
\left\{
\delta\in\Sigma \mid \exists q\in\Pi\colon C^T \sim_q D^T \wedge p\in \ho_q(C^T) = \ho_q(D^T)
\right\}
\\ & \subseteq
\left\{
\delta\in\Sigma \mid I_p(\gamma) = I_p(\delta) = v
\right\}
\subseteq
\bigcup_{\substack{I:\Pi\to\V\\I_p = v}}
g_I[\hat{\Sigma}]
\end{split}
\end{equation}
where we use the notation $\gamma = (C^t)_{t\geq0}$ and $\delta = (D^t)_{t\geq0}$.
By assumption~(ii) on the choice of~$T$, for every $I:\Pi\to\V$ with $I_p = v$, we have 
\begin{equation}
\begin{split}
\mathcal{N} \cap g_I[\hat{\Sigma}]
& =
\left\{
\delta \in g_I[\hat{\Sigma}] \mid \dunif(\gamma, \delta) < 2^{-T}
\right\}
\\ & =
\left\{
\delta \in g_I[\hat{\Sigma}] \mid \dunif\left(g_I^{-1}(\gamma), g_I^{-1}(\delta)\right) < 2^{-T}
\right\}
\\ & =
\left\{
g_I(\delta)\mid \delta\in\hat{\Sigma} \wedge \dunif\left(g_I^{-1}(\gamma), \delta\right) < 2^{-T}
\right\}
\\ & \subseteq
\left\{
g_I(\delta)\mid \delta\in\hat{\Sigma}_p
\right\}
=
g_I[\hat{\Sigma}_p]
\enspace.
\end{split}
\end{equation}
Combining the last two equations, we get
\begin{equation}
\begin{split}
\mathcal{N}
=
\bigcup_{\substack{I:\Pi\to\V\\I_p = v}}
\left(\mathcal{N} \cap g_I[\hat{\Sigma}] \right)
\subseteq
\bigcup_{\substack{I:\Pi\to\V\\I_p = v}}
g_I[\hat{\Sigma}_p]
=
\Sigma_v^p
\enspace.
\end{split}
\end{equation}
The sets~$\Sigma_v^p$, as well as the sets~$\Sigma_v$ as unions of the~$\Sigma_v^p$, are thus $\dunif$-open in~$\Sigma$.
The~$\Sigma_v$ are pairwise disjoint since the~$\hat{\Sigma}_p$ are.
We further have $\Sigma = \bigcup_{v\in\V} \Sigma_v$.

We now define the strong consensus algorithm.
For every configuration $C\in\C$, we set
\begin{equation}
\Delta_q(C)
=
\begin{cases}
v & \text{if }
\{\delta\in\Sigma \mid \exists t\colon C\sim_q D^t \}
\subseteq
\Sigma_{v}
\\
\perp & \text{otherwise}
\end{cases}\label{eq:decisionfunctionuniform}
\end{equation}
The function~$\Delta_q$ is well-defined since the sets~$\Sigma_v$ are pairwise disjoint.

We first show Termination.
Let $\gamma \in \Sigma$, let $I:\Pi\to\V$ be the initial value assignment of~$\gamma$, and let $q\in\ob(\gamma)$.
Since~$\Sigma_v$ is $\dunif$-open in~$\Sigma$, there exists some $\varepsilon >
0$ such that $\{\delta\in\Sigma\mid d_q(\gamma,\delta) < \varepsilon\} = \{\delta\in\Sigma\mid \dunif(\gamma,\delta) <
\varepsilon\}\subseteq \Sigma_v$.
Letting~$T$ be the smallest integer such that $2^{-\chi_q(C^t)} \leq \varepsilon$ for all $t\geq T$, we get $\Delta_q(C^t) = v$ for all $t\geq T$, just like in the proof of Theorem~\ref{thm:char:unif}.

To show Uniform Agreement, assume by contradiction that process~$q$ decides a value $w\neq v$ in configuration~$C$ in execution $\gamma\in\Sigma_v$.
Then, by definition of~$\Delta_q$, we have $\gamma\in\{\delta\in\Sigma \mid \exists t\colon C \sim_q D^t\}\subseteq \Sigma_w$.
But this is impossible since $\Sigma_v \cap \Sigma_w = \emptyset$.

We finish the proof by showing Strong Validity.
Let $\gamma \in \Sigma_v$.
Then, by definition, there exists a $p\in\Pi$ and an $I:\Pi\to\V$ with $I_p = v$ such that $\gamma\in g_I[\hat{\Sigma}_p] \subseteq \Sigma^{(I)}$.
But then, in particular, $I_p(\gamma) = I_p = v$.
\end{proof}

\begin{theorem}\label{thm:equivnonuniform}
For a set of admissible executions $\Sigma$ where non-uniform consensus with arbitrary and independent input assignments is solvable, the following statements are equivalent:
\begin{enumerate}
\item Non-uniform binary consensus is solvable.
\item Foy any input assignment $\hat{I}:\Pi \to \V$, the subset
  of admissible executions $\hat{\Sigma} \subseteq \Sigma$ using $\hat{I}$
  is uniformly broadcastable.
\item Strong non-uniform consensus is solvable for any set~$\V$ of initial values.
\item Weak non-uniform consensus is solvable for any set~$\V$ of initial values.
\end{enumerate}
\end{theorem}
\begin{proof}
The proof is similar to that of \cref{thm:equivuniform}.

The implications (3)$\Rightarrow$(4)$\Rightarrow$(1) are trivial.
The implication (1)$\Rightarrow$(2) follows from Theorem~\ref{thm:binary:to:broadcast}.
To prove the implication (2)$\Rightarrow$(3), we give an algorithm that solves strong consensus, akin to those used in the proofs of Theorem~\ref{thm:char:nonunif}.

Let~$\hat{\Sigma}$ be broadcastable and let~$\hat{\Sigma}_p$ be sets as in the definition of broadcastability.
For an initial value assignment $I:\Pi\to\V$, let $g_I = f_{\hat{I},I}:\hat{\Sigma}\to \Sigma^{(I)}$ be the isomorphism.
For $p\in\Pi$ and $v\in\V$, we define the canonical strong broadcaster decision sets
\begin{equation}
\Sigma_v^p
=
\bigcup_{\substack{I:\Pi\to\V\\I_p=v}} g_I[\hat{\Sigma}_p]
\qquad
\text{and}
\qquad
\Sigma_v
=
\bigcup_{p\in\Pi} \Sigma_v^p
\enspace. \label{eq:strongsetsnonuniform}
\end{equation}
The sets~$\Sigma_v^p$ are $\dnonunif$-open in~$\Sigma$:
For any $\gamma\in\Sigma_v^p$, let~$T$ be a time at which, (i) in execution~$\gamma$, all processes have heard from~$p$ and (ii) $B_{2^{-T}}(g_I^{-1}(\gamma)) \subseteq \hat{\Sigma}_p$ in~$\hat{\Sigma}$ for all $I:\Pi\to\V$ with $I_p = v$, and choose the neighborhood
\begin{equation}
\begin{split}
\mathcal{N}
& =
\left\{
\delta\in\Sigma \mid \dnonunif(\gamma,\delta) < 2^{-T}
\right\}
\\ & =
\left\{
\delta\in\Sigma \mid \exists q\in\Pi\colon C^T \sim_q D^T \wedge q\in\ob(\gamma)\cap\ob(\delta)
\right\}
\\ & \subseteq
\left\{
\delta\in\Sigma \mid \exists q\in\Pi\colon C^T \sim_q D^T \wedge p\in \ho_q(C^T) = \ho_q(D^T)
\right\}
\\ & \subseteq
\left\{
\delta\in\Sigma \mid I_p(\gamma) = I_p(\delta) = v
\right\}
\subseteq
\bigcup_{\substack{I:\Pi\to\V\\I_p = v}}
g_I[\hat{\Sigma}]
\end{split}
\end{equation}
where we use the notation $\gamma = (C^t)_{t\geq0}$ and $\delta = (D^t)_{t\geq0}$.
By assumption~(ii) on the choice of~$T$, for every $I:\Pi\to\V$ with $I_p = v$, we have 
\begin{equation}
\begin{split}
\mathcal{N} \cap g_I[\hat{\Sigma}]
& =
\left\{
\delta \in g_I[\hat{\Sigma}] \mid \dnonunif(\gamma, \delta) < 2^{-T}
\right\}
\\ & =
\left\{
\delta \in g_I[\hat{\Sigma}] \mid \dnonunif\left(g_I^{-1}(\gamma), g_I^{-1}(\delta)\right) < 2^{-T}
\right\}
\\ & =
\left\{
g_I(\delta)\mid \delta\in\hat{\Sigma} \wedge \dnonunif\left(g_I^{-1}(\gamma), \delta\right) < 2^{-T}
\right\}
\\ & \subseteq
\left\{
g_I(\delta)\mid \delta\in\hat{\Sigma}_p
\right\}
=
g_I[\hat{\Sigma}_p]
\enspace.
\end{split}
\end{equation}
Combining the last two equations, we get
\begin{equation}
\begin{split}
\mathcal{N}
=
\bigcup_{\substack{I:\Pi\to\V\\I_p = v}}
\left(\mathcal{N} \cap g_I[\hat{\Sigma}] \right)
\subseteq
\bigcup_{\substack{I:\Pi\to\V\\I_p = v}}
g_I[\hat{\Sigma}_p]
=
\Sigma_v^p
\enspace.
\end{split}
\end{equation}
The sets~$\Sigma_v^p$, as well as the sets~$\Sigma_v$ as unions of the~$\Sigma_v^p$, are thus $\dnonunif$-open in~$\Sigma$.
The~$\Sigma_v$ are pairwise disjoint since the~$\hat{\Sigma}_p$ are.
We further have $\Sigma = \bigcup_{v\in\V} \Sigma_v$.

We now define the strong consensus algorithm.
For every configuration $C\in\C$, we set
\begin{equation}
\Delta_q(C)
=
\begin{cases}
v & \text{if }
\{\delta\in\Sigma \mid \exists t\colon C\sim_q D^t \wedge q\in \ob(\delta) \}
\subseteq
\Sigma_{v}
\\
\perp & \text{otherwise}
\end{cases} \label{eq:decisionfunctionnonuniform}
\end{equation}
The function~$\Delta_q$ is well-defined since the sets~$\Sigma_v$ are pairwise disjoint.

We first show Termination.
Let $\gamma \in \Sigma$, let $I:\Pi\to\V$ be the initial value assignment of~$\gamma$, and let $q\in\ob(\gamma)$.
Since~$\Sigma_v$ is $\dnonunif$-open in~$\Sigma$, there exists some $\varepsilon >
0$ such that $\{\delta\in\Sigma\mid d_q(\gamma,\delta) < \varepsilon \wedge q\in\ob(\delta)\} = \{\delta\in\Sigma\mid \dnonunif(\gamma,\delta) <
\varepsilon\}\subseteq \Sigma_v$.
Letting~$T$ be the smallest integer such that $2^{-\chi_p(C^t)} \leq \varepsilon$ for all $t\geq T$, we get $\Delta_p(C^t) = v$ for all $t\geq T$.

To show Agreement, assume by contradiction that process~$q$ decides a value $w\neq v$ in configuration~$C$ in execution $\gamma\in\Sigma_v$.
Then, by definition of~$\Delta_q$, we have $\gamma\in\{\delta\in\Sigma \mid \exists t\colon C \sim_q D^t \wedge q\in\ob(\delta)\}\subseteq \Sigma_w$.
But this is impossible since $\Sigma_v \cap \Sigma_w = \emptyset$.

We finish the proof by showing Strong Validity.
Let $\gamma \in \Sigma_v$.
Then, by definition, there exists a $p\in\Pi$ and an $I:\Pi\to\V$ with $I_p = v$ such that $\gamma\in g_I[\hat{\Sigma}_p] \subseteq \Sigma^{(I)}$.
But then, in particular, $I_p(\gamma) = I_p = v$.
\end{proof}

We conclude this section by pointing that the practical utility of the
equivalence of consensus with weak and strong validity established
in \cref{thm:equivuniform} and \cref{thm:equivnonuniform} is somewhat
limited: Since the solution algorithms depend on the a priori knowledge of the decision sets, they do not give a clue on how to develop 
a strong consensus algorithm from a weak consensus algorithm in a given
model. In fact, determining and agreeing upon a broadcaster in executions 
that are not $v$-valent is a very hard problem.

\section{Applications}
\label{sec:applications}

In this section, we will apply our topological
characterizations of consensus solvability to several
different examples. Apart from providing a topological explanation
of bivalence proofs (\cref{sec:bivalence}) and folklore results for
synchronous consensus under general omission faults
(\cref{sec:generalomissions}), we will provide a novel characterization of
condition-based asynchronous
consensus~\cite{MRR03:JACM} with strong validity (\cref{sec:conditionconsensus}),
a complete characterization
of consensus solvability for dynamic networks with both closed (\cref{sec:closed})
and non-closed (\cref{sec:nonclosed}) message adversaries, and
a consensus algorithm for asynchronous
systems with weak timely links that does not rely on an
implementation of the $\Omega$ failure detector (\cref{sec:WTL}).

\subsection{Bivalence-based impossibilities}
\label{sec:bivalence}

Our topological results shed some new light on the now standard technique of
bivalence-based impossibility proofs introduced in the celebrated FLP paper~\cite{FLP85},
which have been generalized~\cite{MR02} and used in many different contexts:
Our
results reveal that the forever bivalent executions constructed inductively in 
bivalence proofs \cite{SW89,SWK09,BRSSW18:TCS,WSS19:DC}
are just the common limit of two infinite sequence of executions $\alpha_0,\alpha_1,\dots$
in the $0$-decision set $\Sigma_0$ and  
$\beta_0,\beta_1,\dots$ in the $1$-decision set $\Sigma_1$.

More specifically, what is common to these proofs is that one shows that, for any
consensus algorithm, there is an admissible forever bivalent execution~$\gamma$. This is usually done inductively,
by showing that there is a bivalent initial configuration and that, given a bivalent
configuration $C^{t-1}$ at the end of round $t-1$, there is a 1-round extension
leading to a bivalent configuration $C^t$ at the end of round $t$.
By definition, bivalence of $C^t$ means that there are two admissible executions $\alpha_t$ with 
decision value~$0$ and $\beta_t$ with decision value~$1$ starting out from $C^t$, i.e.,
having a common prefix that leads to $C^t$. Consequently, their distance satisfies $\dnonunif(\alpha_t,\gamma) < 2^{-t}$ and $\dnonunif(\beta_t,\gamma) < 2^{-t}$.
But then closedness of~$\Sigma_0$ and~$\Sigma_1$ implies that $\gamma\in\Sigma_0\cap \Sigma_1$, a contradiction to their disjointness.

By construction, the $(t-1)$-prefixes of $\alpha_t$ and $\alpha_{t-1}$ are the same for
all $t$, which implies that they converge to a limit $\hat{\alpha}$ (and analogously for $\hat{\beta}$), see Figure~\ref{fig:noncompactMA} for an illustration. 
Therefore, these executions match Definition~\ref{def:fairunfair}, and Corollary~\ref{cor:consensusimpfair} implies that the stipulated consensus algorithm cannot be correct.
A specific example is the lossy-link impossibility~\cite{SW89}, i.e., the impossibility of consensus under an oblivious
message adversary for $n=2$ that may choose any graph out of the set $\{\leftarrow,\leftrightarrow,\rightarrow\}$, and the impossibility of solving consensus with vertex-stable source components with insufficient stability interval~\cite{BRSSW18:TCS,WSS19:DC}. 
In the case of the oblivious
lossy-link message adversary using the reduced set $\{\leftarrow,\rightarrow\}$ considered by Coulouma, Godard, and Peters~\cite{CGP15}, consensus is solvable and there is no forever bivalent execution.
Indeed, there exists a consensus algorithm where all configurations reached after the first round are already univalent, see~\cref{sec:closed}.

\subsection{Consensus in synchronous systems with general omission process faults}
\label{sec:generalomissions}

As a more elaborate example of systems where the solvability of non-uniform and uniform
consensus may be different (which also cover the simple running 
examples used in \cref{sec:structure:executions}), we take synchronous systems with up to $f$
general omission process faults~\cite{PT86}. 
For $n\geq f+1$, non-uniform consensus can be solved in 
$f+1$ rounds, whereas solving uniform consensus requires $n \geq 2f+1$. 

The impossibility proof of uniform consensus for $n \leq 2f$ uses a standard
partitioning argument, splitting $\Pi$ into a set~$P$ of processes with $\lvert P\rvert=f$ and~$Q$
with $\lvert Q\rvert =n-f\leq f$. One considers an admissible execution 
$\alpha_0$ where all processes $p\in \Pi$
start with $I_p=0$, the ones in~$P$ are correct, and the ones in~$Q$ are initially
mute; the decision value of the processes in~$P$ must be~$0$ by validity. Similarly,
$\alpha_1$ starts from $I_p=1$, all processes in~$Q$ are correct and the ones in~$P$
are initially mute; the decision value is hence~$1$. For another execution $\alpha$,
where the processes in~$Q$ are correct and the ones in~$P$ are general omission faulty, in the
sense that every $p\in P$ does not send and receive any message to/from~$Q$, one
observes $\alpha\sim_{p}\alpha_0$, i.e., $d_p(\alpha, \alpha_0) < 2^{-t}$ for all $t\geq 0$
and all $p\in P$. Similarly, $\alpha\sim_{q}\alpha_1$ for 
every $q\in Q$. Hence,~$p$ and~$q$ decide on different values in $\alpha$. 

Topologically, this is equivalent to $\dunif(\alpha,\alpha_0)=0$ as well as 
$\dunif(\alpha,\alpha_1)=0$, which implies 
$\alpha\in \Sigma_0$ as well as $\alpha\in \Sigma_1$. Consequently, $\Sigma_0$ and $\Sigma_1$ cannot be
disjoint, as needed for uniform consensus solvability. Clearly, for $n\geq 2f+1$, this argument
is no longer applicable. And indeed, algorithms like the one proposed by Parvedy and Raynal~\cite{PR03:IPDPS} can be 
used for solving uniform consensus.


If one revisits the topological equivalent of the above partitioning argument for $n \leq 2f$
in the \emph{non-uniform} case, it turns out that still $\dnonunif(\alpha,\alpha_0)=0$,
but $\dnonunif(\alpha,\alpha_1)=1$ as all processes in~$Q$ are faulty. Consequently, 
$\alpha \not\in \Sigma_1$. So $\Sigma_0$ and $\Sigma_1$ could partition the space of
admissible executions. And indeed, non-uniform consensus can be solved in $f+1$ rounds 
here. 
In order to demonstrate this by means of our \cref{thm:char:nonunif}, we will sketch how the required decision
sets $\Sigma_v$ can be constructed. We will do so by means of a simple \emph{labeling algorithm}, 
which assigns a decision value $v \in \V$ to every admissible execution~$\gamma$. Note that synchronous
systems are particularly easy to model in our setting, since we can use the number of rounds
as our global time $t$.

Clearly, every process that omits to send its state in some round to a (still) correct processor is 
revealed to every other (still) correct processor at the next round at the latest.
This implies that every correct process $p$ seen by \emph{some} correct process $q$
by the end of the $(f+1)$-round prefix $\gamma|_{f+1}$ in the admissible execution $\gamma$ has also been seen by every other 
correct process during $\gamma|_{f+1}$ as well, since one would need a chain of 
$f+1$ \emph{different} faulty processes for propagating $p$'s state to $q$ otherwise.
Thus, $p$ must have managed to broadcast its initial value
$I_p(\gamma)$ to all correct 
processes during $\gamma|_{f+1}$.

Consequently, if $\gamma|_{f+1} \sim \rho|_{f+1}$, where $\sim$ denotes the transitive
closure (over all processes $p\in \Pi$) of the indistinguishability relation $\sim_p$ for prefixes, they must have the same
set of broadcasters. Our labeling algorithm hence just assigns to $\gamma$
the initial value $I_p$ of the, say, lexically smallest broadcaster $p$ in 
$\gamma|_{f+1}$. The resulting decision sets are trivially open since, for every
$\gamma \in \Sigma_v$, we have $B_{2^{-(f+1)}}(\gamma) \subseteq \Sigma_v$ as well.
The generic non-uniform consensus algorithm from \cref{thm:char:nonunif}
resp.\ \cref{thm:equivuniform} can hence be used for
solving weak resp.\ strong consensus.

\subsection{Asynchronous condition-based consensus}
\label{sec:conditionconsensus}

As an example of asynchronous consensus in shared-memory systems, we consider
the condition-based approach presented by Mostefaoui, Rajsbaum, and Raynal~\cite{MRR03:JACM}. In order to circumvent the FLP impossibility~\cite{FLP85} of 
consensus in the presence of process crashes, the authors considered
restrictions of the vectors of allowed initial values $I(\gamma)$ for
the admissible executions $\gamma \in \Sigma$ of
the $n$ processes in the system. To ensure compatibility with the
notation used in the original paper~\cite{MRR03:JACM}, we will write $I[1],\dots,I[n]$ 
instead of $I_1,\dots,I_n$ for the initial value assignment of
a given admissible execution in this section. For a set $C \subseteq \V^n$ 
of allowed input vectors (called a \emph{condition}) that is a priori known
to all processes, the authors asked for properties $C$
must satisfy such that uniform consensus can be solved in the presence 
of up to $f$ crashes. Note carefully that this is an instance of consensus
where the \emph{arbitrary} input assumption does not apply, albeit the
independent input assumption (recall \cref{def:independentinputassumption}) 
is needed.

Two such properties were identified in~\cite{MRR03:JACM}: (i) the more practical \emph{$f$-acceptability} property, which consists of ``elements'' that can be directly utilized in 
a generic solution algorithm, and
(ii) the more abstract \emph{$f$-legality} condition. Moreover, two different variants of consensus
were considered: (a) non-safe consensus, which only needs to terminate when
the initial values are indeed from $C$, and (b) safe consensus, where the
processes must also terminate for arbitrary inputs in well-behaved (in
particular, fault-free) executions. Interestingly, it turned out
that (i) and (ii), as well as (a) and (b), are equivalent, and
that either variant of consensus can be solved in the presence of up to $f$
crashes if and only if $C$ is $f$-legal or/and $f$-acceptable \cite[Thm.~5.7]{MRR03:JACM} .

The generic non-safe solution algorithm for an $f$-acceptable condition 
$C$ is extremely simple: It only uses one round, where process $p_i$
first writes its initial value $I[i]$ into its entry $V[i]$ of a snapshot 
object $V$ that is initialized to $V[*]=\bot$, and then performs snapshot 
reads that provide its current local view $V_i$ until it finds at least 
$n-f$ non-$\bot$ entries in $V_i$. The latter condition terminates the round, at the end of
which $p_i$ uses the ``elements'' making up $f$-acceptability for 
computing the decision value from its final view $V_i$. Note that a $\bot$ entry 
in $V_i[j]$ can be due to a crash of $p_j$ or just a consequence of the fact that $p_j$ 
has just been slow compared to the at least $n-f$ other processes that managed to provide
non-$\bot$ entries. To make this algorithm compatible with our setting, 
where all executions are infinite, we just add infinitely many empty rounds 
(where no process changes its state or reads/writes $V$). Moreover, we
consider all processes to be obedient and just make at most $f$ of them very
slow when needed, which allows us to directly use our uniform topology.

\medskip

The definition of $f$-legality is based on an undirected 
graph $H(C,f)$, whose vertices are the vectors in $C$ and where there is an
edge $(I1,I2)$ if and only if the Hamming distance between $I1 \in C$ and $I2 \in C$ is
at most $f$. The graph $H(C,f)$ can be expanded into a graph $Gin(C,f)$ of all
the views $V_i$ possibly obtained by any process $p_i$ in the above algorithm: For every
$I\in C$, $Gin(C,f)$ contains all the vertices that are obtained by replacing 
up to $f$ entries of $I$ by $\bot$.  Two vertices $J1,J2 \in Gin(C,f)$ are
connected by an undirected edge if $J1[i]\neq\bot \Rightarrow J1[i]=J2[i]$ 
for every $1\leq i \leq n$, or vice versa. It is not difficult to see
that $I1, I2 \in H(C,f)$ are connected by an edge if and only if the same vertices 
$I1, I2 \in Gin(C,f)$ are connected by a path.

A condition $C$ is $f$-legal if, for each
connected component $G_1,\dots,G_x$ of $Gin(C,f)$, all the vertices in 
the component have at least one input value $v$ in common \cite[Def.~5.2]{MRR03:JACM}. This property translates
to the corresponding connected components $H_1,\dots,H_x$ of $H(C,f)$
as: all vertices in a component must have at least one entry with input value $v$ in
common, and $v$ appears in $f+1$ entries in every vertex. In fact, without
the latter, $v$ would disappear from the view $J$ in $Gin(C,f)$ where 
the at most $f$ entries holding $v$ in $I \in H(C,f)$ are replaced by $\bot$.

The setting for condition-based consensus in \cite{MRR03:JACM} differs from
the one underlying our topological results in the previous sections in two
aspects: (1) It uses a validity condition that is stronger
than our strong validity (SV), as
it does not allow processes to decide on the
initial value of an initially dead process. (2) It does not allow arbitrary
input assignments, which is a pivotal assumption in all our broadcasting-based
characterizations in \cref{sec:broadcastability}. And indeed, as it will
turn out, we do not usually have a \emph{common} broadcaster $p$ in the
connected components of a decision set $\Sigma_v$ here.

In \cref{thm:conditionbasedconsensus} below, we will characterize the solvability of condition-based consensus with strong validity (SV) using our topological approach.
To model (SV), the original $f$-legality condition must be 
weakened to \emph{$f$-quasilegality}: Rather
than assuming that all input assignments $I$ in a connected component $G_i$ in
$Gin(C,f)$, i.e., the vertices also lying in the corresponding connected component
$H_i$ in $H(C,f)$, must have a value $v$ in common that appears in at least
$f+1$ entries in $I$, $f$-quasilegality only requires a common value $v$.

For our proof, 
we exploit the very simple structure of the set of admissible
executions $\Sigma$ of the generic condition-based consensus algorithm, 
and the close relation between $\Sigma$ and $Gin(C,f)$. In fact, $Gin(C,f)$ 
is a graph on all the possible views of the processes (at the end of the 
first round) in any execution. More specifically, for the admissible execution 
$\alpha=\alpha(I) \in \Sigma$ starting from the initial value assignment $I$,
the configuration $\alpha^1=(J_1,¸\dots, J_n)$ after round~1
satisfies $J_j \in G_i \subseteq
Gin(C,f)$ for every $p_j \in \Pi$ and $J_j=\bot$ otherwise. Herein,
$G_i$ is the connected component in $Gin(C,f)$ that contains $I$. This holds
since every $J_j$ is obtained from $I$ by replacing at most $f$ entries with 
$\bot$ in $Gin(C,f)$. Note carefully that every process can hence unambigously 
identify the connected component $G_i$ the current execution belongs to, as 
it only needs to check in which connected component its local view lies. Recall 
that it is assumed that every process knows $C$ and hence $H(C,f)$ and 
$Gin(C,f)$ a priori.

\begin{theorem}[Condition-based consensus characterization]\label{thm:conditionbasedconsensus}
  In the asynchronous shared memory system with at most~$f$ crash faults,
  condition-based consensus with strong validity (SV) can be solved for condition $C$ if and only if $C$ is $f$-quasilegal, in the sense that all the vertices
  in a connected component of $H(C,f)$ have a value $v$ in common.
\end{theorem}
\begin{proof}
We first prove that if $C$ is $f$-quasilegal, then strong consensus is solvable.
With $v_i \in \V$ denoting the common value a priori chosen for the connected 
component $H_i$ (and hence $G_i$), we define the decision sets as $\Sigma_{v_i} = 
\{\gamma | \gamma^1 \in G_i\}$, where $\gamma^1 = D^1$ for $\gamma = (D^t)_{t\geq 0}$. 
By construction, $\Sigma_{v}$ and $\Sigma_{w}$
are disjoint for $w\neq v$. Since our topology is discrete, as the finiteness 
of $Gin(C,f)$ implies that there are only finitely many different admissible executions 
in $\Sigma$, all decision sets (and their connected components) are clopen in $\Sigma$.
Applying the algorithm given in Theorem~\ref{thm:equivuniform}
hence allows to solve consensus. 

On the other hand, to show that consensus cannot be solved if $C$ is not 
$f$-quasilegal, suppose for a contradiction that there is a correct strong 
consensus algorithm without it. We first prove that all executions starting from
an input value assignment $I \in G_i$ in a connected component $G_i \subseteq
Gin(C,f)$, which necessarily also contains all the possible views of all processes 
in $G_i$, lie in the same connected component in $\Sigma$. 
To prove this, it suffices to show by induction that, for any two executions 
$\gamma=\gamma(I)$ and $\delta=\delta(I')$ with
$I, I' \in G_i$, there is a finite sequence of executions $\gamma=\alpha_0,\alpha_1,\dots,\alpha_{k+1}=\delta$ such that, for every $0\leq j < k+1$, $\alpha_j \in G_i$ and $\alpha_j \sim_{q_j} \alpha_{j+1}$ 
for some process $q_j$. This implies $\dunif(\alpha_j,
\alpha_{j+1})=0$ and hence also $\dunif(\gamma,\delta)=0$ as needed.

Since $G_i$ is a connected
component containing $I, I'$, there must be a chain of $k\geq 2$ different initial value assignments $I_0=I_1=I, I_1, \dots, I_{k}=I_{k+1}=I'$ in $G_i$ where $I_\ell$ and $I_{\ell+1}$, $1\leq \ell \leq k-1$, are connected by an edge in $H(C,f)$ (and hence by a path in $G_i$). Moreover, there must be processes $p_1,\dots,p_{k-1}$ such that $I_\ell[p_\ell]\neq I_{\ell+1}[p_\ell]$. For the induction basis $\ell=1$, we choose $\alpha_1=\alpha_1(I_1)$ to be 
any execution where some process $q_0$ has the same view in $\alpha_0^1$ and in $\alpha_1^1$,
and process $q_1$ has the same view $J_1$ in $\alpha_1^1$ and in $\alpha_2^1$, so $\alpha_0 \sim_{q_0} \alpha_1 \sim_{q_1} \alpha_2$. This choice of $\alpha_1$ is possible, since $\alpha_0$ and $\alpha_1$ start from the same $I$, and since $I_1$ and $I_2$
have a Hamming distance between $1$ and $f$ and can hence have a common view $J_1$ with $\bot$ for all processes $q$, including $q_1$, where $I_1[q] \neq I_2[q]$. Note that it is
here where we need the independent (but not arbitrary!) input assignment property \cref{def:independentinputassumption}. For the induction
step, assume that we have already constructed $\alpha_{\ell}$ for $\ell\geq 1$. For $\alpha_{\ell+1}$, we choose an execution where $q_\ell$ has the same view $J_\ell$ in $\alpha_\ell$ and $\alpha_{\ell+1}$ (necessarily with $J_\ell[q_\ell]=\bot$), and $q_{\ell+1}$ has the same view $J_{\ell+1}$ in $\alpha_{\ell+1}$ and $\alpha_{\ell+2}$ (necessarily with $J_{\ell+1}[q_{\ell+1}]=\bot$, unless $\ell+1=k$ already, in which case both $\alpha_{\ell+1}$ and $\alpha_{\ell+2}$ start from $I'$), which leads to $\alpha_\ell \sim_{q_\ell} \alpha_{\ell+1} \sim_{q_{\ell+1}} \alpha_{\ell+2}$ and completes our induction proof.

Since $C$ is not $f$-quasilegal by assumption, there must be a 
connected component $G_i \subseteq Gin(C,f)$ 
that contains initial configurations $I$ and $I'$, such that $I'$ does
not contain any value present in $I$. In order not to violate strong
validity, no executions $\gamma=\gamma(I)$ and $\delta=\delta(I')$ may lie in
the same decision set. However, we have just shown that they lie in the
same connected component in $\Sigma$, which provides the required contradiction.
\end{proof}

\subsection{Dynamic networks with limit-closed message adversaries}
\label{sec:closed}

In this section, we will consider consensus with independent and arbitrary input assignments in dynamic networks under message adversaries \cite{AG13} that are \emph{limit-closed}~\cite{WSM19:OPODIS}, in the sense
that every convergent sequence of executions 
$\alpha_0,\alpha_1,\dots$ with $\alpha_k \in \Sigma$ for every $i$ has a limit $\alpha \in \Sigma$.
An illustration is shown in Figure~\ref{fig:compactMA},
where the purple dots represent a sequence of executions $\alpha_i$ taken from
the connected component $\Sigma_{\gamma_0}$ and $\times$ the limit point 
$\alpha$ at the boundary.
The most prominent examples of limit-closed message adversaries are oblivious ones~\cite{SW89,CGP15,WPRSS23:ITCS}.

We recall that dynamic networks consist of a set of $n$ lock-step
synchronous fault-free processes, which execute a deterministic
consensus algorithm
that broadcasts its entire local state via message-passing in each of the
communication-closed rounds $1,2, \dots$ A message adversary determines
which process $q$ receives the message broadcast by a process $p$ in some
round $t$, via the directed round-$t$ communication graph $\G^t$. Together
with the initial configuration $C^0$ of all the processes, the particular sequence
of communication graphs $\G^1,\G^2, \dots$, which is called communication
pattern, uniquely determines an execution. For example, an oblivious message
adversary is defined by a set~$\mathbf{D}$ of allowed communication graphs
and picks every $\G^t$ arbitrarily from this set.

Since all processes are obedient here, we will only
consider the uniform topology in the sequel.
The set of all process-time graphs $\PTomega$ is compact and the 
transition function $\otau: \PTomega \to \Comega$ is continuous, according 
to \cref{lem:tau:is:cont}, so taking $\otau(\PTomega)$ 
results in a set of configuration sequences that is indeed compact. 
Note that limit-closed message adversaries are hence sometimes refered to 
as \emph{compact} message adversaries.

The following consensus characterization holds even for general message adversaries:

\begin{corollary}[Consensus characterization for general MAs]\label{thm:consensusallMA}
Consensus with independent arbitrary input assignments is solvable under
a general message adversary if and only if (i) all connected components of the set~$\Sigma$ of admissible executions are broadcastable for some process, and (ii) the strong broadcaster decision 
sets $\Sigma_v^p$, $p\in \Pi$, $v\in \V$, given in \cref{eq:Sigmavp}, are closed in $\Sigma$.
\end{corollary}
\begin{proof}
Since there are only finitely many $\Sigma_v^p$, $p\in \Pi$, $v\in \V$, closedness is
equivalent to clopenness here. Hence, \cref{thm:charbroadcastability} can be applied.
\end{proof}

\tikzset{cross/.style={cross out, draw=black, minimum size=2*(#1-\pgflinewidth), inner sep=0pt, outer sep=0pt},
cross/.default={1pt}}

\begin{figure}[ht]
\centering  
\begin{tikzpicture}[scale=0.4]
\def\firstcircle{(0,0) ellipse (3cm and 2cm )}
\def\secondcircle{(0,5) ellipse (3.5cm and 2.5cm)}
\def\thirdcircle{(6,0) ellipse (2.5cm and 1.5cm)}
\def\fourthcircle{(6,3.8) ellipse (2cm and 1.5cm)}
\draw[blue!40!red, ultra thick] \firstcircle node[text=black] {$\Sigma_{\gamma_0}$};
\draw[blue!40!red, ultra thick] \secondcircle node[text=black]  {$\Sigma_{\gamma_0'}$};
\draw[green!50!blue, ultra thick] \thirdcircle node[text=black]  {$\Sigma_{\gamma_1}$};
\draw[green!50!blue, ultra thick] \fourthcircle node[text=black]  {$\Sigma_{\gamma_1'}$};
\begin{scope}[rotate=-20]
\foreach \i in {0,...,5} {
\draw[fill,color=blue!40!red] ({2.2-2*1.3^(-\i)},5.5) circle (0.1);
}
\draw ({2.2-2*1.3^(-5)},5.5) node[cross=4pt,black, ultra thick]{};
\end{scope}
\end{tikzpicture}%
\caption{Examples of two connected components of the decision sets $\Sigma_0=\Sigma_{\gamma_0}\cup \Sigma_{\gamma_0'}$ and
$\Sigma_1=\Sigma_{\gamma_1}\cup \Sigma_{\gamma_1'}$ for consensus under 
a limit-closed message adversary.  
contain all their limit points (marked by $\times$) and have a distance $>0$ by {cor:closeddecsetscompact}.}\label{fig:compactMA}
\end{figure}

We will start our considerations for limit-closed message adversaries
by exploring the structure of the strong decision sets (recall
\cref{def:strong}) of correct consensus algorithms, see \cref{fig:compactMA} for an
illustration.

\begin{corollary}[Strong decision sets for limit-closed MAs]\label{cor:closeddecsetscompact}
For every correct consensus algorithm for a limit-closed message adversary, 
both the strong broadcaster decision sets $\Sigma_v^p$, $p\in \Pi$, $v\in\V$,
and the strong decision sets $\Sigma_v$, $v\in \V$, are disjoint, compact and clopen in $\Sigma$. Moreover, there is some $d>0$ such that $\dunif(\Sigma_v^p,\Sigma_w^q)\geq d > 0$ 
for any $(v,p),(v,p)\neq (w,q) \in \Pi\times \V$, as well as $\dunif(\Sigma_v,\Sigma_w)\geq d > 0$ 
for every $v, v \neq w \in \V$.

In addition, every connected component $\Sigma_{\gamma} \subseteq \Sigma$ is closed and compact,
and for every $\gamma,\delta$ with $\Sigma_{\gamma}\neq \Sigma_{\delta}$, it holds that 
$\dunif(\Sigma_{\gamma},\Sigma_{\delta})>0$. 
\end{corollary}
\begin{proof}
According to \cref{thm:charbroadcastability}, all strong broadcaster decision sets $\Sigma_v^p$ 
are clopen, and hence closed, in $\Sigma$. Since $\Sigma$ is compact for a limit-closed 
message adversary, it follows that every $\Sigma_v^p$ is also compact. \cref{cor:setdistance:min} thus implies $\dunif(\Sigma_v^p,\Sigma_w^q)>0$. Since there are only
finitely many $\Sigma_v^p$, there is hence some $d>0$ that guarantees $\dunif(\Sigma_v^p,\Sigma_w^q)\geq d > 0$ for every $(v,p)\neq (w,q) \in \Pi\times \V$. As $\Sigma_v=\bigcup_{p\in\Pi} \Sigma_v^p$ is
a finite union, the respective results for the strong decision sets follow immediately as well.

Since every connected component $\Sigma_{\gamma}$ of $\Sigma$ that contains $\gamma$ is closed in 
$\Sigma$, as the closure of a connected subspace is also connected~\cite[Lem.~23.4]{Munkres} and a connected component is maximal, $\Sigma_{\gamma}$ is also compact, and 
$\dunif(\Sigma_{\gamma},\Sigma_{\delta})>0$ follows from \cref{cor:setdistance:min}. 
\end{proof}

Unfortunately, \cref{cor:closeddecsetscompact} does \emph{not} allow us to also infer some
minimum distance $d>0$ also for the connected components in general. 
It does hold true, however, if there are only finitely many
connected components. The latter is ensured, in particular, when $\Sigma$ is  
\emph{locally connected}, in the sense that every open set $U(\delta) \subseteq \Sigma$ 
containing $\delta$ also contains some \emph{connected} open set $V(\delta)$: According
to \cite[Thm.~25.3]{Munkres}, all connected components of $\Sigma$ are also open in this 
case. Hence, $\Sigma=\bigcup_{\gamma \in \Sigma} \Sigma_{\gamma}$ is an open covering of $\Sigma$,
and since $\Sigma$ is compact, there is a finite sub-covering $\Sigma=\Sigma_{\gamma^1}' \cup \dots \cup \Sigma_{\gamma^m}'$. Every $\Sigma_{\gamma}$ must hence be equal to one of $\Sigma_{\gamma^1}, \dots, \Sigma_{\gamma^m}$, as connected components are either disjoint or identical.

Unfortunately, however, most limit-closed message adversaries do not guarantee local 
connectedness. In the case of oblivious message adversaries, in particular, it has
been argued~\cite{RSWP23:arxiv} that isolated ``islands'' are created in the evolution 
of the protocol complex, which further develop like the original protocol complex (that
must be not connected for consensus to be solvable).
This phenomenon may be viewed as the result of the ``self-similarity'' that is inherent in the
communication patterns created by such message adversaries, which is not compatible with
local connectedness.

In general, for limit-closed message adversaries that induce infinitely many
connected components in $\Sigma$, one cannot infer openness (and hence 
clopenness) of the connected components: Consider a decision set $\Sigma_v$
that consists of infinitely many connected components. Whereas any connected
component $\Sigma_{\gamma}$ is closed, the set of all remaining connected
components $\Sigma_v\setminus \Sigma_{\gamma}$ need not be closed. It may hence
be possible to pick a sequence of executions $(\alpha_k)$ from 
$\Sigma_v\setminus \Sigma_{\gamma}$ that converges to a limit $\alpha$, and
a sequence $(\beta_k)$ from $\Sigma_{\gamma}$ that converges to $\beta \in 
\Sigma_{\gamma}$, satisfying $\dunif(\alpha,\beta)=0$. By \cref{lem:zerodistancecondition},
this implies $\dunif(\Sigma_{\gamma},\Sigma_v\setminus \Sigma_{\gamma})=0$.
It is important to note, however, that this can only happen for connected components 
in the \emph{same} strong broadcaster decision set $\Sigma_v^p$, as $\dunif(\Sigma_v^p,\Sigma_w^q) \geq d >0$ prohibits a common limit across different decision sets. Consequently, consensus
solvability is not per se impaired by infinitely many connected components, as
\cref{thm:consensusallMA} has shown.

\medskip

We will now make the characterization of \cref{thm:consensusallMA} for limit-closed
message adversaries more operational,
by introducing the $\varepsilon$-approximation of connected components and strong
broadcaster decision sets, typically for some $\varepsilon = 2^{-t}$, 
$t\geq 0$. Informally, it provides the executions that have a $t$-prefix that
cannot be transitively distinguished by some process. Since the number of different 
possible $t$-prefixes is finite, it can be constructed iteratively using finitely 
many iterations:

\begin{definition}[$\varepsilon$-approximations]\label{def:epsilonPSv}
Let $\gamma\in \Sigma$ be an admissible execution. In the minimum topology, 
we iteratively define $\Sigma_{\gamma}^\varepsilon$, for $\varepsilon>0$, as follows:
$\Sigma_{\gamma}^\varepsilon[0]=\{\gamma\}$; for 
$\ell>0$, $\Sigma_{\gamma}^\varepsilon[\ell] = \bigcup_{\alpha \in \Sigma_{\gamma}^\varepsilon[\ell-1]}(B_{\varepsilon}(\alpha) \cap \Sigma)$; and $\Sigma_{\gamma}^\varepsilon=\Sigma_{\gamma}^\varepsilon[m]$ where $m<\infty$ is such
that $\Sigma_{\gamma}^\varepsilon[m]=\Sigma_{\gamma}^\varepsilon[m+1]$. 

For $p\in\Pi$, $v\in\V$, the $\varepsilon$-approximation $\Sigma_v^{p,\varepsilon}$ is defined as $\Sigma_v^{p,\varepsilon}=\bigcup_{\Sigma_\gamma \subseteq \Sigma_v^p}\Sigma_{\gamma}^\varepsilon$.
\end{definition}

Note carefully that $\Sigma_{\gamma}^\varepsilon$ is generally different (in fact, 
larger) than the covering of $\Sigma_{\gamma}$ with $\varepsilon$-balls defined by
$\bigcup_{\delta\in\Sigma_\gamma} B_\varepsilon(\delta) \cap \Sigma$. 
Our $\varepsilon$-approximations satisfy the following properties 
(that actually hold for general message adversaries):

\begin{lemma}[Properties of $\varepsilon$-approximations of connected components]\label{lem:propertiesapproxclosed}
For every $\varepsilon>0$ and every $\gamma, \delta \in \Sigma$, $\varepsilon$-approximations 
have the following properties:
\begin{enumerate}
\item[(i)] $\Sigma_{\gamma}^{\varepsilon'} \subseteq \Sigma_{\gamma}^\varepsilon$ for every $0<\varepsilon'\leq\varepsilon$.
\item[(ii)] $\Sigma_{\gamma}^\varepsilon \cap \Sigma_{\delta}^\varepsilon\neq \emptyset$ implies $\Sigma_{\gamma}^\varepsilon = \Sigma_{\delta}^\varepsilon$.
\item[(iii)] $\Sigma_{\gamma} \subseteq \Sigma_{\gamma}^\varepsilon$.
\end{enumerate}
\end{lemma}
\begin{proof}
To prove (i), it suffices to mention $B_{\varepsilon'}(\alpha) \subseteq B_{\varepsilon}(\alpha)$.
As for (ii), if $\alpha \in \Sigma_{\gamma}^\varepsilon \cap \Sigma_{\delta}^\varepsilon\neq \emptyset$, the iterative construction of $\Sigma_{\gamma}^\varepsilon$ would reach $\alpha$, which would cause it
to also include the whole $\Sigma_{\delta}^\varepsilon$, as the latter also reaches $\alpha$. If (iii) would not hold, $\Sigma_{\gamma}$ could be separated into disjoint open sets, which contradicts its connectivity.
\end{proof}

Obviously, properties (i) and (iii) of the $\varepsilon$-approximation of connected components 
also extend to arbitrary unions of those, and hence to strong broadcaster decision sets. In fact,
for limit-closed message adversaries, provided $\varepsilon$ is chosen sufficiently small,
we get the following result:

\begin{lemma}[$\varepsilon$-approximation of strong broadcaster decision sets]\label{lem:closedepsilonapprox}
For a limit-closed message adversary that allows to solve consensus, there is some 
$\varepsilon>0$ such that, for any $0<\varepsilon'\leq \varepsilon$, it holds that 
$\dunif(\Sigma_v^{p,\varepsilon'}, \Sigma_w^{q,\varepsilon'})>0$ for any $(v,p),(v,p)\neq (w,q) \in \Pi\times \V$. 
\end{lemma}
\begin{proof}
According to Corollary~\ref{cor:closeddecsetscompact}, there is some $d>0$ such that
$\dunif(\Sigma_v^p,\Sigma_w^q)\geq d > 0$. By the extension of 
Lemma~\ref{lem:propertiesapproxclosed}.(iii) to strong broadcaster decision sets,
for any $\varepsilon>0$, $\Sigma_v^p \subseteq \Sigma_v^{p,\varepsilon}$ and $\Sigma_w^q \subseteq \Sigma_w^{q,\varepsilon}$. Therefore, setting $\varepsilon< d/2$ secures $\dunif(\Sigma_v^{p,\varepsilon},\Sigma_w^{q,\varepsilon}) > 0$. By the extension of Lemma~\ref{lem:propertiesapproxclosed}.(i) to strong broadcaster decision sets, we hence also get $\dunif(\Sigma_v^{p,\varepsilon'},\Sigma_w^{q,\varepsilon'}) > 0$. 
\end{proof}

\begin{corollary}[Matching $\varepsilon$-approximation]\label{cor:matchingapproximationclosed}
  For a limit-closed message adversary that allows to solve consensus,
  if $\varepsilon>0$ is chosen in accordance with Lemma~\ref{lem:closedepsilonapprox},
then $\Sigma_v^{p,\varepsilon} = \Sigma_v^p$ for every $p\in\Pi$, $v\in\V$.
\end{corollary}

\begin{theorem}[Operational consensus characterization for limit-closed MAs]\label{thm:closedsufficiency}
A limit-closed message adversary allows to solve consensus if and only if there is some 
$\varepsilon>0$ such that (i) every $\Sigma_{\gamma}^{\varepsilon}$, $\Sigma_\gamma\subseteq\Sigma$, is broadcastable 
for some process, and (ii) every $\Sigma_v^{p,\varepsilon}$, $p\in\Pi$, $v\in\V$, is closed in $\Sigma$.
\end{theorem}
\begin{proof}
Our theorem follows from \cref{thm:consensusallMA} in conjunction with Corollary~\ref{cor:matchingapproximationclosed}.
\end{proof}

\cref{thm:closedsufficiency} implies that if consensus is solvable, then, for every $0<\varepsilon'\leq \varepsilon$, the universal algorithm from \cref{thm:equivuniform} applied to the strong decision sets can be used for actually solving it. And indeed, the consensus algorithm given by Winkler, Schmid, and Moses~\cite[Alg.~1]{WSM19:OPODIS} can be viewed as an instantiation of this fact.

Moreover, \cref{cor:matchingapproximationclosed} implies that 
checking the broadcastability of all the executions in 
$\Sigma_v^{p,\varepsilon}$ can be done by checking the
broadcastability of \emph{finite} prefixes. More specifically, like the decision function $\Delta$ of consensus, the function $T(\alpha)$ that gives the round by which every process in $\alpha\in \Sigma$ has the initial value $I_p(\alpha)$ of the broadcaster $p$ in its view is locally constant for a sufficiently small neighborhood, namely,  $B_{2^{-T(\alpha)}}(\alpha)$, and is hence continuous in any of
our topologies. Since $\Sigma_v^p=\Sigma_v^{p,\varepsilon}$ is compact, $T(\alpha)$ is in fact uniformly continuous and hence attains its maximum $\hat{T}$ in $\Sigma_v^{p,\varepsilon}$. It hence suffices to
check broadcastability in the $t$-prefixes of $\Sigma_v^{p,\varepsilon}$ for $t=\max\{\lfloor \log_2(1/\varepsilon)\rfloor,\hat{T}\}$ in Theorem~\ref{thm:closedsufficiency}.

In~\cite{WSM19:OPODIS}, this has been
translated into the following non-topological formulation
(where $\MA$ corresponds to $\Sigma$, $[\sigma|_r]$ is the set of $r$-prefixes 
of the executions in $\Sigma_\sigma^{2^{-r}}$ in the uniform topology, and 
$\Ker(x)$ is the set of broadcasters in the prefix $x$):

\begin{theorem}[{\cite[Thm.~1]{WSM19:OPODIS}}]
Consensus is solvable under a limit-closed message adversary $\MA$ if and only if
for each $\sigma \in \MA$ there is a round $r$ such that
$\bigcap_{x \in [ \sigma|_r ]} \Ker(x) \neq \emptyset$.
\label{thm:characterization}
\end{theorem}

\subsection{Dynamic networks with non-limit closed message adversaries}
\label{sec:nonclosed}

In this section, we consider consensus with independent and arbitrary input assignments
under message adversaries that are not limit-closed~\cite{FG11,Pfl18:master,WSS19:DC}. A simple example would be a message
adversary, which guarantees that there is some \emph{finite} round $r$ where the
communication graph $\G^r$ is a clique. The communication pattern
where $r=\infty$, i.e., the limiting case $r \to \infty$ (where
the clique graph never happens) is forbidden.

As already mentioned in \cref{sec:closed}, our consensus characterization 
\cref{thm:consensusallMA} also applies here, as does the generic one
in \cref{thm:equivuniform}, of course. Moreover, they can be combined 
with our limit-based characterization \cref{cor:consensusseparation}
and \cref{cor:consensusimpfair}.

What does not work here, however, are our $\varepsilon$-approximations
according to Definition~\ref{def:epsilonPSv}, and everything built on top of it: 
Even if $\varepsilon$ would be made arbitrarily 
small, Lemma~\ref{lem:closedepsilonapprox} does not hold. An illustration is shown in Figure~\ref{fig:noncompactMA}. It is apparent that adding a ball $B_{\varepsilon}(\alpha)$ in the iterative construction of some $\Sigma_{\gamma}^\varepsilon$, where $\dunif(\alpha,\rho)<\varepsilon$ for some forbidden limit sequence $\rho$, inevitably lets the construction grow into some  $\Sigma_{\delta}^\varepsilon$ lying in a different strong broadcaster decision set. Whereas this could be avoided by adapting $\varepsilon$ when coming close to $r$, the resulting approximation would not provide any advantage over directly using our characterization \cref{thm:consensusallMA}.

\begin{figure}[ht]
\centering  
\begin{tikzpicture}[scale=0.8]
\pgfmathsetseed{2}
\draw [blue!40!red, dashed, ultra thick, shift={(0,0)}] 
plot [smooth cycle, tension=1, domain=0:320, samples=18] (\x:{4/3+rand/3}) 
node[left=0cm,above=0.2cm,text=black] {$\Sigma_{\gamma_{0}}\qquad$} ;
\draw [green!50!blue, dashed, ultra thick, shift={(3.1,0)}] 
plot [smooth cycle, tension=1, domain=0:320, samples=18] (\x:{4/3+rand/3}) 
node[left=0cm,above=0.2cm,text=black] {$\Sigma_{\gamma_{1}}\qquad$} ;
\draw [blue!40!red, dashed, ultra thick, shift={(0,3.1)}] 
plot [smooth cycle, tension=1, domain=0:320, samples=18] (\x:{4/3+rand/3}) 
node[left=0cm,above=0.2cm,text=black] {$\Sigma_{\gamma_{0}'}\qquad$} ;
\draw [green!50!blue, dashed, ultra thick, shift={(3.1,3.1)}] 
plot [smooth cycle, tension=1, domain=0:320, samples=18] (\x:{4/3+rand/3}) 
node[left=0cm,above=0.2cm,text=black] {$\Sigma_{\gamma_{1}'}\qquad$} ;
\foreach \i in {0,...,5} {
\draw[fill,color=green!50!blue] ({1.39+1.5^(-\i)},0.3) circle (0.05);
\draw[fill,color=blue!40!red] ({1.39+1.5^(-5)-1.5^(-\i)},0.3) circle (0.05);
}
\draw ({1.36+1.5^(-5)},0.3) node[cross=4pt,black, ultra thick]{};
\end{tikzpicture}%
\caption{Examples of two connected components of the decision sets $\Sigma_0=\Sigma_{\gamma_0}\cup \Sigma_{\gamma_0'}$ and
$\Sigma_1=\Sigma_{\gamma_1}\cup \Sigma_{\gamma_1'}$ for a non-compact message adversary. They are not closed in $\Comega$ and may have
distance 0; common limit points (like for $\Sigma_{\gamma_0}$ and $\Sigma_{\gamma_1}$, marked by $\times$) must hence be excluded by
Corollary~\ref{cor:consensusseparation}.}
\label{fig:noncompactMA}
\end{figure}

These topological findings are of course in accordance with the results
on non-limit closed message adversaries we are aware of.
In particular, the binary consensus algorithm for $n=2$ by Fevat and Godard~\cite{FG11} assumes that the algorithm 
knows a fair execution or a pair of  unfair executions according to 
\cref{def:fairunfair}
a priori, which effectively partition the execution space into two connected 
components.\footnote{Note that there are uncountably many choices for separating $\Sigma_0$ and $\Sigma_1$ here, however.} Such a limit exclusion is also exploited in the counterexample
to consensus task solvability for $n=2$ via a decision map that is
not continuous~\cite{GKM14:PODC}, which has been suggested by Godard and Perdereau \cite{GP20:MSCS}: It excludes just the unfair execution $\alpha$
based on $\{\leftrightarrow,\leftarrow,\leftarrow,\dots\}$, but not the unfair
execution $\beta$ caused by $\{\rightarrow,\leftarrow,\leftarrow,\dots\}$, which satisfies $d_p(\alpha,\beta)=0$ for the right process $p$ and hence makes consensus impossible.

The $(D+1)$-VSRC message adversary $\BStable(D+1)$ \cite{WSS19:DC}
generates executions that are based on single-rooted communication 
graphs in every round, with the additional guarantee that, eventually,
a \emph{$D+1$-vertex-stable root component} ($D+1$-VSRC) occurs.
Herein, a root component is a strongly connected component without in-edges from
outside the component, and a $x$-VSRC is a root component made up
of the same \emph{set} of processes in $x$
consecutive rounds. $D\leq n-1$ is the dynamic diameter of a VSRC,
which guarantees that all root members reach all processes. 
It has been proved~\cite{WSS19:DC}
that consensus is impossible with $\BStable(x)$ for $x\leq D$, whereas
an algorithm exists for $\BStable(D+1)$.
Obviously, $\BStable(D+1)$ effectively excludes all communication patterns
without any 
$D+1$-VSRC. And indeed, the choice $x=D+1$ renders the connected components of $\Sigma$
broadcastable by definition, which is in accordance with \cref{thm:consensusallMA}. 

We also introduced and proved correct an explicit
labeling algorithm for $\BStable(n)$ in \cite{WSN21:FCT}, which effectively operationalizes
the universal consensus algorithm of \cref{thm:equivuniform}: By assigning
a (persistent) \emph{label} $\myFunc(\sigma'|_r)$ to the $r$-prefixes of 
$\sigma \in \Sigma$, it effectively assigns a corresponding unique decision value
$v \in \V$ to $\sigma$, which in turn specifies the strong decision set $\Sigma_v$ containing $\sigma$.
In the language of \cite{WSN21:FCT}, the requirement of every $\Sigma_v$
being open (and closed) in \cref{thm:equivuniform} translates into
a matching assumption on this labeling function as follows
(herein, $\MA$ corresponds to $\Sigma$, $\sigma|_r$ denotes the
$r$-round prefix of execution $\sigma$, and $\sim$ is the transitive
closure over all processes $p$ of the prefix indistinguishability relation $\sim_p$):

\begin{assumption}[{\cite[Assumpt.~1]{WSN21:FCT}}]
\label{ass:indistDelta}
$\forall \sigma \in \MA \enspace \exists r \in \mathbb{N} \enspace \forall \sigma'\in \MA \colon \sigma'|_r \sim \sigma|_r \Rightarrow \myFunc(\sigma'|_r) = \myFunc(\sigma|_r) \ne \emptyset \enspace.$
\end{assumption}

For $\BStable(n)$, it has been proved~\cite[Thm.~12]{WSN21:FCT} that
the given labeling algorithm satisfies this assumption for $r=r_{stab}+4n$,
where $r_{stab}$ is the round where the (first) $D+1$-VSRC in $\sigma$
starts. Consensus is hence solvable by a suitable instantiation of the universal 
consensus algorithm of \cref{thm:equivuniform}.

\subsection{Consensus in systems with an eventually timely $f$-source}
\label{sec:WTL}

It is well-known~\cite{DDS87} that consensus cannot be solved in
distributed systems of $n\geq 2f+1$ (partially) synchronous processes, 
up to which $f$ may crash, which are connected by reliable 
\emph{asynchronous} communication links. For solving consensus, the system model 
has been strengthened by a \emph{weak timely link} (WTL) assumption~\cite{ADGFT04,HMSZ08:TDSC}: there has to be at least one correct process 
$p$ that eventually sends timely to a sufficiently large subset of the processes.

In previous work~\cite{ADGFT04}, at least one \emph{eventually timely $f$-source} $p$ was assumed: After some 
unknown initial period where all end-to-end message delays are arbitrary, every broadcast of $p$
is received by a fixed subset $P\subseteq \Pi$ with $p\in P,|P|\geq f+1$ within some possibly unknown
maximum end-to-end delay $\Theta$. The authors showed how to build the $\Omega$ failure detector
in such a system, which, in conjunction with any $\Omega$-based consensus algorithm (like the one by Most{\'e}faoui and Raynal~\cite{MR01}) can be used
to solve uniform consensus. 

Their $\Omega$ implementation lets every process broadcast a \emph{heartbeat message} every
$\eta$ steps, which forms partially synchronized rounds, and maintains an \emph{accusation
counter} for every process $q$ that counts the number of rounds the heartbeats
of which were not received timely by more than $f$ processes. This is done by letting every
process who does not receive $q$'s broadcast within $\Theta$ send an \emph{accusation message}
for $q$, and incrementing the accusation counter for $q$ if more than $f$ such accusation
messages from different receivers came in. It is not difficult to see that the accusation
counter of a process that crashes grows unboundedly, whereas the accusation counter of
every timely $f$-source eventually stops being incremented. Since the accusation counters
of all processes are exchanged and agreed-upon as well, choosing the process with
the smallest accusation counter (with ties broken by process ids) is a legitimate choice
for the output of $\Omega$. 

This WTL model was further relaxed~\cite{HMSZ08:TDSC}, which 
allows the set $P(k)$ of witnessing receivers of every \emph{eventually moving timely $f$-source} to 
depend on the sending round $k$. The price to be paid for this relaxation is the need to incorporate
the sender's round number in the heartbeat and accusation messages.

In this subsection, we will use our \cref{thm:equivuniform} to prove topologically that 
consensus with strong validity and independent arbitrary input assignments 
can indeed be solved in the WTL model: We will give and prove 
correct an explicit labeling algorithm \cref{alg:lambda}, which assigns a decision value 
$v \in \V$ to every execution $\sigma$ that specifies the decision set $\Sigma_v$ 
containing $\sigma$. Applying our universal algorithm to these decision sets hence allows 
to solve consensus in this model. Obviously, unlike the existing algorithms, our algorithm does 
not rely on an implementation of $\Omega$. 

\medskip

We assume a (slightly simplified)
WTL model with synchronous processes and asynchronous links that are reliable and FIFO, 
with known $\Theta$ for timely links. Whereas we will use the time $t=0, 1, 2, \dots$ our
synchronous processes take their steps as global time, we note that we do not have 
communication-closed rounds here, i.e., have to deal with general executions
according to \cref{def:PTGs} in the appendix.
In an admissible execution $\sigma$, we denote by $F(\sigma)$ the set of up to $f$ processes
that crash in $\sigma$, and $C(\sigma)=\ob(\sigma)=\Pi\setminus F(\sigma)$ the set of correct
processes. For an eventual timely $f$-source $p$, we will denote with
$\rST{p}$ the \emph{stabilization round}, by which it has already started to send timely: a message sent
in round $t \geq \rST{p}$ is received by every $q\in P(t)$ no later than in round 
$t+\Theta-1$, hence is present in $q$'s state at time $t+\Theta-1$. 
Note carefully that this condition is automatically satisfied when
$q$ has crashed by that round. 
We again assume that the processes execute a full-information
protocol, i.e., send their whole state in every round. For keeping the relation
to the existing algorithms, we consider the state message sent by $p$ in
round $t$ to be its $\heartbeat(t)$. Moreover, if the state of
process $q$ at time $t+\Theta-1$ does not contain the reception $\heartbeat(t)$ from process 
$p$, we will say that $q$ broadcasts an \emph{accusation message} $\accusation(p,t)$ for round $t$ of $p$ 
in round $t+\Theta$ (which is of course just part of $q$'s state sent in this round).
If $q$ crashes before round $t+\Theta$, it will never broadcast $\accusation(p,t)$. If $q$ crashes
exactly in round $t+\Theta$, we can nevertheless assume that it either manages to eventually communicate
$\accusation(p,t)$ to all correct processes in the system, or to none: In our full information
protocol, every process that receives $\accusation(p,t)$ will forward also this message
to all other processes when it broadcasts its own state later on.

\begin{definition}[WTL elementary state predicates and variables]\label{def:WTLpredicates}
For process $s$ at time $r\geq 1$, i.e., the end of round $r$, we define the following predicates and state
variables:
\begin{itemize}
\item $\accuse_s^r(p)=\true$ if and only if $s$ did not receive $\heartbeat(r-\Theta)$ from $p$ by time $r$ and thus sent $\accusation(p,t)$.
\item $\nottimelyrec_s^r(q,p,t) = \true$ if and only if $s$ recorded the reception of $\accusation(p,t)$ from $q$ by time $r$.
\item $\nottimely_s^r(p,t) = \true$ if and only if $\nottimelyrec_s^r(q,p,t)=\true$ for at least $n-f$ different $q\in\Pi$.
\item $\accusationcounter_s^r(p) = (|\{ k\leq r: \nottimely_s^r(p,k)=\true \}|,p)$.
\item $\hearedof_s^r(p) = |\{k\leq r : \mbox{$s$ received $\heartbeat(k)$ from $p$ (directly or indirectly) by time $r$}\}|$.
\end{itemize}
\end{definition}
Note that a process $q$ that crashes before time $t+\Theta$ causes $\nottimelyrec_s^r(q,p,t) = \false$ for all 
$r$, and that $p$ is appended in $\accusationcounter_s^r(p)$ for tie-breaking purposes only. 
For every eventually timely $f$-source $p$, the implicit forwarding of accusation messages ensures
that $\accusationcounter_s^r(p)$ will eventually be the same at every correct process $s$ in the limit 
$r\to\infty$.

We now define some predicates that require knowledge of the execution $\sigma$. 
Whereas they cannot be computed locally by the processes in the execution, they
can be used in the labeling algorithm.

\begin{definition}[WTL extended state predicates and variables]\label{def:WTLextendedpredicates}
Given an execution $\sigma$, let the \emph{dominant} eventual
timely $f$-source $p_\sigma$ be the one that leads to the unique smallest value of 
$\accusationcounter_s^\infty(p_\sigma)$, which is the same at every process $s\in\Pi\setminus{F(\sigma)}$.
With $\rST{\sigma}=\rST{p_{\sigma}}$ denoting the stabilization time of the dominant eventual
timely $f$-source in $\sigma$ and $F(\sigma|_r)\subseteq F(\sigma)$ the set of processes that
crashed by time $r$, we also define
\begin{itemize}
\item $\minhearedof_s(\sigma,r) = \min_{p\in\Pi\setminus{F(\sigma|_r)}} \hearedof_s^r(p)$,

\item $\oldenough(\sigma,r) = \true$ if and only if $\forall s \in \Pi\setminus{F(\sigma|_r)}$, both
(i) $\minhearedof_s(\sigma,r) \geq \rST{\sigma}+\Theta$ and (ii) $\forall p\in\Pi\setminus{p_\sigma}: \; \accusationcounter_s^r(p_\sigma) < \accusationcounter_s^r(p)$.

\item $\mature(\sigma,r) = \true$ if and only if $\exists r_0 < r$ such that both
(i) $\oldenough(\sigma,r_0) = \true$ and
(ii) $\forall s \in \Pi\setminus{F(\sigma|_r)} : \; \minhearedof_s(\sigma,r) \geq r_0$.

\end{itemize}
\end{definition}
Note that it may occur that another eventual timely $f$-source $p'\neq p_\sigma$ in $\sigma$ has
a smaller stabilization time $\rST{p'}<\rST{p_{\sigma}}$ than the dominant one, which
happens if $p'$ causes more accusations than $p_\sigma$ before stabilization in total.

The following properties are almost immediate from the definitions:

\begin{lemma}[Properties of $\oldenough$ and $\mature$]\label{lem:propertiesoldenough}
The following properties hold for $\oldenough$:
\begin{enumerate}
\item[(i)] If $\oldenough(\sigma,r)=\true$, then 
$\accusationcounter_s^t(p_\sigma)=\accusationcounter_s^r(p_\sigma)$ for every $s$ that did
not crash by time $t\geq r$.
\item[(ii)] $\oldenough(\sigma,r)$ is stable, i.e., $\oldenough(\sigma,r)=\true \Rightarrow 
\oldenough(\sigma,t)=\true$ for $t\geq r$.
\item[(iii)] (i) and (ii) also hold for $\mature(\sigma,r)$, and $\mature(\sigma,r)=\true 
\Rightarrow \oldenough(\sigma,r)=\true$.
\end{enumerate}
\end{lemma}
\begin{proof}
Since $\oldenough(\sigma,r)=\true$ entails that every process 
$s \in \Pi\setminus{F(\sigma|_r)}$ has received the accusation messages for all rounds
up to $\rST{\sigma}$ since  $\minhearedof_s(\sigma,r) \geq \rST{\sigma}+\Theta$ according to
\cref{def:WTLextendedpredicates}, (i) follows. This also implies (ii), since the
accusation counter of every process $p\neq p_\sigma$ can at most increase after
time $r$. That these properties carry over to $\mature$ is obvious from the definition.
\end{proof}

The following lemma proves that two executions $\sigma$ and $\rho$
with indistinguishable prefixes $\sigma|_r \sim_s \rho|_r$, i.e., 
$(\sigma|_r)^t \sim_s (\rho|_r)^t$ for $0 \leq t \leq r$,
cannot both satisfy $\oldenough(\sigma,r)$ resp.\ $\oldenough(\rho,r)$, 
and, hence, $\mature(\sigma,r)$ resp.\ $\mature(\rho,r)$, except when the dominant eventual 
timely $f$-source is the same in $\sigma$ and $\rho$:

\begin{lemma}\label{lem:notbotholdenough}
Consider two executions $\sigma$ and $\rho$ with $\sigma|_r \sim_s \rho|_r$
for some process $s$ that is not faulty by round $r$ in both $\sigma$ and $\rho$. Then, 
\[
(\oldenough(\sigma,r)= \true \wedge \oldenough(\rho,r)=\true) \Rightarrow p_\sigma = p_\rho.
\]
\end{lemma}
\begin{proof}
As $\oldenough(\sigma,r)= \true$, \cref{def:WTLextendedpredicates} implies
$\forall p\in\Pi\setminus{p_\sigma}: \accusationcounter_s^r(p_\sigma) <
\accusationcounter_s^r(p)$, and similarly
$\forall p\in\Pi\setminus{p_\rho}: \accusationcounter_s^r(p_\rho) <
\accusationcounter_s^r(p)$. Since $\sigma|_r \sim_s \rho|_r$, this is only possible if $p_\sigma = p_\rho$.
\end{proof}

Finally, we need the following technical lemmas:

\begin{lemma}[Indistinguishability precondition]\label{lem:indistprecond}
Suppose $\tau|_{r'} \sim_{s'} \sigma|_{r'}$ is such that $s'$ received a message from $s\neq s'$ 
containing its state in the sending round $r_0' \leq r'-1$ by round $r'$ in $\sigma|_{r'}$ and
hence also in $\tau|_{r'}$. Analogously, suppose $\sigma|_{r} \sim_{s} \rho|_{r}$ is such that 
$s$ received a message from $s'$ containing its state in the sending round $r_0 \leq r-1$ 
by round $r$ in $\sigma|_{r}$ and hence also in $\rho|_{r}$. Then,
\begin{enumerate}
\item[(i)] $\tau|_{r_0'} \sim_s \sigma|_{r_0'}$,
\item[(ii)] $\tau|_{\min\{r_0',r\}} \sim_s \rho|_{\min\{r_0',r\}}$,
\item[(iii)] $\sigma|_{r_0} \sim_{s'} \rho|_{r_0}$,
\item[(iv)] $\tau|_{\min\{r_0,r'\}} \sim_{s'} \rho|_{\min\{r_0,r'\}}$.
\end{enumerate}
\end{lemma}
\begin{proof}
If (i) would not hold, since $s$ sends a message containing its state in round $r_0'$ to
$s'$ both in $\tau|_{r'}$ and in $\sigma|_{r'}$, these two states would be distinguishable
for $s$, which contradicts our assumption. The analogous argument proves (iii). Statement (ii)
follows from combining (i) with $\sigma|_{r} \sim_{s} \rho|_{r}$, (iv) follows
from combining (iii) with $\tau|_{r'} \sim_{s'} \sigma|_{r'}$.
\end{proof}

\begin{lemma}[Heardof inheritance]\label{lem:HOinheritance}
Suppose $\sigma|_r \sim_s \rho|_r$ and $\minhearedof_s(\rho,r)\geq r_0$ for some $1\leq r_0 < r$, 
as it arises in $\mature(\rho,r)=\true$, for example. Then, $\forall p \in \Pi\setminus{F(\rho|_r)}$,
it also holds in $\sigma|_r$ that $\hearedof_s^r(p)\geq r_0$, but not necessarily
$\hearedof_{s}^r(p')\geq r_0$ for $p' \in (\Pi\setminus{F(\sigma|_r)}) \cap F(\rho|_r)$. Consequently,
it may happen that $\minhearedof_s(\sigma,r) < r_0$.
\end{lemma}
\begin{proof}
Since the state of $s$ is the same in $\sigma|_r$ and $\rho|_r$, but the sets $F(\rho|_r)$ and
$F(\sigma|_r)$ may be different, the lemma follows trivially.
\end{proof}

With the abbreviation $C(\sigma|_r)=\Pi\setminus{F(\sigma|_r)}$ for all non-faulty processes
in $\sigma|_r$, and $\sigma|_r \sim_Q \rho|_r$ for $\forall q\in Q:\; \sigma|_r \sim_q \rho|_r$, 
we define the short-hand notation $\sigma|_r \simmaj \rho|_r$ to express indistinguishability
for a majority of (correct) processes, defined by $\exists Q \subseteq C(\sigma|_r)\cap C(\rho|_r), 
|Q|\geq n-f$ such that $\forall q\in Q: \; \sigma|_r \sim_{q} \rho|_r$.


The following lemma guarantees that prefixes that are indistinguishable only for
strictly less than $n-f$ processes are eventually distinguishable for all processes:

\begin{lemma}[Vanishing minority indistinguishability]\label{lem:vanishingminority}
Given $\rho|_{r_0}$, there is a round $r$, $r_0\leq r <\infty$, such that for every 
$\sigma|_{r_0}$ with $\rho|_{r_0} \not\simmaj \sigma|_{r_0}$, it holds that $\rho|_r 
\not\sim \sigma|_r$.
\end{lemma}

\begin{proof}
Due to our reliable link assumption, for every process $s$ that does not fail in $\rho$, there is a 
round $r > r_0$ where $\minhearedof_s(\rho,r)\geq r_0$. Now assume that
there is some $\sigma|_{r_0}$ with $\rho|_{r_0} \sim_Q \sigma|_{r_0}$ for a maximal set $Q$ with
$1\leq |Q| < n-f$, but $\rho|_r \sim_s \sigma|_r$ for some process $s$. Since $s$ receives round-$r_0$ messages from 
$|C(\rho|_n)|\geq n-f$ processes in $\rho|_r$, and $\rho|_r \sim_s \sigma|_r$, process $s$ must receive 
exactly the same messages also in $\sigma|_r$. As at most $|Q| < n-f$ of those messages
may be sent by processes that cannot distinguish $\rho|_{r_0} \sim_Q \sigma|_{r_0}$, 
at least one such message must originate in a process $q'$ with $\rho|_{r_0} \not\sim_{q'} \sigma|_{r_0}$.
In this case, \cref{lem:indistprecond}.(iii) prohibits $\rho|_r \sim_s \sigma|_r$, however, which provides the required contradiction.
\end{proof}

The following lemma finally shows that majority indistinguishability in conjunction
with mature prefixes entails strong indistinguishability properties in earlier
rounds:

\begin{lemma}[Majority indistinguishability precondition]\label{lem:majindistprecond}
Suppose $\tau|_{r} \simmaj \sigma|_{r} \simmaj \rho|_r$ and $\mature(\rho,r)=\true$. Then,
for the round $r_0$ imposed by the latter, it holds that $\tau|_{r_0} \sim_{C(\rho|_r)} 
\sigma|_{r_0} \sim_{C(\rho|_r)} \rho|_{r_0}$, and hence also $\tau|_{r_0} \sim_{C(\rho|_r)} \rho|_{r_0}$.
\end{lemma}
\begin{proof}
Let $S$ resp.\ $Q$ be the set of at least $n-f$ processes causing 
$\sigma|_{r}\simmaj \rho|_{r}$ resp.\ $\sigma|_{r}\simmaj \tau|_{r}$. 
Since $Q \cap S \neq \emptyset$ by the pigeonhole principle, let 
$s\in Q \cap S$. Clearly, $\tau|_{r} \sim_s \sigma|_{r} \sim_s \rho|_{r}$, 
and hence also $\tau|_{r} \sim_s \rho|_{r}$.
Since $\mature(\rho,r)=\true$, \cref{lem:indistprecond}.(i)
in conjunction with \cref{lem:HOinheritance} implies $\rho|_{r_0} \sim_{C(\rho|_{r})} 
\sigma|_{r_0}$, as well as $\rho|_{r_0} \sim_{C(\rho|_{r})} 
\tau|_{r_0}$, and hence also $\sigma|_{r_0} \sim_{C(\rho|_{r})} 
\tau|_{r_0}$ as asserted.
\end{proof}

With these preparations, we can define an explicit labeling algorithm \cref{alg:lambda}
for the WTL model, i.e., an algorithm that computes a label $\myFunc(\sigma|_r)$ 
for every $r$-prefix $\sigma|_r$ of an admissible execution $\sigma$ in our WTL model.
A label can either be $\emptyset$ (still undefined) or else denote a single process 
$p$ (which will turn out to be a broadcaster), and will be persistent in $\sigma$ in the sense that $\myFunc(\sigma|_r)=p \Rightarrow
\myFunc(\sigma|_{r+k})=p$ for every $k\geq 0$. Note that we can hence uniquely also
assign a label $\myFunc(\sigma)$ to an infinite execution. Note that, for defining our
decision sets, we will assign $\sigma$ to $\Sigma_{I_p}$, where $I_p$ is the initial
value of $p=\myFunc(\sigma)$ in $\sigma$.

Informally, our labeling algorithm 
works as follows: If there is some unlabeled 
mature prefix $\rho|_r$, it is labeled either (i) with the label of some already
labeled but not yet mature $\sigma|_r$ if the latter got its label early enough,
namely, by the round $r_0$ where $\oldenough(\rho,r_0)=\true$, or else (ii) with 
its dominant $p_\rho$.

\begin{algorithm}
\caption{Computing $\myFunc$ for each $r$-prefix $\sigma|_r$ in the WTL model.}\label{alg:lambda}
\DontPrintSemicolon
	Initially, let $\myFunc(\sigma|_0) = \emptyset$. \label{line:initWTL} \;
\For{$r = 1, 2, \ldots$}{
\lForEach{$\sigma|_r$}{
$\myFunc(\sigma|_r) \gets \myFunc(\sigma|_{r-1})$} \label{line:monotonyWTL}
\ForEach{$\rho|_r$ with $\myFunc(\rho|_r) = \emptyset$}{
\If{$\exists \sigma|_r \simmaj \rho|_r$ with $\myFunc(\sigma|_r) = p \ne \emptyset$ and
  $\mature(\sigma,r)=\true$}{
  $\myFunc(\rho|_r) \gets p$ \label{line:assignOtherWTL} \;
}
}
\ForEach{$\rho|_r$ with $\myFunc(\rho|_r) = \emptyset$ and $\mature(\rho,r)=\true$}{
 \If{$\exists \sigma|_r \simmaj \rho|_r$ with $\myFunc(\sigma|_{r_0}) = p \ne \emptyset$ for $r_0$ satisfying $\oldenough(\rho,r_0)=\true$}{
  $\myFunc(\rho|_r) \gets p$ \label{line:assignOwnfromOldWTL} \;
}
}
\ForEach{$\rho|_r$ with $\myFunc(\rho|_r) = \emptyset$ and $\mature(\rho,r)=\true$}{
\eIf{$\exists \sigma|_r \simmaj \rho|_r$ with $\myFunc(\sigma|_r) = p \ne \emptyset$ and
  $\mature(\sigma,r)=\true$}{
	$\myFunc(\rho|_r) \gets p$    // Only happens when $\sigma|_r$ got its label in line~\ref{line:assignOwnfromOldWTL}\label{line:assignOtherAfterWTL}\;
  }{
	$\myFunc(\rho|_r) \gets p_\rho$ \label{line:assignOwnWTL}\;
  }
}
}
\end{algorithm}

The following \cref{thm:decsetWTL} in conjunction with \cref{lem:DeltabroadcastWTL}shows that \algref{alg:lambda} computes
labels, which result in strong decision sets that are compatible with the needs of
\cref{thm:equivuniform}. Strong consensus in the WTL model can hence be solved by
means of our universal algorithm.

\begin{theorem}[Strong decision sets for WTL algorithm]\label{thm:decsetWTL}
The set $\Sigma(p) = \{ \sigma \mid \myFunc(\sigma) = p \}$ is open in
the uniform topology, and so is the strong decision set 
$\Sigma_v = \{ \sigma \mid (\myFunc(\sigma) = p) \wedge (I_p = v)\}$.
\end{theorem}

\begin{proof}
We show that, if $\sigma$ is assigned to the partition set
$\Sigma(p)$, then $B_{2^{-(i+D(\sigma))}}(\sigma) \subseteq \Sigma(p)$, 
where $i$ is the smallest round where $\mature(\sigma,i)=\true$
and $D(\sigma)$ is the maximum number of rounds required for
a minority indistinguishability in $\sigma_i$ to go away 
($D(\sigma)=r-r_0$ in the notation of \lemmaref{lem:vanishingminority}),
which implies openness of $\Sigma(p)$.
Note that the corresponding property obviously also holds for the decision set 
$\Sigma_v = \{ \sigma \mid (\myFunc(\sigma) = p) \wedge (I_p = v)\}$.

First of all, in \cref{alg:lambda}, $\myFunc(\sigma|_i)$ gets initialized to $\emptyset$
in line~\ref{line:initWTL} and assigns a label $\ne \emptyset$ at the latest
when $\mature(\sigma,i)=\true$.
Once assigned, this value is never modified again as each assignment,
except the one in line~\ref{line:monotonyWTL}, may only be performed
if the label was still $\emptyset$.

For an unlabeled prefix $\sigma|_i$ that is indistinguishable
to a mature labeled prefix $\rho|_i$, there are two possibilities: Either, its
indistinguishability is a majority one, in which case $\sigma|_i$ gets its label from
$\rho|_i$ in line~\ref{line:assignOtherWTL}, or else the minority indistinguishability
will go away within $D(\sigma)$ rounds.
It thus suffices to show that if a label $\myFunc(\rho|_r) \gets \{ p \}$ is
assigned to a round $r$ prefix $\rho|_r$, then every majority-indistinguishable 
prefix $\sigma|_r \simmaj \rho|_r$ has
either $\myFunc(\rho|_r) = \myFunc(\sigma|_r)$ or $\myFunc(\sigma|_r) = \emptyset$.

We prove  this by induction on $r = 0, 1, \ldots$.
The base for $r=0$ follows directly from line~\ref{line:initWTL}.
For the step from $r-1$ to $r$, assume by hypothesis that, for all round $r-1$
prefixes that already had $\{ p \}$ assigned, all their majority-indistinguishable prefixes
have label $\{ p \}$ or $\emptyset$.
For the purpose of deriving a contradiction, suppose that a label
$\myFunc(\rho|_{r}) \ne \emptyset$ is
assigned to a round $r$-prefix $\rho|_{r}$ in iteration $r$ and there exists some $\sigma|_r$ with
$\sigma|_{r} \simmaj \rho|_{r}$ and $\emptyset \ne \myFunc(\sigma|_{r}) \ne \myFunc(\rho|_{r})$. Let $S$
be the set of involved processes, i.e., $\sigma|_r \sim_s \rho|_r$ for $s\in S$ with $|S|\geq n-f$.

We need to distinguish all the different ways of assigning labels to $\rho|_r$.

Suppose $\sigma|_r$ nor $\rho|_r$ get their labels in round $r$, but not in line~\ref{line:assignOtherWTL}. 
Since both $\mature(\sigma,r)=\true$ and $\mature(\rho,r)=\true$,
\cref{lem:propertiesoldenough}.(iii) in conjunction with \cref{lem:notbotholdenough} reveals
that $p_\sigma=p_\rho$ since $\sigma|_{r} \simmaj \rho|_{r}$. In all cases
except for the one where both $\rho|_r$ and $\sigma|_r$ get their labels
in line~\ref{line:assignOwnfromOldWTL}, we immediately get a contradiction 
since $\myFunc(\rho|_r)=\myFunc(\sigma|_r)$ in any case. Finally,
if $\rho|_r$ and $\sigma|_r$ get their labels in line~\ref{line:assignOwnfromOldWTL}, there is some $\tau|_r \simmaj \rho|_{r}$ with $\mature(\tau,r)=\false$ but $\myFunc(\tau|_{r_0})\neq \emptyset$, where $r_0$ is such that $\oldenough(\rho,r_0)=\true$, and some $\omega_r \simmaj \sigma|_{r}$ with 
the analogous properties in round $r_0'$. Let $Q'$ resp.\ $Q''$ be the sets of at least $n-f$ processes involved in 
$\tau|_r \simmaj \rho|_{r}$ resp.\ $\omega_r \simmaj \sigma|_{r}$.
Since $\mature(\rho,r)=\true$, \cref{lem:majindistprecond} implies $\rho|_{r_0} \sim_{C(\rho|_{r})} \sigma|_{r_0}
\sim_{C(\rho|_{r})} \tau|_{r_0}$ and also $\rho|_{r_0'} \sim_{C(\rho|_{r})} \sigma|_{r_0'} \sim_{C(\rho|_{r})} \omega_{r_0}$, which establishes $\omega_{r_0} \sim_{C(\rho|_{r})} \tau|_{r_0}$. Since, by the induction hypothesis, $\myFunc(\omega_{r_0})=\myFunc(\tau|_{r_0})$,
we again end up with $\myFunc(\rho|_r)=\myFunc(\sigma|_r)$, which provides the required contradiction.

However, we also need to make sure that inconsistent labels cannot be assigned
in line~\ref{line:assignOtherWTL} and any of the other lines, possibly
in different rounds. For a contradiction, we assume a ``generic'' setting that
can be fit to all cases: We assume that $\sigma|_{r'}$ got its label $\myFunc(\sigma|_{r'})=\myFunc(\tau|_{r'})\neq\emptyset$ 
assigned in iteration $r'\leq r$
in line~\ref{line:assignOtherWTL} or line~\ref{line:assignOtherAfterWTL},
since there was some already labeled $\tau|_{r'}\simmaj \sigma|_{r'}$ with $\mature(\tau,r')=\true$ but
$\mature(\sigma,r')=\false$. Moreover, we assume that $\rho|_r$ gets assigned its label 
$\myFunc(\sigma|_r) \neq \myFunc(\rho|_{r})=\myFunc(\omega_r)\neq\emptyset$ in iteration~$r\geq r' > \rST{\tau}+\Theta$
also in line~\ref{line:assignOtherWTL} or in line~\ref{line:assignOtherAfterWTL}, since there is some already labeled $\omega_{r}\simmaj \rho|_{r}$ 
with $\mature(\omega,r)=\true$ but $\mature(\rho,r)=\false$. Note carefully that we can rule out
the possibility that there are two different, say, $\sigma|_{r'}$ and $\sigma'|_{r'}$, with
inconsistent labels, which both match the condition of line~\ref{line:assignOtherWTL} or line~\ref{line:assignOtherAfterWTL}: This is prohibited by the induction hypothesis, except in the case of $r'=r$, 
where the above generic scenario applies.

To also cover the cases where $\rho|_r$ gets it label assigned in the other lines, we can set $\rho|_r=\omega_r$
in our considerations below. Note that the induction hypothesis again rules out the possibility that there 
are two different, say, $\sigma|_{r_0}$ and $\sigma'|_{r_0}$, with inconsistent labels, which both match the 
condition of line~\ref{line:assignOwnfromOldWTL} here, since $r_0<r$. 

Let $Q'\subseteq C(\tau|_{r'})$ be the set of at least $n-f$ processes causing 
$\tau|_{r'}\simmaj \sigma|_{r'}$, and $Q''\subseteq C(\omega_{r})$ be the set 
of at least $n-f$ non-faulty processes causing $\omega_{r}\simmaj \rho|_{r}$.
Since $\mature(\tau,r')=\true$ and $\mature(\omega,r)=\true$, \cref{lem:majindistprecond}
implies 
\begin{align}
&\tau|_{r_0'}  \sim_{C(\tau|_{r'})}  \sigma|_{r_0'}  \sim_{C(\tau|_{r'})}  \rho|_{r_0'}\label{eq:rnullprime}\\ 
&\sigma|_{r_0} \sim_{C(\omega_{r})} \rho|_{r_0}  \sim_{C(\omega_{r})}  \omega_{r_0}\label{eq:rnull}
\end{align}

We first consider the case $r_0' \leq r_0 \leq r'$:
Since $Q'\subseteq C(\tau|_{r'})$, $\tau|_{r'}\sim_{Q'} \sigma|_{r'}$ also 
implies $\tau|_{r_0}\sim_{Q'} \sigma|_{r_0}$. As $\oldenough(\tau,r_0')=\true$,
\cref{lem:propertiesoldenough}.(ii) also ensures $\oldenough(\tau,r_0)=\true$.
Moreover, since obviously $Q' \cap C(\omega_r) \neq \emptyset$ as well, we finally
observe that actually $\tau|_{r_0} \sim_{Q' \cap C(\omega_r)} \omega_{r_0}$. 
By \cref{lem:notbotholdenough}, we hence find that $p_\omega = p_\tau$.
Now there are two possibilities: If actually $\tau|_{r_0} \simmaj \omega_{r_0}$
holds, line~\ref{line:assignOwnfromOldWTL} implies that 
$\myFunc(\omega_r)=\myFunc(\tau|_{r_0})$.
Otherwise, every process will eventually be able to distinguish $\tau|_{r}$ and $\omega_{r}$
and, hence, $\rho|_r$ and $\sigma|_r$ by \cref{lem:vanishingminority}.
Both are contradictions to one of our assumptions $\myFunc(\omega_r)\neq
\myFunc(\tau|_{r_0})$ and $\rho|_r \simmaj \sigma|_r$.

To handle the case $r_0' > r_0$, we note that we can repeat exactly the same
arguments as above if we exchange the roles of $\omega_r$ and $\tau|_{r'}$
and $\sigma|_r$ and $\rho|_r$. 
In the only possible case of $r_0 \leq r_0' \leq r$,
since $Q''\subseteq C(\omega_{r})$, $\omega_{r}\sim_{Q''} \rho|_{r}$ also 
implies $\omega_{r_0'}\sim_{Q''} \rho|_{r_0'}$. As $\oldenough(\omega,r_0)=\true$,
\cref{lem:propertiesoldenough}.(ii) also ensures $\oldenough(\omega,r_0')=\true$.
Moreover, since obviously $Q'' \cap C(\tau|_{r'}) \neq \emptyset$ as well, 
we finally observe that actually $\omega_{r_0'} \sim_{Q'' \cap C(\tau|_{r'})} \tau|_{r_0'}$. By
\cref{lem:notbotholdenough}, we hence find again that $p_\omega = p_\tau$.
The same arguments as used in the previous paragraph establish the required
contradictions.

In the remaining case $r_0' \leq r_0$ but $r_0 > r'$, we have the situation where
$\sigma|_{r'}$ has already assigned its label \emph{before} round $r_0$, where
$\oldenough(\rho,r_0)=\true$. In general, every process may be able to 
distinguish $\rho$ and $\sigma$ (not to speak of $\tau$ and $\omega$) after
$r_0$, and usually $p_\tau\neq p_\omega$, so nothing would prevent $\myFunc(\sigma|_r)\neq \myFunc(\rho|_r)$ if the labeling algorithm would not have taken
special care, namely, in line~\ref{line:assignOwnfromOldWTL}: 
Rather than just assigning $\myFunc(\rho|_r)= \{ p_\omega \}$, it uses the
label of $\sigma|_{r_0}$ and therefore trivially avoids inconsistent labels. 
Note carefully that doing this is well-defined: If there were two different 
eligible $\sigma|_{r_0}$ and $\sigma'|_{r_0}$ available in 
line~\ref{line:assignOwnfromOldWTL}, \eqref{eq:rnullprime} reveals that
$\sigma|_{r_0} \simmaj \sigma'|_{r_0}$, such that their labels must be the
same by the induction hypothesis.

This completes the proof of our theorem.
\end{proof}

The following \cref{lem:DeltabroadcastWTL} finally confirms 
that a non-empty label $p$
assigned to some prefix $\sigma|_r$ is indeed a broadcaster:

\begin{lemma}\label{lem:DeltabroadcastWTL}
If $\myFunc(\sigma|_r) = \{p\}$ is computed by \cref{alg:lambda}, then $(p,0,I_p(\sigma))$
is contained in the view $V_q(\sigma|_r)$ of every process $q\in\Pi\setminus{F(\sigma|_r)}$ that 
has not crashed in $\sigma|_r$.
\end{lemma}

\begin{proof}
We distinguish the two essential cases where
$\rho|_r \in \Sigma_p$ can get its label $\{ p \}$: If $\myFunc(\rho|_r)$ was assigned via 
line~\ref{line:assignOwnWTL}, the dominant $p_\rho$ must indeed have reached all correct
processes in the system according to \cref{def:WTLextendedpredicates}
of $\oldenough(\rho,r_0)$, which is incorporated in $\mature(\rho,r)$.
In all other cases, $\myFunc(\rho|_r)$ was assigned since there is some 
$\sigma|_{r'}\sim_{s'}\rho|_{r'}$, $r' \leq r$, 
with at least $\oldenough(\sigma,r')=\true$. By the same argument as before,
the dominant $p_\sigma$
must have reached every correct process in $\sigma|_{r'}$ already.
As $\minhearedof_{s'}(\sigma,r') \geq \rST{\sigma}+\Theta$ according to the definition 
of $\oldenough(\sigma,r')$ implies also $\minhearedof_{s'}(\rho,r') \geq \rST{\sigma}+\Theta$
since $\sigma|_{r'}\sim_{s'}\rho|_{r'}$, it follows that $p_\sigma$ has also reached all
correct processes in $\rho|_{r'}$ already.
\end{proof}

\section{Conclusions}
\label{sec:conclusions}

We provided a complete characterization of both uniform and non-uniform deterministic consensus solvability in distributed systems with benign process and communication failures using 
point-set topology. Consensus can only be solved when the space of admissible 
executions can be partitioned into disjoint decision sets that are both
closed and open in our topologies. We also showed that this requires exclusion of certain (fair and unfair) limit sequences, which limit broadcastability and 
happen to coincide with the forever bivalent executions constructed
in bivalence and bipotence proofs. The utility and wide applicability of
our characterization was demonstrated by applying it to several different 
distributed computing models.

Part of our future work will be devoted to a generalization of our topological framework to other decision problems.
Since the initial publication of our results, this generalized study has been started by Attiya, Casta\~neda, and Nowak~\cite{ACN23}.
Another very interesting area of future research is to study the homology
of non-compact message adversaries, i.e., a more detailed topological structure
of the space of admissible executions.

\section*{Acknowledgments}
We gratefully acknowledge the suggestions of the reviewers, which stimulated
the inclusion of several additional results and pointed out many ways to improve 
our paper.

\bibliography{lit}


\begin{thebibliography}{49}


\ifx \showCODEN    \undefined \def \showCODEN     #1{\unskip}     \fi
\ifx \showDOI      \undefined \def \showDOI       #1{#1}\fi
\ifx \showISBNx    \undefined \def \showISBNx     #1{\unskip}     \fi
\ifx \showISBNxiii \undefined \def \showISBNxiii  #1{\unskip}     \fi
\ifx \showISSN     \undefined \def \showISSN      #1{\unskip}     \fi
\ifx \showLCCN     \undefined \def \showLCCN      #1{\unskip}     \fi
\ifx \shownote     \undefined \def \shownote      #1{#1}          \fi
\ifx \showarticletitle \undefined \def \showarticletitle #1{#1}   \fi
\ifx \showURL      \undefined \def \showURL       {\relax}        \fi
\providecommand\bibfield[2]{#2}
\providecommand\bibinfo[2]{#2}
\providecommand\natexlab[1]{#1}
\providecommand\showeprint[2][]{arXiv:#2}

\bibitem[\protect\citeauthoryear{Afek, Attiya, Dolev, Gafni, Merritt, and
  Shavit}{Afek et~al\mbox{.}}{1993}]%
        {AADG93:JACM}
\bibfield{author}{\bibinfo{person}{Yehuda Afek}, \bibinfo{person}{Hagit
  Attiya}, \bibinfo{person}{Danny Dolev}, \bibinfo{person}{Eli Gafni},
  \bibinfo{person}{Michael Merritt}, {and} \bibinfo{person}{Nir Shavit}.}
  \bibinfo{year}{1993}\natexlab{}.
\newblock \showarticletitle{Atomic snapshots of shared memory}.
\newblock \bibinfo{journal}{\emph{J. ACM}} \bibinfo{volume}{40},
  \bibinfo{number}{4} (\bibinfo{year}{1993}), \bibinfo{pages}{873–--890}.
\newblock
\urldef\tempurl%
\url{https://doi.org/10.1145/153724.153741}
\showDOI{\tempurl}


\bibitem[\protect\citeauthoryear{Afek and Gafni}{Afek and Gafni}{2013}]%
        {AG13}
\bibfield{author}{\bibinfo{person}{Yehuda Afek} {and} \bibinfo{person}{Eli
  Gafni}.} \bibinfo{year}{2013}\natexlab{}.
\newblock \showarticletitle{Asynchrony from Synchrony}.
\newblock In \bibinfo{booktitle}{\emph{Proceedings of the 14th International
  Conference on Distributed Computing and Networking (ICDCN 2013)}},
  \bibfield{editor}{\bibinfo{person}{Davide Frey}, \bibinfo{person}{Michel
  Raynal}, \bibinfo{person}{Saswati Sarkar}, \bibinfo{person}{Rundrapatna~K.
  Shyamasundar}, {and} \bibinfo{person}{Prasun Sinha}} (Eds.).
  \bibinfo{series}{Lecture Notes in Computer Science},
  Vol.~\bibinfo{volume}{7730}. \bibinfo{publisher}{Springer},
  \bibinfo{address}{Heidelberg}, \bibinfo{pages}{225--239}.
\newblock
\showISBNx{978-3-642-35667-4}
\urldef\tempurl%
\url{https://doi.org/10.1007/978-3-642-35668-1\_16}
\showDOI{\tempurl}


\bibitem[\protect\citeauthoryear{Aguilera, Delporte-Gallet, Fauconnier, and
  Toueg}{Aguilera et~al\mbox{.}}{2004}]%
        {ADGFT04}
\bibfield{author}{\bibinfo{person}{Marcos~K. Aguilera}, \bibinfo{person}{Carole
  Delporte-Gallet}, \bibinfo{person}{Hugues Fauconnier}, {and}
  \bibinfo{person}{Sam Toueg}.} \bibinfo{year}{2004}\natexlab{}.
\newblock \showarticletitle{Communication-efficient Leader Election and
  Consensus with Limited Link Synchrony}. In
  \bibinfo{booktitle}{\emph{Proceedings of the 23th ACM Symposium on Principles
  of Distributed Computing (PODC 2004)}},
  \bibfield{editor}{\bibinfo{person}{Shay Kutten}} (Ed.).
  \bibinfo{publisher}{ACM Press}, \bibinfo{address}{New York},
  \bibinfo{pages}{328--337}.
\newblock
\urldef\tempurl%
\url{https://doi.org/10.1145/1011767.1011816}
\showDOI{\tempurl}


\bibitem[\protect\citeauthoryear{Alpern and Schneider}{Alpern and
  Schneider}{1985}]%
        {AS84}
\bibfield{author}{\bibinfo{person}{Bowen Alpern} {and} \bibinfo{person}{Fred~B.
  Schneider}.} \bibinfo{year}{1985}\natexlab{}.
\newblock \showarticletitle{Defining liveness}.
\newblock \bibinfo{journal}{\emph{Inform. Process. Lett.}}
  \bibinfo{volume}{21}, \bibinfo{number}{4} (\bibinfo{year}{1985}),
  \bibinfo{pages}{181--185}.
\newblock
\urldef\tempurl%
\url{https://doi.org/10.1016/0020-0190(85)90056-0}
\showURL{%
\tempurl}


\bibitem[\protect\citeauthoryear{Attiya, Casta\~neda, and Nowak}{Attiya
  et~al\mbox{.}}{2023}]%
        {ACN23}
\bibfield{author}{\bibinfo{person}{Hagit Attiya}, \bibinfo{person}{Armando
  Casta\~neda}, {and} \bibinfo{person}{Thomas Nowak}.}
  \bibinfo{year}{2023}\natexlab{}.
\newblock \showarticletitle{Topological Characterization of Task Solvability in
  General Models of Computation}.
\newblock In \bibinfo{booktitle}{\emph{Proceedings of the 37th International
  Symposium on Distributed Computing (DISC 2023)}},
  \bibfield{editor}{\bibinfo{person}{Rotem Oshman}} (Ed.).
  \bibinfo{publisher}{Schloss Dagstuhl -- Leibniz-Zentrum f\"ur Informatik},
  \bibinfo{address}{Dagstuhl}, \bibinfo{pages}{24:1--24:23}.
\newblock
\urldef\tempurl%
\url{https://doi.org/10.4230/LIPIcs.DISC.2023.5}
\showURL{%
\tempurl}


\bibitem[\protect\citeauthoryear{Attiya, Casta{\~{n}}eda, and Rajsbaum}{Attiya
  et~al\mbox{.}}{2020}]%
        {ACR20:OPODIS}
\bibfield{author}{\bibinfo{person}{Hagit Attiya}, \bibinfo{person}{Armando
  Casta{\~{n}}eda}, {and} \bibinfo{person}{Sergio Rajsbaum}.}
  \bibinfo{year}{2020}\natexlab{}.
\newblock \showarticletitle{Locally Solvable Tasks and the Limitations of
  Valency Arguments}. In \bibinfo{booktitle}{\emph{Proceedings of the 24th
  International Conference on Principles of Distributed Systems ({OPODIS}
  2020)}}, \bibfield{editor}{\bibinfo{person}{Quentin Bramas},
  \bibinfo{person}{Rotem Oshman}, {and} \bibinfo{person}{Paolo Romano}} (Eds.).
  \bibinfo{publisher}{Schloss Dagstuhl - Leibniz-Zentrum f{\"{u}}r Informatik},
  \bibinfo{pages}{18:1--18:16}.
\newblock
\urldef\tempurl%
\url{https://doi.org/10.4230/LIPIcs.OPODIS.2020.18}
\showDOI{\tempurl}


\bibitem[\protect\citeauthoryear{Attiya and Welch}{Attiya and Welch}{2004}]%
        {AW04}
\bibfield{author}{\bibinfo{person}{Hagit Attiya} {and}
  \bibinfo{person}{Jennifer Welch}.} \bibinfo{year}{2004}\natexlab{}.
\newblock \bibinfo{booktitle}{\emph{Distributed Computing}
  (\bibinfo{edition}{2nd} ed.)}.
\newblock \bibinfo{publisher}{John Wiley {\&} Sons},
  \bibinfo{address}{Hoboken}.
\newblock


\bibitem[\protect\citeauthoryear{Ben-Zvi and Moses}{Ben-Zvi and Moses}{2014}]%
        {BM14:JACM}
\bibfield{author}{\bibinfo{person}{Ido Ben-Zvi} {and} \bibinfo{person}{Yoram
  Moses}.} \bibinfo{year}{2014}\natexlab{}.
\newblock \showarticletitle{Beyond {L}amport's happened-before: on time bounds
  and the ordering of events in distributed systems}.
\newblock \bibinfo{journal}{\emph{J. ACM}} \bibinfo{volume}{61},
  \bibinfo{number}{2} (\bibinfo{year}{2014}), \bibinfo{pages}{13:1--13:26}.
\newblock
\showISSN{0004-5411}
\urldef\tempurl%
\url{https://doi.org/10.1145/2542181}
\showDOI{\tempurl}


\bibitem[\protect\citeauthoryear{Biely and Robinson}{Biely and
  Robinson}{2019}]%
        {BR19:ICDCN}
\bibfield{author}{\bibinfo{person}{Martin Biely} {and} \bibinfo{person}{Peter
  Robinson}.} \bibinfo{year}{2019}\natexlab{}.
\newblock \showarticletitle{On the Hardness of the Strongly Dependent Decision
  Problem}. In \bibinfo{booktitle}{\emph{Proceedings of the 20th International
  Conference on Distributed Computing and Networking (ICDCN 2019)}}.
  \bibinfo{publisher}{ACM Press}, \bibinfo{address}{New York},
  \bibinfo{pages}{120--123}.
\newblock
\showISBNx{978-1-4503-6094-4}
\urldef\tempurl%
\url{https://doi.org/10.1145/3288599.3288614}
\showDOI{\tempurl}


\bibitem[\protect\citeauthoryear{Biely, Robinson, Schmid, Schwarz, and
  Winkler}{Biely et~al\mbox{.}}{2018}]%
        {BRSSW18:TCS}
\bibfield{author}{\bibinfo{person}{Martin Biely}, \bibinfo{person}{Peter
  Robinson}, \bibinfo{person}{Ulrich Schmid}, \bibinfo{person}{Manfred
  Schwarz}, {and} \bibinfo{person}{Kyrill Winkler}.}
  \bibinfo{year}{2018}\natexlab{}.
\newblock \showarticletitle{Gracefully degrading consensus and k-set agreement
  in directed dynamic networks}.
\newblock \bibinfo{journal}{\emph{Theor. Comput. Sci.}}  \bibinfo{volume}{726}
  (\bibinfo{year}{2018}), \bibinfo{pages}{41--77}.
\newblock
\showISSN{0304-3975}
\urldef\tempurl%
\url{https://doi.org/10.1016/j.tcs.2018.02.019}
\showDOI{\tempurl}


\bibitem[\protect\citeauthoryear{Casta{\~{n}}eda, Fraigniaud, Paz, Rajsbaum,
  Roy, and Travers}{Casta{\~{n}}eda et~al\mbox{.}}{2019}]%
        {CFPR19:SSS}
\bibfield{author}{\bibinfo{person}{Armando Casta{\~{n}}eda},
  \bibinfo{person}{Pierre Fraigniaud}, \bibinfo{person}{Ami Paz},
  \bibinfo{person}{Sergio Rajsbaum}, \bibinfo{person}{Matthieu Roy}, {and}
  \bibinfo{person}{Corentin Travers}.} \bibinfo{year}{2019}\natexlab{}.
\newblock \showarticletitle{Synchronous t-Resilient Consensus in Arbitrary
  Graphs}. In \bibinfo{booktitle}{\emph{Proceedings of the 21st Symposium on
  Stabilization, Safety, and Security of Distributed Systems (SSS 2019)}},
  \bibfield{editor}{\bibinfo{person}{Mohsen Ghaffari}, \bibinfo{person}{Mikhail
  Nesterenko}, \bibinfo{person}{S{\'e}bastien Tixeuil}, \bibinfo{person}{Sara
  Tucci}, {and} \bibinfo{person}{Yukiko Yamauchi}} (Eds.).
  \bibinfo{publisher}{Springer}, \bibinfo{address}{Heidelberg},
  \bibinfo{pages}{53--68}.
\newblock
\urldef\tempurl%
\url{https://doi.org/10.1007/978-3-030-34992-9_5}
\showDOI{\tempurl}


\bibitem[\protect\citeauthoryear{Chandra and Toueg}{Chandra and Toueg}{1996}]%
        {CT96}
\bibfield{author}{\bibinfo{person}{Tushar~Deepak Chandra} {and}
  \bibinfo{person}{Sam Toueg}.} \bibinfo{year}{1996}\natexlab{}.
\newblock \showarticletitle{Unreliable failure detectors for reliable
  distributed systems}.
\newblock \bibinfo{journal}{\emph{J. ACM}} \bibinfo{volume}{43},
  \bibinfo{number}{2} (\bibinfo{date}{March} \bibinfo{year}{1996}),
  \bibinfo{pages}{225--267}.
\newblock
\urldef\tempurl%
\url{https://doi.org/10.1145/226643.226647}
\showURL{%
\tempurl}


\bibitem[\protect\citeauthoryear{Charron-Bost and Schiper}{Charron-Bost and
  Schiper}{2009}]%
        {CBS09}
\bibfield{author}{\bibinfo{person}{Bernadette Charron-Bost} {and}
  \bibinfo{person}{Andr\'{e} Schiper}.} \bibinfo{year}{2009}\natexlab{}.
\newblock \showarticletitle{The {H}eard-{O}f model: computing in distributed
  systems with benign faults}.
\newblock \bibinfo{journal}{\emph{Distrib.\ Comput.}} \bibinfo{volume}{22},
  \bibinfo{number}{1} (\bibinfo{date}{April} \bibinfo{year}{2009}),
  \bibinfo{pages}{49--71}.
\newblock
\urldef\tempurl%
\url{https://doi.org/10.1007/s00446-009-0084-6}
\showDOI{\tempurl}


\bibitem[\protect\citeauthoryear{Coulouma, Godard, and Peters}{Coulouma
  et~al\mbox{.}}{2015}]%
        {CGP15}
\bibfield{author}{\bibinfo{person}{{\'{E}}tienne Coulouma},
  \bibinfo{person}{Emmanuel Godard}, {and} \bibinfo{person}{Joseph~G. Peters}.}
  \bibinfo{year}{2015}\natexlab{}.
\newblock \showarticletitle{A characterization of oblivious message adversaries
  for which Consensus is solvable}.
\newblock \bibinfo{journal}{\emph{Theor. Comput. Sci.}}  \bibinfo{volume}{584}
  (\bibinfo{year}{2015}), \bibinfo{pages}{80--90}.
\newblock
\urldef\tempurl%
\url{https://doi.org/10.1016/j.tcs.2015.01.024}
\showDOI{\tempurl}


\bibitem[\protect\citeauthoryear{Dolev, Dwork, and Stockmeyer}{Dolev
  et~al\mbox{.}}{1987}]%
        {DDS87}
\bibfield{author}{\bibinfo{person}{Danny Dolev}, \bibinfo{person}{Cynthia
  Dwork}, {and} \bibinfo{person}{Larry Stockmeyer}.}
  \bibinfo{year}{1987}\natexlab{}.
\newblock \showarticletitle{On the minimal synchronism needed for distributed
  consensus}.
\newblock \bibinfo{journal}{\emph{J. ACM}} \bibinfo{volume}{34},
  \bibinfo{number}{1} (\bibinfo{year}{1987}), \bibinfo{pages}{77--97}.
\newblock
\urldef\tempurl%
\url{https://doi.org/10.1145/7531.7533}
\showURL{%
\tempurl}


\bibitem[\protect\citeauthoryear{Dwork, Lynch, and Stockmeyer}{Dwork
  et~al\mbox{.}}{1988}]%
        {DLS88}
\bibfield{author}{\bibinfo{person}{Cynthia Dwork}, \bibinfo{person}{Nancy
  Lynch}, {and} \bibinfo{person}{Larry Stockmeyer}.}
  \bibinfo{year}{1988}\natexlab{}.
\newblock \showarticletitle{Consensus in the presence of partial synchrony}.
\newblock \bibinfo{journal}{\emph{J. ACM}} \bibinfo{volume}{35},
  \bibinfo{number}{2} (\bibinfo{year}{1988}), \bibinfo{pages}{288--323}.
\newblock
\urldef\tempurl%
\url{https://doi.org/10.1145/42282.42283}
\showURL{%
\tempurl}


\bibitem[\protect\citeauthoryear{Fevat and Godard}{Fevat and Godard}{2011}]%
        {FG11}
\bibfield{author}{\bibinfo{person}{Tristan Fevat} {and}
  \bibinfo{person}{Emmanuel Godard}.} \bibinfo{year}{2011}\natexlab{}.
\newblock \showarticletitle{Minimal Obstructions for the Coordinated Attack
  Problem and Beyond}. In \bibinfo{booktitle}{\emph{Proceedings of the 25th
  {IEEE} International Symposium on Parallel and Distributed Processing,
  ({IPDPS} 2011)}}. \bibinfo{pages}{1001--1011}.
\newblock
\urldef\tempurl%
\url{https://doi.org/10.1109/IPDPS.2011.96}
\showDOI{\tempurl}


\bibitem[\protect\citeauthoryear{Fich and Ruppert}{Fich and Ruppert}{2003}]%
        {FR03}
\bibfield{author}{\bibinfo{person}{Faith Fich} {and} \bibinfo{person}{Eric
  Ruppert}.} \bibinfo{year}{2003}\natexlab{}.
\newblock \showarticletitle{Hundreds of impossibility results for distributed
  computing}.
\newblock \bibinfo{journal}{\emph{Distributed Computing}}  \bibinfo{volume}{16}
  (\bibinfo{year}{2003}), \bibinfo{pages}{121--163}.
\newblock
\urldef\tempurl%
\url{https://doi.org/10.1007/s00446-003-0091-y}
\showDOI{\tempurl}


\bibitem[\protect\citeauthoryear{Fischer, Lynch, and Paterson}{Fischer
  et~al\mbox{.}}{1985}]%
        {FLP85}
\bibfield{author}{\bibinfo{person}{Michael~J. Fischer},
  \bibinfo{person}{Nancy~A. Lynch}, {and} \bibinfo{person}{Michael~S.
  Paterson}.} \bibinfo{year}{1985}\natexlab{}.
\newblock \showarticletitle{Impossibility of distributed consensus with one
  faulty process}.
\newblock \bibinfo{journal}{\emph{J. ACM}} \bibinfo{volume}{32},
  \bibinfo{number}{2} (\bibinfo{year}{1985}), \bibinfo{pages}{374--382}.
\newblock
\urldef\tempurl%
\url{https://doi.org/10.1145/3149.214121}
\showURL{%
\tempurl}


\bibitem[\protect\citeauthoryear{Gafni, Kuznetsov, and Manolescu}{Gafni
  et~al\mbox{.}}{2014}]%
        {GKM14:PODC}
\bibfield{author}{\bibinfo{person}{Eli Gafni}, \bibinfo{person}{Petr
  Kuznetsov}, {and} \bibinfo{person}{Ciprian Manolescu}.}
  \bibinfo{year}{2014}\natexlab{}.
\newblock \showarticletitle{A Generalized Asynchronous Computability Theorem}.
\newblock In \bibinfo{booktitle}{\emph{Proceedings of the 33rd ACM Symposium on
  Principles of Distributed Computing (PODC 2014)}},
  \bibfield{editor}{\bibinfo{person}{Shlomi Dolev}} (Ed.).
  \bibinfo{publisher}{ACM Press}, \bibinfo{address}{New York},
  \bibinfo{pages}{222--–231}.
\newblock
\showISBNx{9781450329446}
\urldef\tempurl%
\url{https://doi.org/10.1145/2611462.2611477}
\showDOI{\tempurl}


\bibitem[\protect\citeauthoryear{Galeana, Rajsbaum, and Schmid}{Galeana
  et~al\mbox{.}}{2022}]%
        {GRS22:ITCS}
\bibfield{author}{\bibinfo{person}{Hugo~Rincon Galeana},
  \bibinfo{person}{Sergio Rajsbaum}, {and} \bibinfo{person}{Ulrich Schmid}.}
  \bibinfo{year}{2022}\natexlab{}.
\newblock \showarticletitle{Continuous tasks and the asynchronous computability
  theorem}.
\newblock In \bibinfo{booktitle}{\emph{Proceedings of the 13th Innovations in
  Theoretical Computer Science Conference ({ITCS} 2022)}},
  \bibfield{editor}{\bibinfo{person}{Mark Braverman}} (Ed.).
  \bibinfo{publisher}{Schloss Dagstuhl - Leibniz-Zentrum f{\"{u}}r Informatik},
  \bibinfo{pages}{73:1--73:27}.
\newblock
\urldef\tempurl%
\url{https://doi.org/10.4230/LIPIcs.ITCS.2022.73}
\showDOI{\tempurl}


\bibitem[\protect\citeauthoryear{Galeana, Schmid, Winkler, Paz, and
  Schmid}{Galeana et~al\mbox{.}}{2023}]%
        {RSWP23:arxiv}
\bibfield{author}{\bibinfo{person}{Hugo~Rincon Galeana},
  \bibinfo{person}{Ulrich Schmid}, \bibinfo{person}{Kyrill Winkler},
  \bibinfo{person}{Ami Paz}, {and} \bibinfo{person}{Stefan Schmid}.}
  \bibinfo{year}{2023}\natexlab{}.
\newblock \bibinfo{title}{Topological Characterization of Consensus Solvability
  in Directed Dynamic Networks}.
\newblock
\newblock
\urldef\tempurl%
\url{http://arxiv.org/abs/2304.02316}
\showURL{%
\tempurl}


\bibitem[\protect\citeauthoryear{Godard and Perdereau}{Godard and
  Perdereau}{2020}]%
        {GP20:MSCS}
\bibfield{author}{\bibinfo{person}{Emmanuel Godard} {and} \bibinfo{person}{Eloi
  Perdereau}.} \bibinfo{year}{2020}\natexlab{}.
\newblock \showarticletitle{Back to the Coordinated Attack Problem}.
\newblock \bibinfo{journal}{\emph{Math. Struct. Comput. Sci.}}
  \bibinfo{volume}{30}, \bibinfo{number}{10} (\bibinfo{year}{2020}),
  \bibinfo{pages}{1089--1113}.
\newblock
\urldef\tempurl%
\url{https://doi.org/10.1017/S0960129521000037}
\showDOI{\tempurl}


\bibitem[\protect\citeauthoryear{Herlihy, Kozlov, and Rajsbaum}{Herlihy
  et~al\mbox{.}}{2013}]%
        {HKR13}
\bibfield{author}{\bibinfo{person}{Maurice Herlihy}, \bibinfo{person}{Dmitry~N.
  Kozlov}, {and} \bibinfo{person}{Sergio Rajsbaum}.}
  \bibinfo{year}{2013}\natexlab{}.
\newblock \bibinfo{booktitle}{\emph{Distributed Computing Through Combinatorial
  Topology}}.
\newblock \bibinfo{publisher}{Morgan Kaufmann}.
\newblock
\showISBNx{978-0-12-404578-1}
\urldef\tempurl%
\url{https://store.elsevier.com/product.jsp?isbn=9780124045781}
\showURL{%
\tempurl}


\bibitem[\protect\citeauthoryear{Herlihy and Shavit}{Herlihy and
  Shavit}{1999}]%
        {HS99:ACT}
\bibfield{author}{\bibinfo{person}{Maurice Herlihy} {and} \bibinfo{person}{Nir
  Shavit}.} \bibinfo{year}{1999}\natexlab{}.
\newblock \showarticletitle{The topological structure of asynchronous
  computability}.
\newblock \bibinfo{journal}{\emph{J. ACM}} \bibinfo{volume}{46},
  \bibinfo{number}{6} (\bibinfo{year}{1999}), \bibinfo{pages}{858--923}.
\newblock
\urldef\tempurl%
\url{https://doi.org/10.1145/331524.331529}
\showDOI{\tempurl}


\bibitem[\protect\citeauthoryear{Hutle, Malkhi, Schmid, and Zhou}{Hutle
  et~al\mbox{.}}{2009}]%
        {HMSZ08:TDSC}
\bibfield{author}{\bibinfo{person}{Martin Hutle}, \bibinfo{person}{Dahlia
  Malkhi}, \bibinfo{person}{Ulrich Schmid}, {and} \bibinfo{person}{Lidong
  Zhou}.} \bibinfo{year}{2009}\natexlab{}.
\newblock \showarticletitle{Chasing the weakest system model for implementing
  {O}mega and consensus}.
\newblock \bibinfo{journal}{\emph{IEEE T. Depend.\ Secure}}
  \bibinfo{volume}{6}, \bibinfo{number}{4} (\bibinfo{year}{2009}),
  \bibinfo{pages}{269--281}.
\newblock
\urldef\tempurl%
\url{https://doi.org/10.1109/TDSC.2008.24}
\showDOI{\tempurl}


\bibitem[\protect\citeauthoryear{Kuhn and Oshman}{Kuhn and Oshman}{2011}]%
        {KO11:SIGACT}
\bibfield{author}{\bibinfo{person}{Fabian Kuhn} {and} \bibinfo{person}{Rotem
  Oshman}.} \bibinfo{year}{2011}\natexlab{}.
\newblock \showarticletitle{Dynamic networks: models and algorithms}.
\newblock \bibinfo{journal}{\emph{SIGACT News}}  \bibinfo{volume}{42(1)}
  (\bibinfo{year}{2011}), \bibinfo{pages}{82--96}.
\newblock
\urldef\tempurl%
\url{https://doi.org/10.1145/1959045.1959064}
\showURL{%
\tempurl}


\bibitem[\protect\citeauthoryear{Kuznetsov, Rieutord, and He}{Kuznetsov
  et~al\mbox{.}}{2018}]%
        {KRH18:PODC}
\bibfield{author}{\bibinfo{person}{Petr Kuznetsov}, \bibinfo{person}{Thibault
  Rieutord}, {and} \bibinfo{person}{Yuan He}.} \bibinfo{year}{2018}\natexlab{}.
\newblock \showarticletitle{An Asynchronous Computability Theorem for Fair
  Adversaries}.
\newblock In \bibinfo{booktitle}{\emph{Proceedings of the 37th {ACM} Symposium
  on Principles of Distributed Computing ({PODC} 2018)}},
  \bibfield{editor}{\bibinfo{person}{Idit Keidar}} (Ed.).
  \bibinfo{publisher}{ACM Press}, \bibinfo{address}{New York},
  \bibinfo{pages}{387--396}.
\newblock
\urldef\tempurl%
\url{https://doi.org/10.1145/3212734.3212765}
\showURL{%
\tempurl}


\bibitem[\protect\citeauthoryear{Lamport}{Lamport}{1978}]%
        {Lam78}
\bibfield{author}{\bibinfo{person}{Leslie Lamport}.}
  \bibinfo{year}{1978}\natexlab{}.
\newblock \showarticletitle{Time, clocks, and the ordering of events in a
  distributed system}.
\newblock \bibinfo{journal}{\emph{Commun. ACM}} \bibinfo{volume}{21},
  \bibinfo{number}{7} (\bibinfo{year}{1978}), \bibinfo{pages}{558--565}.
\newblock
\showISSN{0001-0782}
\urldef\tempurl%
\url{https://doi.org/10.1145/359545.359563}
\showDOI{\tempurl}


\bibitem[\protect\citeauthoryear{Lamport, Shostak, and Pease}{Lamport
  et~al\mbox{.}}{1982}]%
        {LSP82}
\bibfield{author}{\bibinfo{person}{Leslie Lamport}, \bibinfo{person}{Robert
  Shostak}, {and} \bibinfo{person}{Marshall Pease}.}
  \bibinfo{year}{1982}\natexlab{}.
\newblock \showarticletitle{The {Byzantine} generals problem}.
\newblock \bibinfo{journal}{\emph{ACM T. Progr.\ Lang.\ Sys.}}
  \bibinfo{volume}{4}, \bibinfo{number}{3} (\bibinfo{year}{1982}),
  \bibinfo{pages}{382--401}.
\newblock
\urldef\tempurl%
\url{https://doi.org/10.1145/357172.357176}
\showDOI{\tempurl}


\bibitem[\protect\citeauthoryear{Lubitch and Moran}{Lubitch and Moran}{1995}]%
        {LM95:DC}
\bibfield{author}{\bibinfo{person}{Ronit Lubitch} {and} \bibinfo{person}{Shlomo
  Moran}.} \bibinfo{year}{1995}\natexlab{}.
\newblock \showarticletitle{Closed Schedulers: A Novel Technique for Analyzing
  Asynchronous Protocols}.
\newblock \bibinfo{journal}{\emph{Distrib. Comput.}} \bibinfo{volume}{8},
  \bibinfo{number}{4} (\bibinfo{date}{June} \bibinfo{year}{1995}),
  \bibinfo{pages}{203--210}.
\newblock
\showISSN{0178-2770}
\urldef\tempurl%
\url{https://doi.org/10.1007/BF02242738}
\showDOI{\tempurl}


\bibitem[\protect\citeauthoryear{Mattern}{Mattern}{1989}]%
        {Mat89}
\bibfield{author}{\bibinfo{person}{Friedemann Mattern}.}
  \bibinfo{year}{1989}\natexlab{}.
\newblock \showarticletitle{Virtual time and global states of distributed
  systems}.
\newblock In \bibinfo{booktitle}{\emph{Proceedings of the International
  Workshop on Parallel and Distributed Algorithms}},
  \bibfield{editor}{\bibinfo{person}{Michel Cosnard}, \bibinfo{person}{Yves
  Rober}, \bibinfo{person}{Patrice Quinton}, {and} \bibinfo{person}{Michel
  Raynal}} (Eds.). \bibinfo{publisher}{North Holland},
  \bibinfo{address}{Amsterdam}, \bibinfo{pages}{215--226}.
\newblock


\bibitem[\protect\citeauthoryear{Moses and Rajsbaum}{Moses and
  Rajsbaum}{2002}]%
        {MR02}
\bibfield{author}{\bibinfo{person}{Yoram Moses} {and} \bibinfo{person}{Sergio
  Rajsbaum}.} \bibinfo{year}{2002}\natexlab{}.
\newblock \showarticletitle{A layered analysis of consensus}.
\newblock \bibinfo{journal}{\emph{SIAM J. Comput.}} \bibinfo{volume}{31},
  \bibinfo{number}{4} (\bibinfo{year}{2002}), \bibinfo{pages}{989--1021}.
\newblock
\urldef\tempurl%
\url{https://doi.org/10.1137/S0097539799364006}
\showURL{%
\tempurl}


\bibitem[\protect\citeauthoryear{Mostefaoui, Rajsbaum, and Raynal}{Mostefaoui
  et~al\mbox{.}}{2003}]%
        {MRR03:JACM}
\bibfield{author}{\bibinfo{person}{Achour Mostefaoui}, \bibinfo{person}{Sergio
  Rajsbaum}, {and} \bibinfo{person}{Michel Raynal}.}
  \bibinfo{year}{2003}\natexlab{}.
\newblock \showarticletitle{Conditions on input vectors for consensus
  solvability in asynchronous distributed systems}.
\newblock \bibinfo{journal}{\emph{J. ACM}} \bibinfo{volume}{50},
  \bibinfo{number}{6} (\bibinfo{year}{2003}), \bibinfo{pages}{922–--954}.
\newblock
\urldef\tempurl%
\url{https://doi.org/10.1145/950620.950624}
\showDOI{\tempurl}


\bibitem[\protect\citeauthoryear{Most{\'e}faoui and Raynal}{Most{\'e}faoui and
  Raynal}{2001}]%
        {MR01}
\bibfield{author}{\bibinfo{person}{Achour Most{\'e}faoui} {and}
  \bibinfo{person}{Michel Raynal}.} \bibinfo{year}{2001}\natexlab{}.
\newblock \showarticletitle{Leader-based consensus}.
\newblock \bibinfo{journal}{\emph{Parallel Process.\ Lett.}}
  \bibinfo{volume}{11}, \bibinfo{number}{1} (\bibinfo{year}{2001}),
  \bibinfo{pages}{95--107}.
\newblock
\urldef\tempurl%
\url{https://doi.org/10.1142/S0129626401000452}
\showURL{%
\tempurl}


\bibitem[\protect\citeauthoryear{Munkres}{Munkres}{2000}]%
        {Munkres}
\bibfield{author}{\bibinfo{person}{James Munkres}.}
  \bibinfo{year}{2000}\natexlab{}.
\newblock \bibinfo{booktitle}{\emph{Topology} (\bibinfo{edition}{2nd} ed.)}.
\newblock \bibinfo{publisher}{Prentice Hall}, \bibinfo{address}{Hoboken}.
\newblock


\bibitem[\protect\citeauthoryear{Nowak}{Nowak}{2010}]%
        {Now10:master}
\bibfield{author}{\bibinfo{person}{Thomas Nowak}.}
  \bibinfo{year}{2010}\natexlab{}.
\newblock \emph{\bibinfo{title}{Topology in Distributed Computing}}.
\newblock \bibinfo{thesistype}{Master's\ thesis}. \bibinfo{school}{Embedded
  Computing Systems Group, Technische Universit\"at Wien}.
\newblock


\bibitem[\protect\citeauthoryear{Nowak, Schmid, and Winkler}{Nowak
  et~al\mbox{.}}{2019}]%
        {NSW19:PODC}
\bibfield{author}{\bibinfo{person}{Thomas Nowak}, \bibinfo{person}{Ulrich
  Schmid}, {and} \bibinfo{person}{Kyrill Winkler}.}
  \bibinfo{year}{2019}\natexlab{}.
\newblock \showarticletitle{Topological Characterization of Consensus under
  General Message Adversaries}. In \bibinfo{booktitle}{\emph{Proceedings of the
  28th {ACM} Symposium on Principles of Distributed Computing ({PODC} 2019)}},
  \bibfield{editor}{\bibinfo{person}{Faith Ellen}} (Ed.).
  \bibinfo{publisher}{ACM Press}, \bibinfo{address}{New York},
  \bibinfo{pages}{218--227}.
\newblock
\urldef\tempurl%
\url{https://doi.org/10.1145/3293611.3331624}
\showDOI{\tempurl}


\bibitem[\protect\citeauthoryear{{Parvedy} and {Raynal}}{{Parvedy} and
  {Raynal}}{2003}]%
        {PR03:IPDPS}
\bibfield{author}{\bibinfo{person}{P.~R. {Parvedy}} {and} \bibinfo{person}{M.
  {Raynal}}.} \bibinfo{year}{2003}\natexlab{}.
\newblock \showarticletitle{Uniform agreement despite process omission
  failures}.
\newblock In \bibinfo{booktitle}{\emph{Proceedings of the 17th International
  Parallel and Distributed Processing Symposium (IPDPS 2003)}},
  \bibfield{editor}{\bibinfo{person}{Jack Dongarra}} (Ed.).
  \bibinfo{publisher}{IEEE Press}, \bibinfo{address}{New York},
  \bibinfo{pages}{22--26}.
\newblock
\urldef\tempurl%
\url{https://doi.org/10.1109/IPDPS.2003.1213388}
\showDOI{\tempurl}


\bibitem[\protect\citeauthoryear{Perry and Toueg}{Perry and Toueg}{1986}]%
        {PT86}
\bibfield{author}{\bibinfo{person}{Kenneth~J. Perry} {and} \bibinfo{person}{Sam
  Toueg}.} \bibinfo{year}{1986}\natexlab{}.
\newblock \showarticletitle{Distributed agreement in the presence of processor
  and communication faults}.
\newblock \bibinfo{journal}{\emph{IEEE T. Software Eng.}}
  \bibinfo{volume}{SE-12}, \bibinfo{number}{3} (\bibinfo{year}{1986}),
  \bibinfo{pages}{477--482}.
\newblock
\urldef\tempurl%
\url{https://doi.org/10.1109/TSE.1986.6312888}
\showURL{%
\tempurl}


\bibitem[\protect\citeauthoryear{Pfleger}{Pfleger}{2018}]%
        {Pfl18:master}
\bibfield{author}{\bibinfo{person}{Daniel Pfleger}.}
  \bibinfo{year}{2018}\natexlab{}.
\newblock \emph{\bibinfo{title}{Knowledge and Communication Complexity}}.
\newblock \bibinfo{thesistype}{Master's\ thesis}. \bibinfo{school}{Embedded
  Computing Systems Group, Technische Universit\"at Wien}.
\newblock


\bibitem[\protect\citeauthoryear{Raynal and Stainer}{Raynal and
  Stainer}{2013}]%
        {RS13:PODC}
\bibfield{author}{\bibinfo{person}{Michel Raynal} {and} \bibinfo{person}{Julien
  Stainer}.} \bibinfo{year}{2013}\natexlab{}.
\newblock \showarticletitle{Synchrony Weakened by Message Adversaries vs
  Asynchrony Restricted by Failure Detectors}.
\newblock In \bibinfo{booktitle}{\emph{Proceedings of the 32nd ACM Symposium on
  Principles of Distributed Computing (PODC 2013)}},
  \bibfield{editor}{\bibinfo{person}{Gadi Taubenfeld}} (Ed.).
  \bibinfo{publisher}{ACM Press}, \bibinfo{address}{New York},
  \bibinfo{pages}{166--175}.
\newblock
\urldef\tempurl%
\url{https://doi.org/10.1145/2484239.2484249}
\showURL{%
\tempurl}


\bibitem[\protect\citeauthoryear{Robinson and Schmid}{Robinson and
  Schmid}{2011}]%
        {RS10:TCS}
\bibfield{author}{\bibinfo{person}{Peter Robinson} {and}
  \bibinfo{person}{Ulrich Schmid}.} \bibinfo{year}{2011}\natexlab{}.
\newblock \showarticletitle{The {A}synchronous {B}ounded-{C}ycle model}.
\newblock \bibinfo{journal}{\emph{Theor.\ Comput.\ Sci.}}
  \bibinfo{volume}{412}, \bibinfo{number}{40} (\bibinfo{year}{2011}),
  \bibinfo{pages}{5580--5601}.
\newblock
\urldef\tempurl%
\url{https://doi.org/10.1016/j.tcs.2010.08.001}
\showDOI{\tempurl}


\bibitem[\protect\citeauthoryear{Santoro and Widmayer}{Santoro and
  Widmayer}{1989}]%
        {SW89}
\bibfield{author}{\bibinfo{person}{Nicola Santoro} {and} \bibinfo{person}{Peter
  Widmayer}.} \bibinfo{year}{1989}\natexlab{}.
\newblock \showarticletitle{Time is Not a Healer}.
\newblock In \bibinfo{booktitle}{\emph{Proceedings of the 6th Annual Symposium
  on Theoretical Aspects of Computer Science (STACS 1989)}}.
  \bibinfo{publisher}{Springer}, \bibinfo{address}{Heidelberg},
  \bibinfo{pages}{304--313}.
\newblock


\bibitem[\protect\citeauthoryear{Schmid, Weiss, and Keidar}{Schmid
  et~al\mbox{.}}{2009}]%
        {SWK09}
\bibfield{author}{\bibinfo{person}{Ulrich Schmid}, \bibinfo{person}{Bettina
  Weiss}, {and} \bibinfo{person}{Idit Keidar}.}
  \bibinfo{year}{2009}\natexlab{}.
\newblock \showarticletitle{Impossibility results and lower bounds for
  consensus under link failures}.
\newblock \bibinfo{journal}{\emph{SIAM J. Comput.}} \bibinfo{volume}{38},
  \bibinfo{number}{5} (\bibinfo{year}{2009}), \bibinfo{pages}{1912--1951}.
\newblock
\urldef\tempurl%
\url{https://doi.org/10.1137/S009753970443999X}
\showDOI{\tempurl}


\bibitem[\protect\citeauthoryear{Winkler, Paz, Rincon~Galeana, Schmid, and
  Schmid}{Winkler et~al\mbox{.}}{2023}]%
        {WPRSS23:ITCS}
\bibfield{author}{\bibinfo{person}{Kyrill Winkler}, \bibinfo{person}{Ami Paz},
  \bibinfo{person}{Hugo Rincon~Galeana}, \bibinfo{person}{Stefan Schmid}, {and}
  \bibinfo{person}{Ulrich Schmid}.} \bibinfo{year}{2023}\natexlab{}.
\newblock \showarticletitle{The Time Complexity of Consensus Under Oblivious
  Message Adversaries}.
\newblock In \bibinfo{booktitle}{\emph{Proceedings of the 14th Innovations in
  Theoretical Computer Science Conference (ITCS 2023)}},
  \bibfield{editor}{\bibinfo{person}{Yael Tauman~Kalai}} (Ed.).
  \bibinfo{publisher}{Schloss Dagstuhl -- Leibniz-Zentrum f{\"u}r Informatik},
  \bibinfo{address}{Dagstuhl}, \bibinfo{pages}{100:1--100:28}.
\newblock
\urldef\tempurl%
\url{https://doi.org/10.4230/LIPIcs.ITCS.2023.100}
\showDOI{\tempurl}


\bibitem[\protect\citeauthoryear{Winkler, Schmid, and Moses}{Winkler
  et~al\mbox{.}}{2019a}]%
        {WSM19:OPODIS}
\bibfield{author}{\bibinfo{person}{Kyrill Winkler}, \bibinfo{person}{Ulrich
  Schmid}, {and} \bibinfo{person}{Yoram Moses}.}
  \bibinfo{year}{2019}\natexlab{a}.
\newblock \showarticletitle{A Characterization of Consensus Solvability for
  Closed Message Adversaries}. In \bibinfo{booktitle}{\emph{Proceedings of the
  23rd International Conference on Principles of Distributed Systems ({OPODIS}
  2019)}}. \bibinfo{publisher}{Schloss Dagstuhl - Leibniz-Zentrum f{\"{u}}r
  Informatik}, \bibinfo{address}{Dagstuhl}, \bibinfo{pages}{17:1--17:16}.
\newblock
\urldef\tempurl%
\url{https://doi.org/10.4230/LIPIcs.OPODIS.2019.17}
\showDOI{\tempurl}


\bibitem[\protect\citeauthoryear{Winkler, Schmid, and Nowak}{Winkler
  et~al\mbox{.}}{2021}]%
        {WSN21:FCT}
\bibfield{author}{\bibinfo{person}{Kyrill Winkler}, \bibinfo{person}{Ulrich
  Schmid}, {and} \bibinfo{person}{Thomas Nowak}.}
  \bibinfo{year}{2021}\natexlab{}.
\newblock \showarticletitle{Valency-Based Consensus Under Message Adversaries
  Without Limit-Closure}. In \bibinfo{booktitle}{\emph{Prceedings of the 23rd
  International Symposium on Fundamentals of Computation Theory ({FCT} 2021)}},
  \bibfield{editor}{\bibinfo{person}{Evripidis Bampis} {and}
  \bibinfo{person}{Aris Pagourtzis}} (Eds.). \bibinfo{publisher}{Springer},
  \bibinfo{address}{Heidelberg}, \bibinfo{pages}{457--474}.
\newblock
\urldef\tempurl%
\url{https://doi.org/10.1007/978-3-030-86593-1\_32}
\showDOI{\tempurl}


\bibitem[\protect\citeauthoryear{Winkler, Schwarz, and Schmid}{Winkler
  et~al\mbox{.}}{2019b}]%
        {WSS19:DC}
\bibfield{author}{\bibinfo{person}{Kyrill Winkler}, \bibinfo{person}{Manfred
  Schwarz}, {and} \bibinfo{person}{Ulrich Schmid}.}
  \bibinfo{year}{2019}\natexlab{b}.
\newblock \showarticletitle{Consensus in directed dynamic networks with
  short-lived stability}.
\newblock \bibinfo{journal}{\emph{Distrib.\ Comput.}} \bibinfo{volume}{32},
  \bibinfo{number}{5} (\bibinfo{year}{2019}), \bibinfo{pages}{443--458}.
\newblock
\urldef\tempurl%
\url{https://doi.org/10.1007/s00446-019-00348-0}
\showURL{%
\tempurl}


\end{thebibliography}
\bibliographystyle{ACM-Reference-Format}

\appendix

\section{Process-Time Graphs} 
\label{sec:model}

In the main body of our paper,
we have formalized our topological results in terms of admissible
executions
in the generic system model introduced in \cref{sec:general:model}. In this
section, we will show that they also hold a topological space consisting
of other objects, namely, \emph{process-time graphs~\cite{BM14:JACM}}. In a nutshell, a process-time
graph describes the process scheduling and all communication occurring
in a run, along with the
set of initial values.

Actually, since we consider deterministic algorithms only, a process-time
graph corresponds to a \emph{unique} execution (and vice
versa). This equivalence, which actually results from a
\emph{transition function} that is continuous in all our topologies
(see \cref{lem:tau:is:cont}), will eventually allow us to use our topological
reasoning in either space alike.

In order to define process-time graphs as generic as possible, we will resort
to an intermediate \emph{operational system model} that is essentially
equivalent to the very flexible general system model from Moses and
Rajsbaum~\cite{MR02}. 
Crucially, it will also instantiate the weak
clock functions $\chi_p(C^t)$ stipulated in our generic model in
\cref{sec:general:model}, which must satisfy $\chi_p(C^t)\leq t$
in every admissible execution $(C^t)_{t\geq 0} \in \Sigma$.
Since $t$ represents some global notion of time here (called \emph{global
real time} in the sequel), ensuring this property is sometimes not
trivial. More concretely,
whereas $t$ is inherently known at every process in the case of
lock-step synchronous
systems like dynamic networks under message adversaries~\cite{WSS19:DC},
for example, this is not the case for purely asynchronous systems~\cite{FLP85}.

\subsection{Basic operational system model}
\label{sec:basicoperationalmodel}

Following Moses and Rajsbaum~\cite{MR02}, we consider message passing or shared memory distributed systems made up of a set $\Pi$ of~$n\geq 2$ 
processes. We stipulate a global discrete clock with values taken from
$\N_0=\N \cup \{0\}$, which represents global real time in multiples of some arbitrary unit
time. Depending on the particular distributed computing model, this global clock may 
or may not be accessible to the processes.

Processes are modeled as communicating state machines that encode
a deterministic distributed algorithm (protocol) $\P$. At every real time
time $t \in \N_0$, process $p$ is in some \emph{local state} 
$L_p^t \in \L_p \cup \{\bot_p\}$, where $\bot_p\not\in \L_p$ is a special state representing
that process $p$ has failed.\footnote{This failed state $\bot_p$ is the only essential
difference to the model of Moses and Rajsbaum~\cite{MR02}, where faults are implicitly caused by
a deviation from the protocol. 
This assumption makes sense for constructing ``permutation layers'',
for example, where it is not the environment that crashes a process at will,
but rather the layer construction, which implies that some process takes
only finitely many steps. Such a process just remains in the local state reached 
after its last computing step. In our setting, however, the fault state of all 
processes is solely controlled by the omniscient environment. Hence, we can safely
use a failed state $\bot_p$ to gain simplicity
without losing expressive power.} 
Local state transitions of $p$ are caused by local \emph{actions} taken from the
set $\ACT_p$, which may be internal bookkeeping operations and/or the
initiation of shared memory operations resp.\ of sending messages; their
exact semantics may vary from model to model. 
Note that a single action may consist of finitely
many non-zero time operations, which are initiated simultaneously but may complete at different
times.
The deterministic protocol $\P_p: \L_p \to \ACT_p$, representing $p$'s part in
$\P$, is a function that specifies the local action $p$ is ready to perform when 
in state $L_p \in \L_p$. 
We do not restrict the actions $p$ can perform when in state $\bot_p$.

In addition, there is an additional non-deterministic state machine called
the \emph{environment} $\epsilon$, which represents the adversary that is responsible 
for actions outside the sphere of control of the processes' protocols. 
It 
controls things like the completion of shared memory operations initiated earlier resp.\ 
the delivery of previously sent messages, the occurrence of process and 
communication failures, and (optionally) the occurrence of \emph{external environment events}
that can be used for modeling oracle inputs like failure detectors~\cite{CT96}. 
Let $\act_\epsilon$ be the set of all possible combinations of such
\emph{environment actions} (also called \emph{events} for conciseness later on). 
We assume that the environment keeps track of pending shared-memory 
operations resp.\ sent messages in its \emph{environment state} $L_\epsilon \in \L_\epsilon$. 
The environment is also in charge of process \emph{scheduling}, i.e., 
determines when a process performs a state transition, which will be
referred to as \emph{taking a step}. Formally, we assume that
the set $\ACT_\epsilon$ of all possible environment actions consists of all pairs 
$(\Sched,e)$, made up of the set of processes $\Sched \subseteq \Pi$ that 
take a step and some $e\in\act_\epsilon$ (which may both be empty as well). 
The non-deterministic \emph{environment protocol} $\P_\epsilon \subseteq \G \times
(\ACT_\epsilon \times \L_\epsilon)$ is an arbitrary relation that, given the
current global state $G\in\G$ (defined below, which also contains the current
environment state $L_\epsilon \in \L_\epsilon)$, chooses the next environment action 
$E=(\Sched,e)\in\ACT_\epsilon$ and the successor environment state 
$L_\epsilon' \in \L_\epsilon$. Note carefully that we assume that only $E$ is actually
chosen non-deterministically by $\P_\epsilon$, whereas $L_\epsilon'$ is
determined by a transition function $\tau_\epsilon: \G \times \ACT_\epsilon \to 
\L_\epsilon$ according to $L_\epsilon'=\tau_\epsilon(G,E)$.

Finally, a \emph{global state} of our system (simply called state) is an element of 
$\G=\L_\epsilon\times \L_1 \times \cdots \times \L_n$. Given a global state $G \in \G$,
$G_i$ denotes the local state of process $i$ in $G$, and $G_\epsilon$ denotes the state
of the environment in $G$. Recall that it is $G_\epsilon$ that keeps track of in-transit
(i.e., just sent) 
messages, pending shared-memory operations 
etc.\footnote{A different, but equivalent, conceptual
model would be to assume that the state of a processor consists of a visible state
and, in the case of message passing, message buffers that hold in-transit messages.} 
We also write $G=(G_\epsilon,C)$, where the vector of the
local states $C=(C_1,\dots,C_n)=(G_1,\dots,G_n)$ of all the processes is called \emph{configuration}.
Given $C$, the component $C_i$ denotes the local state of process $i$ in $C$, and the
set of all possible configurations is denoted as $\C$. Note carefully that there may
be global configurations $G\neq G'$ where the corresponding configurations 
satisfy $C = C'$, e.g., in the case of different in-transit messages.

A \emph{joint action} is
a pair $(E,A)$, where $E=(\Sched,e)\in\ACT_\epsilon$, and $A$ is a
vector
with index set $\Sched$ such that $A_p \in \ACT_p$ for
$p \in \Sched$. When the joint action $E$ is applied to global state $G$ 
where process $p$ is in local state $G_p$, then $A_p=\P_p(G_p)$ is the
action prescribed by $p$'s protocol. Note that some environment actions,
like message receptions at process $p$ require $p \in \Sched$, i.e.,
``wake-up'' the process.  
For example, a joint action $(E,A)$ that causes $p$ to send a message $m$ to
$q$ and process $r$ to receive a message $m'$ sent to it by process $s$ 
earlier, typically works as follows: (i) $p$ is caused to take a step, where
its protocol $\P_p$ initiates the sending of $m$; (ii) the environment adds 
$m$ to the send buffer of the communication channel from $p$ to $q$ (maintained in the environment state $L_\epsilon$); (iii) the environment moves $m'$ from the send 
buffer of the communication channel from $s$ to $r$ (maintained in the environment 
state $L_\epsilon$) to the receive buffer (maintained in the local state of $r$), 
and (iv) causes $r$ to take a step. 
It follows that the local state $L_r$ of process $r$ reflects the content of message $m'$ 
immediately after the step scheduled along with the message reception.

With $\ACT$ denoting the set of all possible joint actions,
the \emph{transition function} $\tau: \G \times \ACT \to \G$ describes the
evolution of the global state $G$ after application of the joint action $(E,A)$,
which results in the successor state $G'= \tau(G,(E,A))$.
A \emph{run} of $\P$ is an infinite sequence of global states $G^0,G^1,G^2,\dots$ generated by an infinite sequence of joint actions. 
In order to guarantee a stable global state at integer times,
we assume for simplicity that the joint actions occur atomically and instantaneously at times $0.5,1.5,2.5,\dots$,
i.e., that $G^{t+1}=\tau(G^t,(E^{t.5},A^{t.5}))$. $G^0$ is the \emph{initial state}
of the run, taken from the set of possible initial states $\G^0$. Finally, $\Psi$
denotes the subset of all \emph{admissible runs} of our system. $\Psi$ is typically
used for enforcing liveness conditions like ``every message sent to a correct process is 
eventually delivered'' or ``every correct process takes infinitely many steps''.

Unlike Moses and Rajsbaum~\cite{MR02}, we handle process failures explicitly in the state of the
processes, i.e., via the transition function: If some joint action 
$(E^{t.5},A^{t.5})$ contains $E^{t.5}=(\Sched,e)$, 
where $e$ requests some process $p$ to fail, this will force 
$G_p^{t+1}=\bot_p$ in the successor state $G^{t+1}= \tau(G^t,(E^{t.5},A^{t.5}))$, 
irrespective of any other
operations in $e$ (like the delivery of a message) that would otherwise affect $p$. 
All process failures are persistent, that is, we require that
all subsequent environment actions $E^{t'.5}$ for $t'\geq t$ also request $p$ to
fail. As a convention, we consider every $E^{t'.5}$ where $p$ fails as $p$
taking a step as well. Depending on the type of process failure, 
failing may cause $p$ to stop its protocol-compliant internal 
computations, to drop all incoming messages, and/or to stop sending further 
messages. In the case of crash failures,
for example, the process may send a subset of the outgoing messages demanded
by $\P_p$ in the very first failing step and does not perform any protocol-compliant
actions in future steps. A send omission-faulty process does the same, except that
it may send protocol-compliant messages to some processes also in future steps. 
A receive omission-faulty process may omit to process some of its received messages 
in every step where it fails, but sends protocol-compliant messages to every
receiver. A general omission-faulty process combines the possible behaviors of
send and receive omissions. Note that message loss can also be modeled in a different
way in our setting: Rather than attributing an omission failure to the sender or
receiver process, it can also be considered a communication failures caused by
the environment. The involved sender process $p$ resp.\ receiver process $q$ continue to 
act according to its protocol in this case, i.e., would not enter the fault state 
$\bot_p$ resp.\ $\bot_q$ here.

Since we only consider deterministic protocols, a run $G^0,G^1,G^2,\dots$
is uniquely determined by the initial configuration $C^0$ and the sequence 
of tuples $(L_\epsilon^0,E^{0.5}), (L_\epsilon^1,E^{1.5}), \dots$ consisting 
of tuples $(L_\epsilon^t,E^{t.5})$ of environment state and environment actions  for 
$t\geq 0$. Let $\Gomega$ resp.\ $\Comega$ be the set of all infinite \emph{runs} resp.\ \emph{executions} (configuration sequences), with $\Psi \subseteq \Gomega$ resp.\ $\Sigma\subseteq \Comega$ 
denoting the set of \emph{admissible} runs resp.\ executions
that result from admissible environment action sequences
$E^{0.5},E^{1.5},\dots$; after all, they may be required to satisfy
liveness constraints like fairness that cannot be expressed via the
transition function.

Our assumptions on the
environment protocol, namely, $L_\epsilon^{t+1}=\tau_\epsilon(G^t,E^{t.5})$, 
actually imply that a run $G^0,G^1,G^2,\dots$, and thus 
also the corresponding execution $C^0,C^1,C^2,\dots$, is already
uniquely determined by the initial state $G^0=(L_\epsilon^0,C^0)$ and the sequence
of chosen environment actions $E^{0.5},E^{1.5},\dots$.
Since $L_\epsilon^0$ is fixed and the environment actions abstract away almost all of the
internal workings of the protocols and their complex internal states, it should be
possible to uniquely describe the evolution of a run/execution just by means of the sequence
$E^{0.5},E^{1.5},\dots$. In the following, we will show that this is indeed the case.

\subsection{Implementing global time satisfying the weak clock property}\label{sec:cc}

Our topological framework crucially relies on the ability
to distinguish/not distinguish two local states 
$\alpha_p^t$ and $\beta_p^t$ in two executions $\alpha$ and 
$\beta$ at global real time $t$. 
Clearly, this is easy for an omniscent observer
who knows the corresponding global states and can thus
verify that $\alpha_p^t$ and $\beta_p^t$ arise from the same global
time $t$. Processes cannot do that in asynchronous systems, however, since
$t$ is not available to the processes and hence cannot be included
in $\alpha_p^t$ and $\beta_p^t$. Consequently, two \emph{different} sequences of environment actions (called \emph{events} in the sequel for conciseness)
$E_\alpha^{0.5},E_\alpha^{1.5},\dots,E_\alpha^{(t-1).5}$ and
$E_\beta^{0.5},E_\beta^{1.5},\dots,E_\beta^{(t'-1).5}$, applied to
the same initial state, may produce the \emph{same} state $\alpha_p^t=\beta_p^{t'}$. This
happens when they are causal shuffles of each other, i.e., 
reorderings of the steps of the processes that are in accordance 
with the happens-before relation~\cite{Lam78}. Hence, the (in)distinguishability of
configurations does not necessarily match the (in)distinguishability
of the corresponding event sequences. 

Whereas our generic system model does not actually
require processes to have a common notion 
of time, it does require that the weak clock functions $\chi_p$ do
not progress faster than global real time. We will accomplish this
in our operational system model by defining some alternative notion
of global time that \emph{is} accessible to the processes. Doing this
will also rule out the problem spotted above, i.e., ensure that
runs (event sequences) uniquely determine executions (configuration
sequences).

There are many conceivable ways for defining global time, including 
the following possibilities:

(i) In the case of lock-step synchronous distributed systems, like
dynamic networks under message adversaries~\cite{NSW19:PODC,WSM19:OPODIS,
WSN21:FCT}, nothing needs to be done here since all processes inherently
know global real time $t$.

(ii) In the case of asynchronous systems with a
majority of correct processes, the arguably
most popular approach for message-passing systems
(see e.g. \cite{MR01,ADGFT04,HMSZ08:TDSC}) is the simulation
of \emph{asynchronous communication-closed rounds}: Processes organize rounds
$r=1,2,\dots$ by locally waiting until $n-f$ messages sent in the current
round $r$ have been received. These $n-f$ messages are then processed, which
defines both the local state at the beginning of the next round $r+1$
and the message sent to everybody in this next round. Late messages
are discarded, and early messages are buffered locally (in the state of
the environment) until the appropriate round is reached. The very same
approach can also be used in shared-memory systems with 
immediate snapshots \cite{AADG93:JACM}, where a process can safely
wait until it sees $n-f$ entries in a snapshot.
Just using the round numbers
as global time, i.e., choosing $t=r$, is all that is needed for defining
global time in such a model.

(iii) In models without communication-closed rounds~\cite{FLP85,RS10:TCS},
a suitable notion of global time can be derived from
other\footnote{We note that both synchronous and asynchronous
communication-closed rounds, as well as the executions $\Comega$
defined in our generic system model in \cref{sec:general:model},
are of course also sequences of consistent cuts.} definitions of
\emph{consistent cuts}~\cite{Mat89}. We will show how this can be done
in our operational system model based on Mattern's \emph{vector clocks}.
Our construction will exploit the fact that a local state 
transition of a process happens only when it takes a step in our model:
In between
the $\ell$\textsuperscript{th} and $(\ell+1)$\textsuperscript{th} step of any fixed process $p$, 
which happens at time $(t_p(\ell)-1).5$ and $(t_p(\ell+1)-1).5$, respectively,
only environment actions (external environment events, message deliveries, shared memory completions), 
if any, can happen, which change the state of the environment but not the local state of $p$.

We will start out from the sequence of arbitrary \emph{cuts}~\cite{Mat89}
$IC^0,IC^1,IC^2,\dots$ (indexed by an integer \emph{index} $k\geq 0$) occurring 
in a given run $G^0,G^1,G^2,\dots$ (which itself is indexed by the global real time $t$), 
where the \emph{frontier} $IF^k$ of $IC^k$ is 
formed by the local states of the processes after they have 
taken their $k$\textsuperscript{th} step, i.e., $IF^0=IC^0=C^0$ and $IF^k=(G_1^{t_1(k)},\dots,G_n^{t_n(k)})$
for $k\geq 1$, with $(t_p(k)-1).5$ being the time when process $p$
takes its $k$\textsuperscript{th} step. Note that the latter is applied to $p$'s state $IF_p^{k-1}$ in
the frontier $IF^{k-1}$ of $IC^{k-1}$ and processes all the external environment events and
all the messages received/shared
memory operations completed since then. Recall the convention that every 
environment action where process $q$ fails is also considered as $q$ taking a step. 

Clearly, except in lock-step synchronous
systems, $t_p(k)\neq t_q(k)$, so $IC^0,IC^1,IC^2,\dots$ 
can be viewed as the result of applying a trivial ``synchronic layering'' in terms 
of Moses and Rajsbaum~\cite{MR02}. Unfortunately, though, any $IC^k$ may be an
\emph{inconsistent} cut, as messages sent by a fast process $p$ in its $(k+1)$\textsuperscript{th} step
may have been received by a slow process $q$ by its $k$\textsuperscript{th} step. $IC^k$ would violate 
causality in this case, i.e., would not be
left-closed w.r.t.\ Lamport's happens-before relation~\cite{Lam78}.

Recall that we restricted our attention to consensus algorithms using 
full-information protocols, where
every message sent contains the entire state transition history
of the sender. As a consequence, we do not significantly lose applicability
of our results by further restricting the protocol and the supported distributed 
computing models as follows:

\begin{enumerate}
\item[(i)] In a single state transition of $\P_p$, process $p$, can
\begin{itemize}
\item actually receive all messages delivered to it since its last step,
\item initiate the sending of at most one message to every process, resp., 
\item initiate at most one single-writer multiple-reader shared memory operation 
in the shared memory owned by some other process (but no restriction on 
operations in its own shared memory portion).
\end{itemize}
\item[(ii)] In addition to (optional) external environment events, the environment protocol only provides 
\begin{itemize}
\item $\fail(q) \in \act_\epsilon$, which tells process $q$ to fail,
\item $\recv(q,p,t_k)\in \act_\epsilon$, which identifies the message 
$m$ to be delivered to process $q$ (for reception in its next step) by the pair $(p,t_k)$, where $p$ 
is the sending process and $t_k.5$ is the time when the sending of $m$ has been initiated, resp.,
\item $\done(q,p,t_\ell,t_k) \in \act_\epsilon$, which identifies the
completed shared memory operation (to be processed in its next step), 
in the shared memory owned by $p$,
as the one initiated by process $q\neq p$ in its step at 
time $t_\ell.5$; in a read-type operation, it will return to $q$ the shared memory 
content based on $p$'s state at time $t_k$, with $t_\ell \leq t_k$.
\end{itemize}
\end{enumerate}

In such a system, given any cut $IC^k$, it is possible to determine
the unique largest \emph{consistent} cut $CC^k \subseteq IC^k$ \cite{Mat89}.
By construction, $CC^0=IC^0$, and the frontier $CF^k$ of $CC^k$, $k\geq 1$, 
consists of the local states of all processes $q\in \Pi$ reached by having taken some 
$\ell(q)$\textsuperscript{th} step, $0\leq \ell(q)\leq k$, with at least one process $p$ having taken its
$k$\textsuperscript{th} step, i.e., $\ell(p)=k$ and thus $CF^k_p=IF^k_p$, and $CF^k_q = IF_q^{\ell(q)}$
with $0\leq \ell(q)\leq k$ for all processes $q$. Note carefully that $\ell(q)<k$
happens when, in $IC^k$, process $q$ receives some message/data initiated 
at some step $>k$ at or before its own $k$\textsuperscript{th} step  but after its 
$\ell(q)$\textsuperscript{th} step. 

Whereas the environment protocol could of course determine all 
the consistent cuts $CC^0,CC^1,CC^2,\dots$ based on the corresponding sequence
of global configurations, this is typically not 
the case for the processes (unless in the special case of a 
synchronous system). However, in distributed systems adhering to the above constraints,
processes can obtain this knowledge (that is to say, their local share of a consistent cut)
via \emph{vector clocks}~\cite{Mat89}.
More specifically, it is possible to implement a vector clock
$k_p=(k_p^1,\dots,k_p^n)$ at process $p$, where $k_p^p$ counts
the number of steps taken by $p$ itself so far, and $k_p^q$, $q\neq p$,
gives the number of steps that $p$ knows that $q$ has taken so far.
Vector clocks are maintained as follows: Initially, $k_p=(0,\dots,0)$, and every message
sent resp.\ every shared memory operation data written by $p$ gets
$k_p$ as piggybacked information (after advancing $k_p^p$). 
At every local state transition in $p$'s protocol $P_p$, 
$k_p^p$ is advanced by 1. Moreover, when a previously received 
message/previously read data value (containing the originating process
$q$'s vector clock value $\hat{k}_q$) is to be processed in the step, 
$k_p$ is adjusted to the maximum of its previous value and $\hat{k}_q$ 
component-wise, i.e., $k_p^q=\max\{k_p^q,\hat{k}_q^q\}$ for $q\neq p$.
Obviously, all this can be implemented transparently atop of any protocol $\P$ 
running in the system.

Now, given the sequence of global states $AC^0,AC^1,AC^2,\dots$ of the processes 
running the so augmented protocol in some run $G^0,G^1,G^2,\dots$, there is
a well-known algorithm for computing the maximal consistent cut $ACC^k$ 
for the non-consistent cut $AIC^k$ formed by the frontier $AIF^k$ of
the local states of the processes after every process has taken 
its $k$\textsuperscript{th} step: Starting from $\ell:=k$, process $p$ searches for
the sought $\ell(p)$ by checking the vector clock value $k_p(\ell)$ 
of the state after its own $\ell$\textsuperscript{th} step. It stops searching and sets $\ell(p):=\ell$
iff $k_p(\ell)$ is less or equal to $(k,\dots,k)$ component-wise.
The state $AIF_p^{\ell(p)}$ is then process $p$'s contribution in the 
frontier $ACF^k$ of the consistent cut $ACC^k$. Clearly, from $ACC^0,ACC^1,ACC^2,\dots$,
the sought sequence of the consistent cuts $CC^0,CC^1,CC^2,\dots$ can be obtained trivially by discarding all
vector clock information. Therefore, even the processes can compute their
share, i.e., their local state, in $CC^k$ for every $k$. 

By construction, the sequence of consistent cuts $CC^0, CC^1, CC^2,\dots$,
and hence the sequence of its frontiers $CF^0,CF^1,CF^2,\dots$,
completely describe the evolution of the local states of the processes 
in a run $G^0,G^1,G^2,\dots$. In our operational model, we will hence 
just use the indices $k$ of $CC^k$ as global time for specifying 
executions: Starting from the initial state $CC^0$, 
we consider $CC^k$ as the result of applying \emph{round} $k\geq 1$ 
to $CC^{k-1}$ (as we did in the case of lock-step rounds).

\subsection{Defining process-time graphs}\label{sec:pt}

No matter how consistent cuts, i.e., global time, is implemented,
from now on, we just overload the notation used so far and denote by
$C^k$ the frontier $CF^k$ in the consistent cut at global time $k$.
So given an infinite 
execution $\alpha$, we again denote by $\alpha^t$ the
$t$\textsuperscript{th} configuration (= the consistent cut with index $t$) 
in $\alpha$. 

Clearly, by construction, every $C^k$ is \emph{uniquely} determined by $C^0$ and all the events that
cause the steps leading to $C^k$. In particular, we can define 
a vector of events $E^k$, where $E_p^k$ is a set containing all the events that must
be applied to $C_p^{k-1}$ in order to arrive at $C_p^k$. Note carefully
that  a process $p$ that does not make a step, i.e., is not scheduled
in $E^k$ and thus has the same non-$\bot_p$ state in $C^{k-1}$ and $C^k$, 
does not have any event 
$\recv(p,*,*)\in E_p^k$ (resp.\ $\done(p,*,*)\in E_p^k$) 
by construction, i.e., $E_p^k=\emptyset$. 
Otherwise, $E_p^k$ contains a ``make a step'' event, all (optional) external
environment events, 
and $\recv(p,*,*)$ for all messages that have been sent to $p$ in steps
within $C^{k-1}$ and are delivered to $p$ after its previous step but 
before or at its $k$\textsuperscript{th} step (resp.\ $\done(p,*,*,*)$ for all completed
shared memory operation initiated by $p$ in steps within $CC^{k-1}$ and 
completed after $p$'s previous step but before or at its $k$\textsuperscript{th} step).
Note that $E_p^1$ cannot contain any
$\recv(p,*,*)$, as no messages have been sent before (resp.\ no 
$\done(p,*,*,*)$, as no shared memory operations have been initiated before).

As a consequence of our construction, the mismatch problem spotted at the beginning of \sectionref{sec:cc} 
no longer exists, and we can reason about executions and the corresponding event sequences 
alike.  

\medskip

Rather than considering $C^0$ in conjunction with $E^1,\dots,E^k$, however, 
we will consider the corresponding
\emph{process-time graph $k$-prefix} $PTG^k$~\cite{BM14:JACM} instead, which we will now 
define. Since we are only interested in 
consensus algorithms here, we assume that every process has a dedicated
initial state for every possible initial value~$v$, taken from a finite input
domain~$\V$. For every assignment of
initial values $x\in \V^n$ to the~$n$ processes in the initial configuration $C^0$, 
we inductively construct the following sequence of process-time graph prefixes $PTG^t$:

\begin{definition}[Process-time graph prefixes]\label{def:PTGs}
For every $k\geq 0$, the \emph{process-time graph $k$-prefix} $PTG^k$ of a given run
is defined as follows:
\begin{itemize}
\item
The process-time graph $0$-prefix~$PTG^0$ contains the nodes $(p,0,I_p)$ for 
all processes $p\in\Pi$, with initial value $I_p\in\V$, and no edges.

\item
The process-time graph $1$-prefix~$PTG^1$ contains the nodes $(p,0,I_p)$ and $(p,1,f)$ 
for all processes $p\in\Pi$, where $f=\bot$ if $\fail(p)\in E^1$ (which models the
case of an initially dead process~\cite{FLP85}), and $f=*$ otherwise, where
$*$ is some encoding (e.g., some failure detector output) of the external environment 
events $\in E^1$.
It contains a (local) edge from $(p,0,I_p)$ to $(p,1,f)$ and no other edges.

\item
The process-time graph $k$-prefix~$PTG^k$, $k\geq 2$, contains $PTG^{k-1}$
and the nodes $(p,k,f)$ for all processes $p\in\Pi\setminus\{q \mid E_q^k=\emptyset\}$, 
where $f=\bot$ if $\fail(p)\in E^k$, and $f=*$ otherwise. It contains a (local) edge from
$(p,\ell,f_\ell)$ to $(p,k,f)$ (if the latter node is present at all, 
i.e., when $E_p^k\neq\emptyset$), where $\ell$ is maximal among all nodes $(p,*,*)$ in $PTG^{k-1}$.
For message passing systems, it also contains an edge from $(p,s,f_s)$, $1 \leq s < k$, to $(q,k,f)$ 
iff $\recv(q,p,s) \in E^{k}$. For shared memory systems, it contains an
edge from $(p,\ell,f_\ell)$, $1 \leq \ell < k$, to $(q,k,f)$  if and only if $\done(q,p,s,\ell) \in E^{k}$;
this reflects the fact that the returned data originate from $p$'s step $\ell$ and
not from the step $s$ where $q$ has initiated the shared memory operation.
\end{itemize}
The \emph{round-$\ell$ process-time graph} $PT^\ell$, for $0\leq \ell \leq k$,
which represents the contribution of round $\ell$ to $PTG^k$, is defined
as (i) $PT^0=PTG^0$ and the set of vertices $PT^{\ell}=PTG^\ell\setminus PTG^{\ell-1}$ 
along with all their incoming edges (which all originate in $PTG^{\ell-1}$).
\end{definition}

Figure~\ref{fig:ptlockstep} shows an example of a process-time graph prefix
occuring in a run with lock-step synchronous or asynchronous rounds.
The nodes are horizontally aligned according to global time, progressing
along the vertical axis.

\begin{figure}
\begin{tikzpicture}[>=latex']
\node[black!30!green, very thick, draw, circle,label=left:{$(1,0,1)$}] (C10) at (-3, 0)  {};
\node[black!30!green, very thick, draw, circle,label=left:{$(2,0,0)$}] (C20) at ( 0, 0)  {};
\node[black!30!green, very thick, draw, circle,label=left:{$(3,0,1)$}] (C30) at ( 3, 0)  {};

\node[black!30!green, very thick, draw, circle,label=left:{$(1,1,*)$}] (C11) at (-3, 2)  {};
\node[black!30!green, very thick, draw, circle,label=left:{$(2,1,*)$}] (C21) at ( 0, 2)  {};
\node[draw, circle,            label=left:{$(3,1,*)$}] (C31) at ( 3, 2)  {};

\node[black!30!green, very thick, draw, circle,label=left:{$(1,2,*)$}] (C12) at (-3, 4)  {};
\node[draw, circle,            label=left:{$(2,2,*)$}] (C22) at ( 0, 4)  {};
\node[draw, circle,            label=left:{$(3,2,*)$}] (C32) at ( 3, 4)  {};

\node[black!30!green, very thick, draw, circle,label=left:{$(1,3,*)$}] (C13) at (-3, 6)  {};
\node[draw, circle,            label=left:{$(2,3,*)$}] (C23) at ( 0, 6)  {};
\node[draw, circle,            label=left:{$(3,3,*)$}] (C33) at ( 3, 6)  {};

\draw[black!30!green, very thick,->] (C10) -- (C11);
\draw[black!30!green, very thick,->] (C20) -- (C21);
\draw[->] (C30) -- (C31);

\draw[black!30!green, very thick,->] (C11) -- (C12);
\draw[->] (C21) -- (C22);
\draw[->] (C31) -- (C32);
\draw[black!30!green, very thick,->] (C21) -- (C12);
\draw[->] (C31) -- (C22);
\draw[->] (C11) -- (C22);
\draw[->] (C11) -- (C32);

\draw[->] (C12) -- (C13);
\draw[->] (C22) -- (C23);
\draw[->] (C32) -- (C33);
\draw[->] (C22) -- (C13);
\draw[->] (C22) -- (C33);

\draw[thick,blue] plot [smooth] coordinates { (C10) (C20) (C30) } -- +(1,0) node[anchor=west] {$PTG^0$};
\draw[thick,blue] plot [smooth] coordinates { (C11) (C21) (C31) } -- +(1,0) node[anchor=west] {$PTG^1$};
\draw[thick,blue] plot [smooth] coordinates { (C12) (C22) (C32) } -- +(1,0) node[anchor=west] {$PTG^2$};
\draw[thick,blue] plot [smooth] coordinates { (C13) (C23) (C33) } -- +(1,0) node[anchor=west] {$PTG^3$};

\end{tikzpicture}
\caption{Example of a process-time graph prefix $PTG^3$ of a lock-step
  execution at time $t=3$, for $n=3$ processes and initial values $x=(1,0,1)$.
Process~$1$'s view $V_{1}(PT^2)$ is highlighted in bold green.}
\label{fig:ptlockstep}
\end{figure}

Figure~\ref{fig:pt} shows an example of a process-time graph prefix occuring
in a run with processes that do not execute in lock-step rounds and may crash.
Nodes are again horizontally aligned according to global time, progressing
along the vertical axis. The
frontier $C^k$ of the $k$\textsuperscript{th} consistent cut, reached at the end of round $k$, 
is made up of $C_p^k=\{(p,\ell_p(k),*) \in PTG^k \mid \mbox{$0\leq\ell_p(k)\leq k$ is maximal}\}$. That is, starting
from the (possibly inconsistent) cut made up of the nodes $(p,k,*)$ of all 
processes, one has to go down for process $p$ until the first node is reached where 
no edge originating in a node with time $>k$ has been received.

\begin{figure}
\begin{tikzpicture}[>=latex']
\node[black!30!green, very thick, draw, circle,label=left:{$(1,0,1)$}] (C1I) at (-3, 0)  {};
\node[black!30!green, very thick, draw, circle,label=left:{$(2,0,0)$}] (C2I) at ( 0, 0)  {};
\node[black!30!green, very thick, draw, circle,label=left:{$(3,0,1)$}] (C3I) at ( 3, 0)  {};

\node[black!30!green, very thick, draw, circle,label=left:{$(1,1,*)$}] (C10) at (-3, 2)  {};
\node[black!30!green, very thick, draw, circle,label=left:{$(2,1,*)$}] (C20) at ( 0, 2)  {};
\node[black!30!green, very thick, draw, fill, circle,label=left:{$(3,1,\bot)$}] (C30) at ( 3, 2)  {};

\node[black!30!green, very thick, draw, circle,label=left:{$(1,2,*)$}] (C11) at (-3, 4)  {};
\node[draw, circle, fill,           label=left:{$(3,2,\bot)$}] (C31) at ( 3, 4)  {};

\node[black!30!green, very thick, draw, circle,label=left:{$(1,3,*)$}] (C12) at (-3, 6)  {};
\node[draw, circle, fill,           label=left:{$(3,3,\bot)$}] (C32) at ( 3, 6)  {};

\node[draw, circle,label=left:{$(1,4,*)$}] (C13) at (-3, 8)  {};
\node[draw, circle,            label=left:{$(2,2,*)$}] (C23) at ( 0, 8)  {};
\node[draw, circle, fill,           label=left:{$(3,4,\bot)$}] (C33) at ( 3, 8)  {};

\node[draw, circle,label=left:{$(1,5,*)$}] (C14) at (-3, 10)  {};
\node[draw, circle,            label=left:{$(2,3,*)$}] (C24) at ( 0, 10)  {};
\node[draw, circle, fill,           label=left:{$(3,5,\bot)$}] (C34) at ( 3, 10)  {};

\draw[black!30!green, very thick,->] (C1I) -- (C10);
\draw[black!30!green, very thick,->] (C2I) -- (C20);
\draw[black!30!green, very thick,->] (C3I) -- (C30);

\draw[thick,blue] plot [smooth] coordinates { (C1I) (C2I) (C3I) } -- +(1,0) node[anchor=west] {$PTG^0$};

\draw[black!30!green, very thick,->] (C10) -- (C11);
\draw[->] (C20) -- (C23);
\draw[->] (C30) -- (C31);
\draw[black!30!green, very thick,->] (C20) -- (C11);
\draw[->] (C11) -- (C33);
\draw[->] (C10) -- (C32);
\draw[->] (C20) -- (C31);
\draw[black!30!green, very thick,->] (C30) -- (C12);

\draw[thick,blue] plot [smooth] coordinates { (C10) (C20) (C30) } -- +(1,0) node[anchor=west] {$PTG^1$};

\draw[black!30!green, very thick,->] (C11) -- (C12);
\draw[->] (C31) -- (C32);
\draw[->] (C12) -- (C23);
\draw[->] (C12) -- (C33);

\draw[thick,blue] plot [smooth] coordinates { (C11) (C20) (C31) } -- +(1,0) node[anchor=west] {$PTG^2$};

\draw[->] (C12) -- (C13);
\draw[->] (C32) -- (C33);

\draw[thick,blue] plot [smooth] coordinates { (C12) (C23) (C32) } -- +(1,0) node[anchor=west] {$PTG^3$};

\draw[->] (C13) -- (C14);
\draw[->] (C23) -- (C24);
\draw[->] (C33) -- (C34);
\draw[->] (C13) -- (C24);

\end{tikzpicture}
\caption{Example of a process-time graph prefix in a non-lockstep execution of a system of $n=3$ processes 
with initial values $x=(2,0,1)$, where process $p_3$ crashes in its step at time~$1$, in
round~1. 
The vertical axis is the global time axis, and nodes
at the same horizontal level occur at the same global time. 
The length of the edges represent end-to-end delay of
a message resp.\ the access latency of a shared memory operation. Process~$1$'s local view $V_{1}(PTG^3)$ in
$PTG^3$ is highlighted in bold green.}
\label{fig:pt}
\end{figure}

\medskip

Let $\PTG^t$ be the set of all possible process-time graph $t$-prefixes, and $\PTG^\omega$ be the set of all posible infinite process-time graphs, for all possible runs of our system.
Note carefully that $\PTG^t$, as well every set $\P^\ell$ of round-$\ell$ process-time graphs 
for finite $\ell$, is necessarily \emph{finite} (provided the encoding
($*$) for external environment events has a finite domain, which we assume). Clearly, 
$\PTG^t$ resp.\ $\PTG^\omega$ can be expressed as a finite resp.\ infinite 
sequence $(P^0,\dots,P^t) \in \P^0 \times \P^1 \times \dots \times \P^t = \PTG^t$ 
resp.\ $(P^0,P^1,\dots) \in \P^0 \times \P^1 \times \dots = \PTomega$ of 
round-$\ell$ process time graphs.\footnote{Note that we slightly 
abuse the notation $\PTomega$ here, which normally represents 
$\PT \times\PT\times \dots$.} 

We will denote by $PS\subseteq \PTomega$
the set of all admissible process-time graphs in the given model,
and by $\Sigma \subseteq \C^\omega$ the corresponding set of admissible
executions.
Note carefully that process-time graphs are \emph{independent} of
the (decision function of the) consensus algorithm, albeit they do depend on the initial values.

Due to the one-to-one of process-time graphs and executions established before, the topological machinery developed in \cref{sec:structure:executions}--\cref{sec:consensus} 
for $\Sigma\subseteq \Comega$ can also be applied to $PS \subseteq \PTomega$. Since, in sharp contrast 
to the set of configurations~$\C$, the set of process-time graphs $\PT^t$ is 
finite for any time~$t$ and hence compact in the discrete topology,
Tychonoff's theorem\footnote{Tychonoff's theorem states that any product of
compact spaces is compact (with respect to the product topology).} implies 
compactness of the $p$-view topology on $\PTomega$. 

Whereas this is not necessarily the case for~$\Comega$, we can prove compactness
of the image of $\PTomega$ under an appropriately defined operational
transition function:
Given the original transition function $\tau_\epsilon: \G \times \ACT_\epsilon \to 
\L_\epsilon$, it is possible to define a PTG transition function
$\otau: \PTomega \to \Comega$ that just provides the (unique) execution for a
given process-time graph. The following \cref{lem:tau:is:cont}
shows that $\otau$ is continuous in any of our topologies.

\begin{lemmarep}[Continuity of $\otau$]\label{lem:tau:is:cont}
For every $p\in \Pi$, the PTG transition function $\otau : \PTomega \to \Comega$ 
is continuous when both $\PTomega$ and $\Comega$ are endowed with any of
$d_p$, $p\in \Pi$, $\dunif$, $\dnonunif$.
\end{lemmarep}
\begin{proof}
Let $U\subseteq \Comega$ be open with respect to~$d_p$, and let 
$a \in \otau^{-1}[U]$.
Since~$U$ is open and $\otau(a)\in U$, there exists some $\varepsilon > 0$ such that
$B_\varepsilon\big( \otau(a) \big) \subseteq U$.
Let $t\in\N$ such that $2^{-t} \leq \varepsilon$.
We will show that $B_{2^{-t}}(a) \subseteq \otau^{-1}[U]$.
For this, it suffices to show that $\otau\big[B_{2^{-t}}(a)\big] \subseteq U$.
By the equivalence of process-time graph prefixes and the corresponding consistent cuts, which
is ensured by construction, it follows for the views of process $p$ that $V_p(a^t) = V_p(b^t)$ implies
$V_p(\otau(a)^t) = V_p(\otau(b)^t)$.
Using this in \cref{eq:Pviewpseudometric} implies
\[
\otau\big[B_{2^{-t}}(a)\big] 
\subseteq
B_{2^{-t}}(\otau(a))
\subseteq
B_{\varepsilon}(\otau(a))
\subseteq
U
\enspace,
\]
which proves that $\otau^{-1}[U]$ is open as needed. 

The proof for $\dunif$ resp.\ $\dnonunif$ is analogous, except that
\cref{eq:Pviewpseudometric} must be replaced by \cref{eq:dunif} resp.\
\cref{eq:dnonunif}.
\end{proof}

Since the image of a compact space under a continuous function is compact, 
it hence follows that the set $\otau[\PTomega] \subseteq \Comega$ of admissible
executions is a compact subspace of $\Comega$. 
The common structure of $\PTomega$ and its image under the PTG transition
function~$\otau$, implied by the continuity
of $\tau$, hence allows us to reason in either of these spaces.
In particular, with \cref{def:vvalent}, the analog of
\cref{thm:char:unif} and \cref{thm:char:nonunif} read as follows:

\begin{definition}[$v$-valent process-time graph]\label{def:vvalent}
We call a process-time graph $z_v$, for $v\in\V$, $v$-valent, if it
starts from an initial configuration where all processes $p\in\Pi$ have the same 
initial value $I_p=v$.
\end{definition}

\begin{theorem}[Characterization of uniform consensus]\label{thm:char:unif:PS}
Uniform consensus is solvable if and only if there exists
a partition of the set~$PS$ of admissible process-time graphs
into sets $PS_v$, $v\in\V$, such that
the following holds:
\begin{enumerate}
\item Every $PS_v$ is an open set in~$PS$ with respect to the
uniform topology induced by~$\dunif$.
\item If process-time graph $a \in PS$ is $v$-valent, then $a \in PS_v$.
\end{enumerate}
\end{theorem}

\begin{theorem}[Characterization of non-uniform consensus]\label{thm:char:nonunif:PS}
Non-uniform consensus is solvable if and only if there exists
a partition of the set~$PS$ of admissible process-time graphs
into sets $PS_v$, $v\in\V$, such that
the following holds:
\begin{enumerate}
\item Every $PS_v$ is an open  set in~$PS$ with respect to the
non-uniform topology induced by~$\dnonunif$.
\item If process-time graph $a \in PS$ is $v$-valent, then $a \in PS_v$.
\end{enumerate}
\end{theorem}

\end{document}